\documentclass{article}
\usepackage{natbib}
\usepackage{url}
\usepackage{graphicx} 
\usepackage{amsmath}
\usepackage{amsthm}
\usepackage{amsfonts}
\usepackage{amssymb}
\usepackage{geometry}
\usepackage{appendix}
\usepackage{float}
\usepackage{xcolor}
\usepackage{booktabs}
\usepackage{adjustbox}
\usepackage{threeparttable}
\usepackage{comment}
\usepackage{tikz, pgfplots}
\pgfplotsset{compat=1.16}

\allowdisplaybreaks

\newtheorem{assumption}{Assumption}
\newtheorem{theorem}{Theorem}
\newtheorem{lemma}{Lemma}

\begin{document}
\sloppy
\title{Estimating Export-productivity Cutoff Contours with Profit Data:\\ A Novel Threshold Estimation Approach.\thanks{
First ArXiv date: February 25, 2025. Egger acknowledges funding from the Swiss National Science Foundation under grant number IC0010-228258-10001498. We thank the editor in charge (Gianmarco Ottaviano) and two anonymous reviewers for numerous helpful comments on an earlier version of the manuscript. We also thank Ruobing Huang, Carlos Lamarche, Andres Santos, and Xiye Yang for very helpful comments.
}
}
\author{Peter H. Egger\thanks{Professor of Economics, ETH Zürich. Email: pegger@ethz.ch.} \ and \ Yulong Wang\thanks{Associate Professor of Economics, Lehigh University. Email: yuw925@lehigh.edu.}}
\date{\today}
\maketitle

\begin{abstract}
This paper develops a novel method to estimate firm-specific market-entry thresholds in international economics, allowing fixed costs to vary across firms alongside productivity. Our framework models market entry as an interaction between productivity and observable fixed-cost measures, extending traditional single-threshold models to ones with set-valued thresholds. Applying this approach to Chinese firm data, we estimate export-market entry thresholds as functions of domestic sales and surrogate variables for fixed costs. The results reveal substantial heterogeneity and threshold contours, challenging conventional single-threshold-point assumptions. These findings offer new insights into firm behavior and provide a foundation for further theoretical and empirical advancements in trade research.
\end{abstract}

\newpage

\section{Introduction}
Many economic problems involve thresholds at which the behavior of agents changes discretely. In international economics, such thresholds often arise in unobservable variables, with firm productivity being a leading example. When firms face fixed costs of market entry, discrete choices such as exporting, multinational production, or outsourcing generate kinks or discontinuities in observable outcomes such as profits or sales. Understanding the nature of these thresholds is central to the empirical analysis of firm selection and market participation.

A prominent framework in international trade assumes that firms differ only in productivity, while fixed costs of exporting are common across firms. In the seminal model of \citet{melitz2003impact}, firms decide whether to export based on productivity alone, and exporter selection is governed by a single scalar productivity cutoff. Firms with productivity above this cutoff serve both domestic and foreign markets, while those below it sell domestically only. This structure implies a sharp selection rule and generates kinked relationships between productivity and firm outcomes. Numerous extensions and empirical studies adopt this logic, using observable proxies such as prices, revenues, or sales to infer productivity-based export cutoffs \citep{bernard2003plants,chaney2008distorted,arkolakis2010market,adao2020aggregate}.

Although productivity is unobservable, its structural relationship to observable firm outcomes is often monotone and well understood, which facilitates empirical identification. Fixed costs, on the contrary, are latent and lack a clear theoretical mapping to observables. As a result, empirical work typically assumes fixed export costs to be homogeneous or absorbs their heterogeneity into unobserved error terms. Under this assumption, exporter selection can be summarized by a single cutoff in productivity or its observable counterpart.

However, both theory and evidence suggest that fixed export costs may vary systematically between firms. Firms differ in their access to distribution networks, financial conditions, ownership structure, and market reach, all of which can affect the costs of entering foreign markets. Several theoretical contributions allow for heterogeneous fixed costs and show that exporter selection may no longer be governed by a single cutoff \citep{arkolakis2010market,helpman2016trade}. Empirical studies also suggest that financial variables, ownership, and other characteristics of firms can proxy fixed costs and exhibit substantial variation between firms \citep{wippern1966financial,manova2008credit,manova2013credit}.

When fixed costs vary between firms alongside productivity, exporter selection cannot be characterized by a scalar productivity threshold. Instead, the cutoff becomes set-valued: for a given level of productivity, some firms export while others do not, depending on their fixed costs. Equivalently, the export productivity threshold varies with observable characteristics of the firm that shift fixed costs. In such settings, exporter selection is governed by a contour in the joint space of productivity and fixed-cost shifters rather than by a single point.

This paper develops an econometric framework to estimate such heterogeneous export thresholds. We consider settings in which firms’ total profits exhibit a kink with respect to a running variable such as domestic sales, which serves as a monotone proxy of productivity. Crucially, we allow the location of this kink to depend on observable firm characteristics that proxy for fixed costs. The resulting threshold is modeled as an unknown nonparametric function of these characteristics.

Methodologically, our approach extends regression-kink models, which traditionally assume a constant cutoff common to all observations \citep{hansen2000sample,hansen2017regression}. In contrast, we allow the location of the kink to vary systematically between firms. The econometric object of interest is therefore not a scalar threshold parameter, but a \emph{threshold contour} that maps fixed-cost proxies into productivity or domestic-sales cutoffs.

We propose a kernel-based estimation procedure that recovers this contour locally. At each value of the fixed-cost proxy, the threshold is estimated using a locally weighted regression-kink estimator combined with a grid search over candidate kink locations. Repeating this procedure across the support of the fixed-cost proxy yields an estimate of the threshold function. We establish the large sample properties of the estimator and show that the bias structure differs from that of standard regression-kink models due to the functional nature of the threshold. The method can be extended to accommodate endogenous running variables using a control-function approach.

The proposed framework delivers economically interpretable objects that are directly relevant for empirical trade analysis. Rather than imposing a homogeneous export cutoff, the estimator quantifies how export thresholds shift with observable firm characteristics. This allows researchers to study heterogeneous selection patterns and to assess whether exporter premia and productivity differences are driven by productivity alone or by interactions between productivity and fixed costs.

We apply the method to a large cross section of Chinese manufacturing firms using matched production and customs data. The application exploits detailed profit and sales information, along with firm-level proxies for fixed costs such as financial expenses and fixed assets. The results reveal substantial heterogeneity in export cutoffs between firms and strongly reject the hypothesis of a scalar-valued export threshold. Consistent with the presence of heterogeneous fixed costs, the estimated threshold contour implies that firms with higher fixed-cost proxies require higher productivity to profitably export.

Accounting for heterogeneous thresholds alters implied selection patterns. In particular, we do not find clear evidence of positive exporter selection once fixed-cost heterogeneity is taken into account. This finding contrasts with the conventional wisdom derived from models with homogeneous fixed costs and highlights the importance of allowing flexible selection mechanisms in empirical work.

Although the empirical application focuses on China, the econometric framework is broadly applicable to a wide range of economic settings in which discrete choices are governed by heterogeneous, unobserved thresholds. Examples include multinational production decisions, outsourcing versus in-house production, and multiproduct firm choices \citep{helpman2004export,antras2004global,nocke2014globalization}. More generally, the approach can be used whenever an outcome exhibits kinked behavior driven by an unobserved threshold that varies with observable characteristics.

The remainder of the paper is organized as follows. Section \ref{sec:literature} discusses more related literature and our contribution. Section~\ref{sec:theory} presents the theoretical motivation and illustrates how heterogeneous fixed costs generate set-valued thresholds. Section~\ref{sec:econometrics} develops the econometric framework, and Section \ref{sec:estimator} introduces the properties of the estimator. Section~\ref{sec:empirical} presents the empirical application to Chinese firm data. The final section concludes. 
Details of the asymptotic theory, simulation exercises, and proofs are collected in the Appendix. 

\section{Related Literature and Contribution}\label{sec:literature}
\subsection{Firm Heterogeneity and Trade with Selection into Markets}\label{sec:lit_hetfirms}
This paper relates to the large literature on firm heterogeneity and exporter selection following \citet{melitz2003impact}. Empirical studies typically infer export cutoffs under the assumption of homogeneous fixed export costs, implying a single productivity threshold separating exporters from nonexporters \citep[e.g.,][]{bernard1999exceptional,bernard2007firms,chaney2008distorted,eaton2011anatomy,corcos2012productivity,melitz2012gains,melitz2014heterogeneous,melitz2015new}. 

This framework has been highly successful in explaining endogenous selection of firms into markets and, in turn, endogenous productivity and comparative advantage at the country level. However, it typically abstracts from systematic variation in other variables than productivity between firms.\footnote{Notably, \cite{adao2020aggregate} propose a framework where firm heterogeneity is fundamental and applies to productivity, trade costs, and fixed costs.}
The original model of \cite{melitz2003impact} features two types of determinants of all firm-specific and aggregate endogenous variables: firm-level productivity and otherwise common factors such as market size, trade costs, and market-specific fixed access costs. With this design and set of assumptions, the model predictions are very sharp. First, for all firms in a given country of origin, there is a single destination-market-specific cutoff productivity and marginal firm. For any firm with a lower productivity, serving that market is not profitable, and any firm with a higher productivity will make positive profits when serving that market. As a consequence, the productivity distributions between exporters and nonexporters in any given country and sector will be distinct and non-overlapping, and so will be the sales of firms in the domestic market. Moreover, all firms serving the toughest market (i.e., the one with the fewest firms exporting from a given country) will serve any other market, including the domestic one. And the ranking of firms in the toughest market will be the same in all other markets. Finally, the size rank of firms in the domestic market, which is considered to be the least tough one by assumption, is decisive for the rank of firms in all other markets.

However, these sharp predictions are regularly found to be violated in the data. It is true that foreign markets tend to be tougher than domestic ones on average, and exporters tend to be larger \textit{on average} than nonexporters overall and in the domestic market. Hence, there is a first-order stochastic dominance of exporters over nonexporters in size and productivity. However, some large firms do not export at all and some exporters are rather small and unproductive. For example, \cite{dai2016} document this for China, \cite{helpman2016trade} for Brazil, \cite{blank2022structural} for Germany, and \cite{blum2024abc} for Chile. In fact, the productivity distributions of the exporting and nonexporting firms appear to have a large common support (\citealp{blum2024abc}), and the order of firms within a sector from a single country is not preserved across different destination markets (\citealp{blank2022structural}). This should not be interpreted as contradicting the selection mechanism at work in the model of \cite{melitz2003impact}. However, it suggests that there must be other sources of firm heterogeneity beyond productivity. 

An important dimension of such heterogeneity is that of fixed market-access costs. If fixed costs vary between firms with imperfect (negative) correlation with productivity,  the same productivity level can lead to negative, zero, or positive profits for some firms in a market. Then, market-specific profits are zero for a mass of firms with different levels of productivity. In other words, there is not a single marginal firm (or productivity threshold) per market, but a set. Depending on the variance of fixed market-access costs between firms in a market, the variance in cutoff productivity and the size of marginal firms may be smaller or larger. With a large variation in fixed costs between firms, changes in common economic fundamentals (e.g., through changes in tariffs or other policies) may trigger the exit or entry of firms with a significant variation in their size. This is different from the standard model of \cite{melitz2003impact} where marginal changes in policy would lead to marginal changes in the left tail of the size and productivity distribution of firms.

Unlike many outcome variables such as prices or sales, fundamental variables such as variable ad-valorem trade costs (except for tariffs) or fixed market-access costs are not directly observable. It is common practice to parameterize these variables as a function of observable characteristics (see, e.g., \citealp{anderson2004trade,egger2024QE}). The body of work that addresses fixed market-access costs is substantially smaller than the one that addresses trade costs.

Regarding fixed costs, earlier work documents that they are heterogeneous for various reasons: adjustment costs of investment (\citealp{doyle2001fixed}), financial costs associated with the financial structure of firms (\citealp{wippern1966financial}), differences in credit constraints between firms (\citealp{manova2013credit,muuls2015exporters}), and differences in fixed costs due to tenure heterogeneity in the market between firms (\citealp{egger2025gravity}). A consequence of the heterogeneity in fixed costs with firms that differ in productivity is that cutoff productivity values regarding entry and survival in a market are heterogeneous and depend on the relevant fixed costs.

Our contribution complements this literature by providing an econometric framework that directly estimates how export cutoffs vary with observable fixed-cost proxies, rather than imposing homogeneity or absorbing heterogeneity into unobserved components.

\subsection{Estimating Regression-kink Models with Set-valued Thresholds as a Nonparametric Function of Observables}\label{sec:lit_metrics}
Regarding econometrics, our main contribution is to develop a regression kink–type estimator in which the location of the kink is determined by an unknown function rather than as a scalar parameter. Existing regression-kink models typically assume a scalar-valued cutoff point and focus on identifying changes in slopes at that location \citep{hansen2017regression}. Related approaches in the threshold and regression-kink literature similarly treat the cutoff as homogeneous across observations \citep[e.g.,][]{henderson2017m,chiou2018nonparametric}.

In contrast, we allow the cutoff or threshold value to vary systematically with observation-level characteristics.\footnote{In our application, we permit the threshold productivity (or cutoff domestic sales size) of firms to depend nonparametrically on observable measures of firm-level fixed export costs.} This transforms the econometric object of interest from a point-identified threshold into a threshold contour that depends on observables. The estimation proceeds locally with respect to the observable threshold shifters using kernel weighting combined with a grid search over candidate thresholds, followed by recovery of the cutoff function and inference on its shape.

This generalization has several econometric implications. First, the parameter of interest is a functional rather than a scalar, requiring local estimation and aggregation across values of the threshold shifters (in the application, those are observable fixed-cost proxies). Second, because the threshold enters the model non-differentiably and varies across observations, the bias and convergence properties differ from those in constant-threshold designs, requiring careful treatment of bandwidth choice and inference. Third, the estimator delivers economically interpretable objects. In our application, we can infer how productivity cutoffs behind the firm-level selection into exporting differ across firms, depending on their underlying fixed costs.

Our approach is also related to the use of nonparametric methods in international trade \citep{delgado2002firm} and to recent advances in nonparametric inference \citep{cattaneo2024binscatter}. Relative to this work, our contribution lies in targeting heterogeneous selection thresholds in firm-level decision problems (here, regarding exporting) and providing a practical estimation procedure tailored to such settings.

\color{black}
\section{Theoretical Motivation: The Decision to Export}\label{sec:theory}

We use $i$ to index firms and consider the case of their sales in a single domestic and a single, composite foreign country, so that country indices can be suppressed. Use $S_i$ to denote the total sales or revenue value of firm $i$. And use $G_i$ and $X_i$ for domestic and export sales, respectively, so that $S_i=G_i+X_i$. Clearly, not all firms are exporters, whereby for many $i$ in the data $X_i=0$. However, theoretically all (in data, most) firms sell domestically. Hence, $S_i>0$ for active firms and $S_i=G_i$ for nonexporters.

\subsection{Common Fixed Costs Per Firm Type}
The theory of equilibrium firm selection was developed by \cite{hopenhayn1992entry} and popularized as well as first cast in general large-open-economy equilibrium for firm selection in domestic versus export markets by \cite{melitz2003impact}. This work assumes that firms face two important types of periodical costs: variable costs scale with the size of operations, and fixed costs do not. \cite{melitz2003impact} suggested that domestic-market fixed costs $F^G$ and export-market fixed costs $F^X$ differ and are common to firms and not otherwise indexed. The presence of fixed costs makes operating profits, $O_i=O^G_i+O^X_i$, different from total profits, $\Pi_i=\Pi^G_i+\Pi^X_i$, where $\Pi^X_i=\max\{0,O^X_i-F^X\}$. With domestic and export fixed costs being positive, $F^G,F^X>0$, we obtain $O_i,O^G_i,O^X_i>0$ and $O^G_i=\Pi^G_i+F^G$ as well as $O^X_i=\Pi^X_i+F^X$. While operating profits are positive, total profits are smaller and only nonnegative with $\Pi^G_i,\Pi^X_i\geq 0$ for active firms. Total profits can be written as: 
\begin{eqnarray}
\Pi_i&=&O^G_i-F^G + 1^X_{i}(O^X_i-F^X)\geq 0,
\end{eqnarray}
where $1^X_{i}$ is a binary indicator that is unity for operating export profits, $O^X_i$, being at least as large as exporting fixed costs, $F^X$. A large body of work on firm selection invokes so-called constant markups over marginal costs, $\mu=(1/\beta_G)>1$. In \cite{melitz2003impact}, variable operating costs are inversely multiplicative in productivity, and prices are proportional to these costs as well as the markup $\mu$. The statements $O_i=\beta_G S_i$, $O^G_i=\beta_G G_i$, and $O^X_i=\beta_G X_i$ are customary. Regarding revenues in and operating profits associated with exporting relative to domestic sales for any exporting firm $i$, we obtain
\begin{eqnarray}
\frac{X_i}{G_i}=\frac{O^X_i}{O^G_i}=B>0.
\end{eqnarray}
Hence, for exporters, $S_i=(1+B)G_i$, and for every firm, $S_i=(1+1^X_iB)G_i$, where $1^X_i$ is the aforementioned binary indicator for exporters.

On a further note regarding fixed costs, a customary assumption is that $F^X>F^G$, ensuring that every firm sells domestically and the home market is the last one for companies to exit. We could also state that $F^X=\Phi^XF^G$ with $\Phi^X>1$ commonly assumed. Then, $\Phi^X-1>0$  could be dubbed an export-entry fixed-cost margin. One consequence of the domestic-versus-foreign fixed-cost ranking and of productivity sorting is that all exporters make profits in the domestic market, $\Pi^G_i|(X_i> 0)>0$. We could state that total fixed costs are $F_i=F^G+1^X_iF^X=(1+1^X_i\Phi^X)F^G$.

Let us return to the firm indexation of the sales and operating-profit variables and state that there are only two dimensions through which firms differ in \cite{melitz2003impact}: one is productivity (and, in turn, marginal costs) and one is fixed costs to the extent that only exporters incur $(\Phi^X-1)F^G$. It turns out that because $\Phi^X>1$, not only total fixed costs $F_i$ but also total sales $S_i$ and total operating profits $O_i$ jump upward for the marginal exporter, while total profits are kinked and only change their slope (become steeper) when comparing firms with higher versus lower productivity above versus below the marginal exporter at the threshold-productivity level.

When pooling data for domestic-only sellers (nonexporters) and exporters, we could write kinked profits as
\begin{eqnarray}\label{e:tprofA}
\Pi_i=\left\{\begin{array}{ll}\beta_GG_i-F^G & \text{for domestic sellers},\\
\beta_XG_i-\Phi^X F^G & \text{for exporters},
\end{array}\right.
\end{eqnarray}
where $\beta_X=(1+B)\beta_G$, and 
$B\beta_G G_\ast=(\Phi^X-1)F^G$ for the marginal exporter with $G_\ast$.

The slopes of total profits with respect to domestic sales are
\begin{eqnarray}\label{e:marg}
\frac{\partial\Pi_i}{\partial G_i}=\left\{\begin{array}{ll}
\frac{\partial O_i}{\partial G_i}=\beta_G\equiv (1/\mu) & \text{for nonexporters},\\
\frac{\partial O_i}{\partial G_i}=\beta_X\equiv (1+B)\beta_G & \text{for exporters}.
\end{array}\right .
\end{eqnarray}
The profits the marginal exporter makes are identical to the ones that the domestically-only-selling firm with marginally lower productivity (and marginally lower domestic sales $G_i$) makes. Total profits are a smooth function increasing in $G_i$ with a positive slope $\beta_G$ for domestic-only sellers, a bigger positive slope $(1+B)\beta_G$ for exporters, and a kink for the marginal exporter (that is indifferent about exporting, and where $\Pi^X_i=\Pi^X_\ast=0$. The latter has domestic revenues of $G_i=G_\ast$, which we refer to as the export-cutoff domestic revenues.\footnote{The marginal exporter does not make profits in the export market but does so in the home market.} 

\subsection{Firm-specific Fixed Exporting Costs}
If not only productivity but also foreign market-entry costs were firm-specific (see \citealp{jorgensen2008}, \citealp{arkolakis2010market}, and \citealp{helpman2016trade} for alternative microfoundations relative to \citealp{melitz2003impact}), the parameter $\Phi^X$ would carry a firm subscript. The modified profit function then reads\footnote{Note that one could either assume $F_{i}=F^G+F^X_{i}$ (common domestic fixed costs) or $F_{i}=F^G_{i}+F^X_{i}$ (heterogeneous domestic as well as export fixed costs) without any qualitative impact on the fundamental consequences of fixed-cost heterogeneity, namely that thresholds or cutoffs delineating the options of choice are set-valued.}
\begin{eqnarray}\label{e:tprofB}
\Pi_i=\left\{\begin{array}{ll}\beta_G G_i-F^G & \text{for domestic sellers},\\
\beta_XG_i-\Phi^{X}_i F^G & \text{for exporters}.
\end{array}\right.
\end{eqnarray}

In general, if productivity is the only parameter that is firm-specific in this problem, there is a single productivity threshold -- compactly measured by a single associated domestic sales level $G_\ast$ -- above which all firms are exporters and below which they are not. If another parameter is firm-specific as well, such as foreign market-entry costs, the threshold is represented by a set or contour of threshold productivity related to but not only a function of domestic sales. 

In the absence of measurement error in sales and profit data, one could identify $F^G$ from profits minus markup times domestic sales. And in the absence of measurement error in export sales and profit data, one could then identify the scaling factor $(1+B)$ as well as $\Phi^X$ or $\Phi^{X}_{i}$. However, parameterizing fixed costs and estimating $\beta_G$ as well as $1+B$ is preferable in the presence of stochastic shocks. 

To assist the discussion, we include Figures 1 and 2 here for illustration. These figures involve profits of domestic sellers or from domestic sales, $\Pi(G_i,F^G)$, total profits of exporters, $\Pi(G_i,F_i)$, total fixed costs of domestic sellers and exporters, $F^G$ and $F_i=F^G+F^X_i$, respectively, as well as the domestic-sales threshold for the marginal exporter at observed fixed costs, $\gamma(F_i)$, or at surrogate fixed-cost measures $m_i$, $\gamma(m_i)$ 

\begin{figure}[H]
\begin{center}
\begin{tikzpicture}[scale=0.75]
\draw [<->] (0,-6) -- (0,7);
\draw [->] (0,0) -- (10,0); 
\draw (0,-1) -- (10,5); 
\draw (0,-4) -- (7,5); 
\draw (0,-2.5) -- (7,6.5); 
\draw (0,-5.5) -- (7,3.5);
\draw[dashed, red] (2.1875,-4) -- (2.1875,4);
\draw[dashed, red] (6.5625,-4) -- (6.5625,4);
\draw[decoration={brace,mirror,raise=7pt},decorate,red]
  (0,-2.5) -- node[left=8pt] {Crit.~range of $F_i$} (0,-5.5);
\draw[decoration={brace,mirror,raise=7pt},decorate, red]
  (2.1875,-4) -- node[below=8pt] {Crit.~range of $\gamma(F_i=F^G+F^X_i)$} (6.5625,-4);
\draw (-0.2,0)--(-0.2,0) node{$0$} ;
\draw (-0.55,-1)--(-0.55,-1) node{$-F^G$};
\draw (-0.3,7.3)--(-0.3,7.3) node{$\Pi_i(G_i,F^G);\Pi_i(G_i,F_i)$};
\draw (9,-0.3)--(9,-0.3) node{Dom.~sales $G_i$};
\draw[decoration={brace,raise=7pt},decorate,red]
  (7,6.5) -- node[right=8pt] {$\Pi_i(G_i,F_i)$} (7,3.5);
\draw (9.5,3.7)--(9.5,3.7) node{$\Pi^G_i(G_i,F_i)$};

\end{tikzpicture}
\end{center}
\caption{This figure illustrates the positive selection of firms into export market.}
\label{fig:illus1}
\end{figure}

Figure \ref{fig:illus1} addresses the case of a \textit{positive selection} of firms in the export market as in \cite{melitz2003impact}. However, it considers exporting fixed costs to be firm-specific, $F^X_i$. Figure \ref{fig:illus1} suggests that, if not only the productivity or domestic-sales \textit{running variable} are firm-indexed but so are fixed costs (for exporting, $F^X_i$, or also domestic sales, $F^G_i$), then the threshold (or kink) is not degenerate in a point but becomes a set illustrated by the range where the profits of exporters intersect with those of nonexporters in Figure \ref{fig:illus1}. With the considered variation in total fixed costs $F_i=F^G+1^X_{i}F^X_i$, the range of cutoff domestic sales (with associated productivities) spans the area between the red vertical broken lines on the horizontal axis, denoted by $\gamma(\cdot)$. The propensity of entering the export market and expected total profits rise if (i) the export market offers positive revenues and additional profits without affecting operations (sales and profits) in the domestic market, and (ii) the variance in fixed market-access costs is not too large. A larger variance in fixed costs makes selection and this pattern more fuzzy.\footnote{Whether the fixed costs of only exporting or of exporting as well as domestic selling are firm-specific is immaterial for the general insight that the cutoff productivity and domestic sales are set- rather than scalar-valued.} 

If the fixed market-access costs were higher for the domestic than the foreign market, $F^G>F^X_i$, there could be negative selection in that either none or the less productive firms would be found more likely among the exporters (requiring that the slope of domestic profits exceeds the one of export profits, unlike as drawn in Figure \ref{fig:illus1}).

\begin{figure}

\begin{center}
\begin{tikzpicture}[scale=0.75]
\draw [->] (0,0) -- (0,9);
\draw [->] (0,0) -- (10,0); 
\draw[red] (1,3) parabola (9,9);
\draw[blue] (1,2.5) parabola (9,6);
\draw[teal] (1,1.5) -- (9,4);
\draw[red] (8,9)--(8,9) node{$\gamma_1(m_i)$};
\draw[blue] (8,6)--(8,6) node{$\gamma_2(m_i)$};
\draw[teal] (8,4.1)--(8,4.1) node{$\gamma_3(m_i)$};
\draw (-0.3,9.6)--(-0.3,9.6) node{Dom.~sales marg.~exporter};
\draw (-0.3,9.2)--(-0.3,9.2) node{cutoff functions $\gamma(m_i)$};
\draw (8,-0.3)--(8,-0.3) node{Surrog.~fixed-cost var.~$m_i$};
\draw (-0.2,0)--(-0.2,0) node{$0$} ;

\end{tikzpicture}
\end{center}
    \caption{This figure illustrates the unknown threshold contour for different functional forms of true firm-specific fixed costs, $F_{i}=F^G+1^X_{i}F^X_{i}$, as a function of an observable surrogate variable of these fixed costs, $F_i=F(m_{i})$.}
    \label{fig:illus2}
\end{figure}

Figure \ref{fig:illus2} addresses another point, related to the measurement of fixed costs in the data. Recall that economic theory does not provide guidance about the exact functional form of the relationship between observable surrogate variables and true fixed market-access costs. The figure illustrates how different forms of said relationship might affect the shape of a given domestic-sales-fixed-cost contour when mapped into domestic-sales-observable-fixed-cost-variable space.

\subsection{Discussion}
We close this section by discussing the policy implication. In \citet{melitz2003impact}, domestic and export fixed costs, here denoted as $F^G$ and $F^X$, are both common scalars for all firms. \citet{melitz2003impact} is silent about how observables used for the empirical parameterization of fixed costs are linked to the total fixed costs of operation, $F_i=F^G+1^X_iF^X$. But empirical work has to take a stance on this. We are primarily interested in the case, where the fixed costs $F^X$ or even $F^G$ are firm-specific as had been considered in \cite{doyle2001fixed},  \cite{jorgensen2008}, \cite{manova2013credit}, \cite{muuls2015exporters}, \cite{helpman2016trade}, or \cite{egger2025gravity}. 

If both productivity and fixed costs are stochastic (i.e., they differ between firms from the same origin in a sales market), two consequences emerge theoretically. First, marginal changes in, say, fixed exporting costs $F^X_{i}$ can lead to the entry of large firms into the market, even if the average fixed market-entry costs are small. This would not happen in \citet{melitz2003impact}, as the marginal entrants (or exiters, if fixed costs are raised) are always the smallest among all firms that can survive in the market (or, with exit, that sold products prior to their exit from the market). Moreover, in \citet{melitz2003impact} there is a sharp prediction of the pecking order of firms when comparing different sales markets with each other. For instance, the firms selling to the toughest foreign market will serve all less tough markets, too, and market toughness is a common feature to all firms. Similarly, all exporters will be the largest sellers in the domestic market (as long as more firms sell domestically than export, which is common). This is not the case with firm-specific fixed costs, specifically, fixed exporting costs. 
    
Overall, we need to know how specific policy measures affect the respective fixed costs, and how firms differ in terms of the latter to predict how the competitive environment (through entry and exit) will be shaped by policy. In any case, with heterogeneity in productivity as well as fixed costs, the predictions model predictions are more nuanced than in \citet{melitz2003impact}. For instance, \cite{blank2022structural} document dramatic deviations from the pecking order predicted by an original Melitz model of firms across different sales markets of German exporters of goods as well as services. However, this does not invalidate Melitz' insights regarding the endogenous selection of firms into markets or the endogeneity of the entry-dependent average productivity observed. Melitz' basic arguments remain fully intact with heterogeneous fixed costs as well, but the predictions regarding the size of entrants and their pecking order across markets can be substantially different. 

\section{From Theory to Econometric Analysis}\label{sec:econometrics}

Economic theory provides direct implications for our econometric analysis. However, note that the relationships introduced above are measured in levels and profits are measured as a difference of operating profits and fixed costs. The latter represents a linear relationship in levels. In empirical work, it is customary to regularize data using logarithmic or hyperbolic sine transformations for numerical reasons. The structural additive, linear-in-levels form of the profit function does not lend itself to a logarithmic or hyperbolic-sine regularization, as the resulting form would then be nonlinear after transformation. Therefore, we use customary standardization in terms of the variation in the data. See Section \ref{sec:data} for more details.

To establish a clear connection between the theoretical relationship in terms of unnormalized data and the econometric analysis based on normalized data, we first introduce some further notation. Specifically, we use the convention of using lower-case letters to refer to normalized counterparts to unnormalized ones in Section \ref{sec:theory}. To this end, let $\pi_i$ and $g_i$
denote the normalized measures of total profits $\Pi_i$ and domestic sales $G_i$ of firm $i$, respectively. Moreover, let us introduce some firm-specific profit shifters in a column vector $x_i$. The latter can include a constant as well as variables which reasonably affect profits through domestic fixed costs.\footnote{In the model writeup in Section \ref{sec:theory}, we did not explicitly account for firm-specific elements in domestic fixed costs. This would have required an index $i$ on $F^G$. Here, we permit specific ownership indicator types as candidate variables in $x_i$. Specifically, we consider foreign ownership and the neighbor firm's peer effect in the empirical analysis. In general, there is no need for considering such variables apart from the constant in $x_i$.} Finally, we introduce a stochastic, firm-specific error term $u_i$.

The piecewise-linear model in equations \eqref{e:tprofA} and \eqref{e:tprofB} based on normalized data can then be expressed as a regression-kink model:
\begin{equation}
\label{eq:kink}
\pi_i = \beta_G (g_i-\gamma)_{-} + \beta_X (g_i - \gamma)_{+} + x_i'\beta_C + u_i,
\end{equation}
where, for any $a$, $(a)_{-}=\min\{a,0\}$ and $(a)_{+}=\max\{a,0\}$.
Hence, the plus and minus subscripts indicate domains of the running variable above and below the cutoff. 
Here,  $\beta_G$ corresponds to $(1/\mu)$  and $\beta_X$ corresponds to $(1+B)\beta_G$ in Section \ref{sec:theory}. These are the structural slope parameters on domestic sales in \eqref{e:marg}, translating them into operating profits for the two types of firms. With positive selection, $\beta_G$ would apply to domestic sellers and $\beta_X$ to exporters, and with negative selection, the opposite would be the case. From a theoretical perspective, under the assumptions adopted in \cite{melitz2003impact}, one would expect a positive selection with $\beta_X>\beta_G$, as $\beta_X/\beta_G=(1+B)>1$. The coefficient $\beta_C$ on $x_i$ measures the impact of additional controls on outcome.

Equation \eqref{eq:kink} follows the regression-kink framework with an unknown kink point $\gamma$, as studied by \cite{hansen2017regression}. 
The threshold $\gamma$ depends on firm-level fixed costs $F_{i}$, with the theoretical model implying the linear relationship
\[  
\gamma(F_{i})=F_{i}/(B\beta_G).
\]
Substituting $\gamma(F_{i})$ into \eqref{eq:kink}, we obtain:
\begin{equation}
\label{eq:kink12}
\pi_i = \beta_G (g_i-\gamma(F_i))_{-} + \beta_X (g_i - \gamma(F_i))_{+} + x_i'\beta_C + u_i.
\end{equation}

We emphasize that direct measures of total sales $S_i$ (or normalized $s_i$) and of domestic sales $G_i=S_i-X_i$ (or normalized $g_i$) are relatively straightforward to obtain. In contrast, direct measures of fixed costs $F_{i}$ are not directly observable. Hence, $F_i$ in $\gamma(F_i)$ must be based on observables that potentially are not linearly related to $F_i$.\footnote{For example, we use financial cost or fixed assets as surrogate variables for $F_i$ in Section \ref{sec:empirical}.} In what follows, we denote empirical measures or surrogate variables of $F_i$ by $m_i$. Moreover, we allow $\gamma(m_i)$ to take a nonparametric form. 

Ultimately, we consider the model:
\begin{equation}
\label{eq:kink2}
\pi_i = \beta_G (g_i-\gamma(m_i))_{-} + \beta_X (g_i - \gamma(m_i))_{+} + x_i'\beta_C + u_i.
\end{equation}
Comparing \eqref{e:tprofB} and \eqref{eq:kink2}, we replace $\Pi_i$ and $G_i$ with their normalized counterparts $\pi_i$ and $g_i$, respectively, and we substitute $F_{i}$ with $m_i$ to emphasize that this is a surrogate for the firm-level fixed-cost measure.
The following table summarizes the respective notations in theoretical versus the empirical frameworks. 
We observe data of $(\pi_i,g_i,m_i,x_i')'$, and the parameters of interest are $(\beta_G,\beta_X,\beta_C',\gamma(m_i))'$. The function $\gamma(m_i)$ is embedded in the exporting decision as an unknown function, $1[g_i-\gamma(m_i)>0]$, where $1[\cdot]$ denotes the indicator function. 
Such decision function is non-differentiable by construction and hence challenging existing econometric methods. 
Therefore, we develop a novel nonparametric estimation method in the following section and derive its asymptotic properties.

\begin{table}[ht]
\centering
\renewcommand{\arraystretch}{1.1}
\begin{tabular}{|l|l|l|l|}
\hline
\textbf{Theoretical Model} &           & \textbf{Econometric Model}       &                \\ \hline
Profit                     & $\Pi_i$   & Normalized profit                & $\pi_i$        \\ \hline
Domestic sales             & $G_i$     & Normalized dom.~sales            & $g_i$          \\ \hline
Fixed costs                & $F_{i}$   & Observable fixed-cost measures   & $m_i$          \\ \hline
                           &           & Other covariates                 & $x_i$          \\ \hline
\end{tabular}
\caption{Comparison of Theoretical and Econometric Models}
\end{table}

Note that the adopted approach will identify a threshold value of domestic sales, $\gamma(m_i)$ under the adopted assumptions underlying the estimator, if it exists. In the presence of a threshold, we know that there is selection in exporting versus domestic sales. And the signs and values of slope parameters $\{\beta_G,\beta_X\}$ will inform us about whether the selection is positive or negative. Positive selection means that firms with higher domestic sales will more likely be exporters and be found to the right of (or above) the threshold. Then, we would expect $\{\beta_G,\beta_X\}>0$ and $\beta_G<\beta_X$. In contrast, negative selection means that firms with higher domestic sales to the right of the threshold are less likely to be exporters. 

\section{Nonparametric Regression-kink Estimation}\label{sec:estimator}
\subsection{Overview of the Estimator}\label{sec:overview}
To estimate the kink threshold in a nonparametric fashion, we adopt a kernel-based approach. 
The intuition behind our method is straightforward, though the technical implementation is complicated by the inherent non-differentiability at the kink point. 
This section outlines the key ideas and provides heuristic insights. 
Full technical details are deferred to the Appendix.

Our approach focuses on local kernel estimation at each point $m$. 
Assuming that $\gamma(\cdot)$ is continuously differentiable, we approximate it by a constant in a local neighborhood. 
This allows us to perform a locally weighted least squares regression for any interior point $m$ in the support of $m_i$. 
Let $\theta = (\beta_G,\beta_X,\beta_C',\gamma(m))'$ denote the vector of parameters. 
We estimate it via:
\begin{equation}\label{eq:nprk}
    \hat{\theta} = \min_{\theta}\sum_{i=1}^n k_{i,n} 
    \left( \pi_i - \beta_G (g_i-\gamma)_{-} - \beta_X (g_i - \gamma)_{+} - x_i'\beta_C    \right)^2,
\end{equation}
where 
\[k_{i,n}= \frac{ K\left( \frac{m_i-m}{b_n} \right) }{\sum_{j=1}^n K\left(\frac{m_j-m}{b_n} \right)},
\]
$K(\cdot)$ denotes a kernel function, and $b_n$ the bandwidth. 
In the special case where $K(t)= 1[|t|<1]$ (i.e., the uniform kernel), this setup essentially reduces to the standard regression kink model in \eqref{eq:kink}, but applied only to observations in a local neighborhood of $m$, specifically those satisfying $\{|m_j-m|<b_n \}$. 

The estimation procedure for \eqref{eq:nprk} is conducted in two steps. 
First, for any fixed candidate value of $\gamma$, we solve a weighted least squares problem:
\begin{equation}\label{eq:coeff}
    S_n(\gamma) = \min_{(\beta_G,\beta_X,\beta_C)}\sum_{i=1}^n k_{i,n}
    \Big( \pi_i - \beta_G (g_i-\gamma)_{-} - \beta_X (g_i - \gamma)_{+} - x_i'\beta_C   \Big)^2.
\end{equation}
This step amounts to regressing $\pi_i$ on the covariates $((g_i-\gamma)_{-},(g_i-\gamma)_{+},x_i')'$. 

Next, we perform a grid search over possible values of $\gamma$. 
To facilitate this, we normalize $g_i$ to have zero mean and unit variance, and then define a fine grid $\Gamma_n$. 
For each $\gamma \in \Gamma_n$, we compute $S_n(\gamma)$ and obtain the estimator:
\begin{equation}\label{eq:gamma}
\hat{\gamma} = \arg\min_{\gamma \in \Gamma_n} S_n(\gamma).
\end{equation}
Finally, we plug in $\hat{\gamma}$ into \eqref{eq:coeff} to obtain the estimators for the coefficients $\hat{\beta}= (\hat{\beta}_G(\hat{\gamma}),\hat{\beta}_X(\hat{\gamma}),\hat{\beta}_C(\hat{\gamma})'  )'$. 
We note that, unlike in the customary threshold model \citep[e.g.,][]{hansen2000sample}, the objective function $S_n(\gamma)$ in the regression-kink model is continuous in $\gamma$. Therefore, the solution in \eqref{eq:gamma} is unique almost surely.

This entire procedure is carried out pointwise for each $m$, resulting in the estimate $\hat{\gamma} = \hat{\gamma}(m)$ as in \eqref{eq:gamma}. 
By repeating this process over a grid $M$ of values $m$ (for example, the central 98\% quantile range of $m_i$), we obtain the estimated kink function $\hat{\gamma}(\cdot)$.

\subsection{Summary of the Asymptotic Properties}\label{sec:asymptotics}
This section summarizes the large sample properties of the proposed estimator and explains how they inform the empirical implementation in Section \ref{sec:empirical}. All formal assumptions and proofs are provided in the Appendix.

We first establish consistency of the estimated threshold function $\hat{\gamma}(m)$. Under standard regularity conditions, including smoothness of the population threshold function and appropriate bandwidth choices, $\hat{\gamma}(m)$ converges uniformly to the true threshold function $\gamma_0(m)$ over the interior of the support of the fixed-cost proxy $m$. The uniform convergence rate combines a variance component driven by kernel smoothing and a bias component due to local approximation. This result implies that the estimated threshold contour recovers the shape of exporter selection boundaries in large samples. 

Uniform consistency of $\hat{\gamma}(m)$ is particularly important for the second-step estimation of the regression coefficients. In the empirical application, the threshold estimate is evaluated at observation-specific values of the fixed-cost proxy. The uniform convergence result guarantees that the resulting first-stage error remains controlled across firms and does not accumulate in the second step. As a consequence, the estimator of the regression coefficients $\beta=(\beta_G,\beta_X,\beta_C')'$ is root-$n$ consistent and asymptotically normal.

These asymptotic properties also justify the use of conventional variance estimation and bootstrap procedures for inference on $\beta$. While standard bootstrap methods are known to fail for inference on the threshold parameter $\gamma(\cdot)$ \citep{yu2014bootstrap}, the present setting differs because the bootstrap is used only for the regression coefficients, which are smooth functionals and admit an asymptotically linear representation. In the empirical application, we therefore rely on the bootstrap to construct confidence intervals for $\beta$.

From a practical perspective, the asymptotic theory informs key implementation choices in the empirical analysis. Bandwidths are selected to satisfy undersmoothing conditions that render first-stage biases asymptotically negligible for inference on $\beta$, while still allowing precise estimation of the threshold contour. The restriction of attention to the interior of the support of $m$ avoids boundary effects and ensures uniform approximation quality. These choices are reflected in the stability of the estimated threshold contours and in the robustness of the regression-coefficient estimates across specifications.

Overall, the asymptotic results establish that the proposed estimator delivers two complementary objects of interest: a nonparametric estimate of heterogeneous export thresholds and standard, well-behaved inference for the parameters governing the slopes of firms' profit functions. This combination underpins the empirical findings in Section~\ref{sec:empirical} and supports the economic interpretation of heterogeneous fixed costs as a key driver of exporter selection.

\subsection{Extension: Allowing for  Endogeneity}\label{sec:endo}

The baseline model assumes that the running variable $g_i$ is exogenous conditional on
observable covariates. Specifically, profits satisfy \eqref{eq:kink2} with $\mathbb{E}[u_i|g_i,m_i,x_i]=0$. In many applications, however, $g_i$ may be correlated with unobserved determinants of profits, violating this exogeneity condition.

To accommodate endogeneity of $g_i$, we adopt a control-function approach. Let $w_i$ denote an instrument that affects $g_i$ but is excluded from the profit equation. We specify the
reduced-form relationship
\begin{equation}
g_i = \eta'w_i + v_i, 
\label{eq:first_stage}
\end{equation}
where $v_i$ captures the endogenous component of $g_i$. The control variable is defined as the residual $\hat v_i$ from estimating \eqref{eq:first_stage}. Identification relies on the condition $\mathbb{E}[u_i|m_i,x_i,v_i]=0$.

Under this assumption, the profit equation can be written as
\begin{equation}\label{eq:profit_cf}
\pi_i = \beta_G (g_i-\gamma(m_i))_- + \beta_X (g_i-\gamma(m_i))_+ + x_i'\beta_C + \beta_V \hat v_i + \varepsilon_i,
\end{equation}
where $\varepsilon_i$ is mean independent of $(g_i,m_i,x_i,\hat v_i)$. The control variable $\hat v_i$ is included as an additional regressor in both the threshold estimation and the second-stage regression, restoring conditional exogeneity of the running variable.

Estimation proceeds exactly as in the baseline case, treating $\hat v_i$ as an additional covariate. In the first step, the threshold function $\gamma(m)$ is estimated locally using kernel weighting, now conditioning on $(m_i,\hat v_i)$. In the second step,
the regression coefficients $(\beta_G,\beta_X,\beta_C',\beta_V)'$ are estimated using the estimated threshold.

The asymptotic properties derived for the exogenous case continue to hold under this extension. Since the control variable enters the model smoothly and is estimated at a parametric rate, it does not affect the convergence rate of the threshold estimator. Uniform consistency of $\hat{\gamma}(m)$ over the interior of the support is preserved, and the estimator of the regression coefficients remains root-$n$ consistent and asymptotically normal. 

In Section \ref{sec:empirical}, we implement this extension using an external instrument and show that accounting for endogeneity affects the magnitude of the estimates but leaves the qualitative pattern of heterogeneous threshold behavior unchanged.

\section{Empirical Application}\label{sec:empirical}
\subsection{Background}\label{sec:background}
While, along with total and export sales, total profits as well as surrogate variables for fixed costs are readily observable, empirical work on firm selection largely focuses on price data (see \citealp{khandelwal2010long}, \citealp{hallak2013}) or sales data and related productivity estimates for identifying threshold effects (see \citealp{aw1995}, \citealp{vanbiesebroeck2005}, \citealp{arkolakis2010market}, \citealp{bustos2011trade}, \citealp{duan2022}, \citealp{arkolakis2021extensive}). Instead, data on profits are virtually never used in spite of their wide availability from accounting data. We exploit such data for China and describe them in more details in the next subsection.

\subsection{Data}\label{sec:data}
We employ data from the Chinese Annual Survey of Industrial Firms (CASIF), which contains accounting data for a large set of companies throughout China. The data are available between 1998 and 2015, but the reporting standards and the information provided vary over time. Due to varying revenue-threshold provisions regarding reporting requirements, the sample sizes of firms differ over the years. As the proposed econometric model involves nonparametric estimation, we prefer data with a larger sample size. For that reason, we select data from the year 2008, as the sample size is particularly large in that year. 

We apply cleaning steps that are customary prior to using the data (see \citealp{upward2013weighing}, and \citealp{wang2013trading}, for typical data-cleaning steps with the CASIF data). First, we require total revenues (domestic sales plus exports), export revenues, total profits, financial cost and fixed assets (both alternative measures of total domestic plus exporting fixed costs) to be (i) reported at all and (ii) to be positive for a company.\footnote{We additionally require foreign capital holdings in a firm to take on nonnegative values.} Second, there exist some extremely large values for total profits, which compromise the least-squares estimation. Therefore, we trim off the the firms with the largest 15\% of total profits.\footnote{We experimented with alternative trimming thresholds, which led to similar results and findings as reported later. Note that this is not surprising, as one would expect the exporting-threshold productivity or domestic sales not to lie in the right tail of the profit distribution with the data at hand and in general. In any case, such adjustments to avoid the influence of obvious data errors are customary with firm-level data in general.} Overall, this leads to a cross-sectional data sample with $n=187,720$ firms, which we use for estimation.

The dependent variable in our analysis is based on the total profit of firm $i$, $\Pi_i$. The key kink-generating variables are domestic-sales, $G_i$, and financial-cost measures as surrogate variables for fixed market-access costs, $(F_i,F^X_i)$. In this regard, we employ the firm-specific measures \texttt{Financial cost}$_i$ and \texttt{Fixed assets}$_i$ available in the data. This choice is motivated by a literature in finance regarding the heterogeneous financial structure and associated financial costs of firms in finance (\citealp{wippern1966financial}) and the one on fixed assets as collateral in the literature on financial constraints and trade (\citealp{manova2008credit,manova2013credit}).\footnote{We assume that the considered financial variables affect fixed costs conditional on domestic sales, the running variable, which is a function of prices and output.} Finally, we consider two variables behind the shifters included in $x_i$: \texttt{Rel.for.capital}$_i$ (relative foreign capital share in total capital of firm $i$), and \texttt{Avg.neighb.sales}$_i$ (the average neighboring firms' total sales). \texttt{Rel.for.capital}$_i$ is included in $x_i$, because earlier work suggested that foreign-company involvement may affect profits due to profit shifting, the transfer of technology, etc. \texttt{Avg.neighb.sales}$_i$ is included, because close-by firms may induce spillovers. For the latter, we determine a circle with a diameter of one-degree around the firm $i$ and compute the average total sales of the other firms there. It should be noted that those spillovers do not run one-way from other firms to firm $i$ but also in the other direction. Therefore, \texttt{Avg.neighb.sales}$_i$ should be considered an endogenous regressor and will be in some of the models (see \citealp{kelejian1998generalized}). In any case, we hypothesize that \texttt{Rel.for.capital}$_i$ (due to better access to internal finance) and  \texttt{Avg.neighb.sales}$_i$ (due to knowledge spillovers regarding market access) have the potential of shifting fixed market-access costs.\footnote{Note that one could alternatively consider including such regressors in $m_i$ rather than $x_i$ However, \texttt{Rel.for.capital} is positive for relatively few firms, so it cannot simply be included as a factor in $m_i$.}

Table \ref{tab:summary_stats} presents some summary statistics of the cleaned variables used.\footnote{Beyond what is shown in the table, we could add that 40\% of the included companies are exporters and 8\% receive foreign capital.} With variables in levels as the ones considered here, some regularization is necessary to normalize the large degree of heterogeneity as documented in Table \ref{tab:summary_stats}.\footnote{The standard deviation is considerably larger than the mean for the continuous variables, and the interquartile range (the distance between the 75th and the 25th quantile, $Q(0.75)-Q(0.25)$) attests to large tails in the unnormalized data.} This is the case to prevent an excessive degree of heteroskedasticity in estimation and to scale some of the regression parameters. It is customary to log-transform data in international economics -- if not the outcome variable, then at least the explanatory ones used in a log-additive index in exponential-family models. However, it should be noted that total profits (i) are additive in levels in a function that is proportional to domestic sales and (ii) governed by a switching function between domestic sellers and exporters. Hence, with this outcome, a log transformation is not useful, as mentioned above. 

Accordingly, we demean each continuous regressor and normalize the resulting variable by the standard deviation.\footnote{We can normalize the variables in the proposed way without affecting the theoretical relationship for three reasons: (i) the transformations are monotone; (ii) we estimate separate constants to the left and the right of the threshold; and (iii) the threshold set is based on nonparametric estimation which is particularly robust to the transformation of variables. The proposed standardization of variables facilitates the interpretation of the estimated parameters. A one-standard-deviation increase in an included covariate with some parameter $\beta$ then leads to a $\beta$-fold change in standardized profits.}
We refer to the standardized total profits, domestic sales, and financial costs as $\pi_i$, $g_i$, and $m_i$, respectively. The latter two together generate the set of productivity- or domestic-sales-cutoff parameters for marginal exporters.

\begin{table}[ht]
\centering
\caption{Summary Statistics for $n=187,720$ Firms in 2008.}
\begin{threeparttable}
\begin{adjustbox}{width=0.9\textwidth}
\begin{tabular}{lrrrrrrr}
\toprule
Variable & Mean & Std. Dev. & Min & $Q(0.25)$ & $Q(0.5)$ & $Q(0.75)$ & Max \\
\midrule
\multicolumn{8}{c}{Variable underlying normalized firm outcome $\pi_i$: profits $\Pi_i$} \\
\midrule
$\Pi_i$ & 1,586.29 & 1,851.77 & 1.00 & 253.00 & 807.00 & 2254.00 & 8007.00 \\
\midrule
\multicolumn{8}{c}{Variable underlying normalized firm-level $g_i$: domestic sales $G_i$} \\
\midrule
$G_i$ & 38,177.35 & 78,095.16 & 1.00 & 9,213.00 & 19,560.00 & 43,164.00 & 7,867,483.00 \\
\midrule
\multicolumn{8}{c}{Variables alternatively underlying normalized firm-level $m_i$: fixed costs $\phi_i$} \\
\midrule
\texttt{Financial cost}$_i$ & 522.73 & 2,215.06 & 0.00 & 15.00 & 119.00 & 410.50 & 254,784.00 \\
\texttt{Fixed assets}$_i$ & 11,314.66 & 54,136.14 & 1.00 & 1,539.00 & 3,882.00 & 9,595.00 & 9,207,768.00 \\
\midrule
\multicolumn{8}{c}{Variables underlying as normalized profit shifters in $x_i$} \\
\midrule
\texttt{Rel.for.cap.}$_i$ & 0.09 & 1.53 & 0.00 & 0.00 & 0.00 & 0.00 & 349.04 \\
\texttt{Avg.neighb.sales}$_i$ & 271.25 & 4,114.29 & 5.77 & 14.82 & 57.46 & 140.94 & 822,934.00 \\
\bottomrule
\end{tabular}
\end{adjustbox}
\begin{tablenotes}
\footnotesize
\item Notes: $i$ indexes firms, and $\Pi_i$ is total profit (in local currency). 
\texttt{Financial cost} represents financing costs, and \texttt{Fixed assets} denotes fixed capital assets. \texttt{Rel.for.cap} measures the ratio of foreign asset holdings in total assets of a firm.  
\texttt{Avg.neighb.sales} is an average of the total sales of other firms in the same circle around firm $i$ with one-degree diameter.
\end{tablenotes}
\end{threeparttable}
\label{tab:summary_stats}
\end{table}

\subsection{Results}

We use the proposed estimator to estimate the set (or contours) of threshold productivities depending on fixed costs with the aforementioned data. 
For robustness, we consider three different models.
Model 1 is the benchmark as in \eqref{eq:model}, 
\begin{equation}\label{eq:emp:model1}
    \pi_i = \beta_G (g_i-\gamma(m_i))_{-} + \beta_X (g_i - \gamma(m_i))_{+} + x_i'\beta_C + u_i,
\end{equation}
where $\pi_i$ are normalized profits, $g_i$ are normalized domestic sales, $m_i$ is one of the two alternative normalized measures of fixed costs (\texttt{Financial cost}$_i$ or \texttt{Fixed assets}$_i$) , and $x_i$ includes the foreign capital share relative to total capital (\texttt{Rel.for.cap.}$_i$) and the average of neighboring firms' total sales (\texttt{Avg.neighb.sales}$_i$). 
Beyond those variables we include a constant, whose estimate we suppress. 

Our second model considers that $g_i$ might be endogenous. After all, the markup charged by companies affects the consumer price and subsequently (domestic as well as foreign) demand. Moreover, local demand shocks will induce simultaneous effects on sales as well as associated profits. How important these sources of endogeneity are for parameter estimation depends on the data and is an empirical question. To address this potential endogeneity, we create two instrumental variables. 

The first one we call \texttt{Demand~from~firms}$_i$, and it represents an output-share-weighted distance to the potential users of firm $i$'s output as an input in China's domestic market. To this end, we use China's sector-to-sector input-output table for the year 2007 in conjunction with firm-level sales data from CASIF. The latter data contain a sector identifier and permit a mapping of firms to input-output sectors. We assume that a firm in our data is representative of the sector in which it belongs. Accordingly, the shares of use sectors of this firm's output are the same as those of the making sector in which the firm belongs. Moreover, we can compute the share that each firm other than $i$ has in the total sales of the use sector and intermediate purchases of the making sector. As we have information on each firm's coordinates based on the spatial information available from CASIF, we can compute firm-to-firm distances. We assume that demand declines with distance in a log-log relationship with a unitary parameter on distance (which is archtypical for gravity relationships). The latter establishes smaller weights in the input demand of more distant than from less distant potential input users. Putting together, we obtain a weighted input demand from the other firms in CASIF which reflects input-output as well as gravity (distance) aspects in the variable \texttt{Demand~from~firms}$_i$. What is key here is that the latter is systematically related to the demand of other companies for firm $i$'s output without requiring that this prediction is accurate. 

Second, we establish a variable called \texttt{Min.dist.port}$_i$, which is the distance of firm $i$ to the closest international sea port among the $14$ major considered ones.\footnote{The 14 major ports are, in alphabetical order (with rounded latitudes and longitudes in parentheses): 
Basuo port (lat.~19.0915, lon.~ 108.6702); 
Dalian port (lat.~38.9167; lon.~121.6833);
Guangzhou port (lat.~23.0939, lon.~113.4378); Haikou port (lat.~20.0526, lon.~20.0526); Nanjing port (lat.~32.2125, lon.~118.8120); Nantong port (lat.~32.0303, lon.~120.8747); Ningbo-Zhoushan port (lat.~29.8667, lon.~121.5500); Qingdao port (lat.~36.0833, lon.~120.3170); Rizhao port (lat.~35.4167, lon.~119.4333); Shanghai port (lat.~~31.3664, lon.~121.6147); Shenzhen port (lat.~22.5039, lon.~113.9447); Tianjin port (lat.~38.9727, lon.~117.7845); 
Xiamen port (lat.~24.5023, lon.~118.0845); and Yingkou port (lat.~40.6936, lon.~122.2662).} 
This variable uses two ingredients: (i) the information on the longitude and latitude of firm $i$ as already used with the construction of \texttt{Avg.neighb.sales}$_i$ as well as \texttt{Demand~from~firms}$_i$; and (ii) 
the information on the longitude and latitude of each one of the $14$ Chinese ports mentioned in the previous footnote. We compute the great-circle distance of each firm to each port and pick the shortest one to obtain \texttt{Min.dist.port}$_i$. 
This variable is a measure of variable export-market-access costs associated with a firm's location within China. While this variable should directly affect exports, it should matter also for the relative importance of the domestic market for a firm.

We collect these instrumental variables in the row vector $w_i=$ (\texttt{Demand~from~firms}$_i$,\ \texttt{Min.dist.port}$_i$). We apply the estimator in Section \ref{sec:endo} by first regressing $g_i$ on $w_i$ and then adding the residual $\hat{v}_{gi}$ as an additional regressor in a control-function approach. 
This leads to our Model 2:
\begin{equation}\label{eq:emp:model2}
    \pi_i = \beta_G (g_i-\gamma(m_i))_{-} + \beta_X (g_i - \gamma(m_i))_{+} + x_i'\beta_C + \hat{v}_{gi}\beta_V + u_i,
\end{equation}
which is the same as Model 1 except for the control function. 

Our third model considers that domestic sales ($g_i$) and the average sales of neighboring firms (\texttt{Avg.neighb.sales}$_i$) are both endogenous. The latter variable's endogeneity is plausible as neighboring firms' sales capture peer effects and their relevance indicates the presence of inter-firm (productivity or other) spillovers. Due to circular relationships, weighted other firms' outcomes will likely induce an estimation bias (\citealp{kelejian1998generalized}). Accordingly, we regress these two variables on the instruments $w_i$ and obtain the residuals $\hat{v}_{gi}$ and $\hat{v}_{xi}$, respectively, one for each endogenous regressor. 
Adding both residuals to the regression function, we specify Model 3 as:
\begin{equation}\label{eq:emp:model3}
    \pi_i = \beta_G (g_i-\gamma(m_i))_{-} + \beta_X (g_i - \gamma(m_i))_{+} + x_i'\beta_C + \hat{v}_{gi}\beta_{V_g} + \hat{v}_{xi}\beta_{V_x} + u_i,
\end{equation}
which is the same as Model 1, except for the control function. 
We estimate the threshold $\gamma(\cdot)$ first. 
To this end, we employ the nonparametric regression-kink model \eqref{eq:nprk}, using a standard Gaussian kernel with a rule-of-thumb bandwidth $b_n = n^{-1/5}$ in all three models. 
In addition, we also implement the undersmoothing bandwidth $b_n = n^{-1/3.5}$ as a robustness check. 
The estimated threshold $\hat\gamma(\cdot)$ becomes less smooth, as a consequence of the undersmoothing bandwidth. 
But its increasing feature and the estimation results for other coefficients are very similar, suggesting that our findings are robust to the bandwidth choice. 
We present these additional results in the Appendix. 

Figure \ref{fig:emp_threshold1} plots the estimated threshold contour based on Model 1 in \eqref{eq:emp:model1}, using \texttt{Financial cost}$_i$ (left panel) and \texttt{Fixed assets}$_i$ (right panel) respectively as alternative surrogate variables for fixed costs. Both plots exhibit clearly nonlinear and increasing patterns, suggesting that the domestic-sales-threshold increases with higher fixed costs at the firm level, as expected. 
Figures \ref{fig:emp_threshold2} and \ref{fig:emp_threshold3} display the estimated thresholds based on Models 2 and 3, respectively. 
All these figures clearly and robustly suggest that the threshold assumes higher domestic-sales values (productivity) for higher fixed costs.\footnote{Assumption~2.5 concerns the smoothness of the \emph{population} threshold function $\gamma_0(m)$ and does not impose smoothness on its finite sample estimator. Our local constant kernel estimator $\hat{\gamma}(m)$ is evaluated pointwise on a grid and then connected for visualization, which can produce apparent non-smooth features in finite samples. These features reflect estimation variability rather than violations of the smoothness assumption on $\gamma_0(m)$.
} 

In addition, the estimated threshold contour features an economically meaningful dispersion pattern across firms. The difference in estimated fixed costs between the 20th and 80th percentiles corresponds to a gap in the implied export cutoff of approximately 25\% in the scaled productivity (domestic-sales) dimension. This magnitude is nontrivial relative to the overall dispersion of the running variable and to the observed exporter-nonexporter gap in the data, indicating that fixed-cost heterogeneity generates meaningful variation in export selection thresholds. 
This suggests that heterogeneity in fixed costs plays an important, but not implausibly large, role in shaping export participation.

\begin{figure}[H]
\centering
\includegraphics[width=0.45\linewidth]{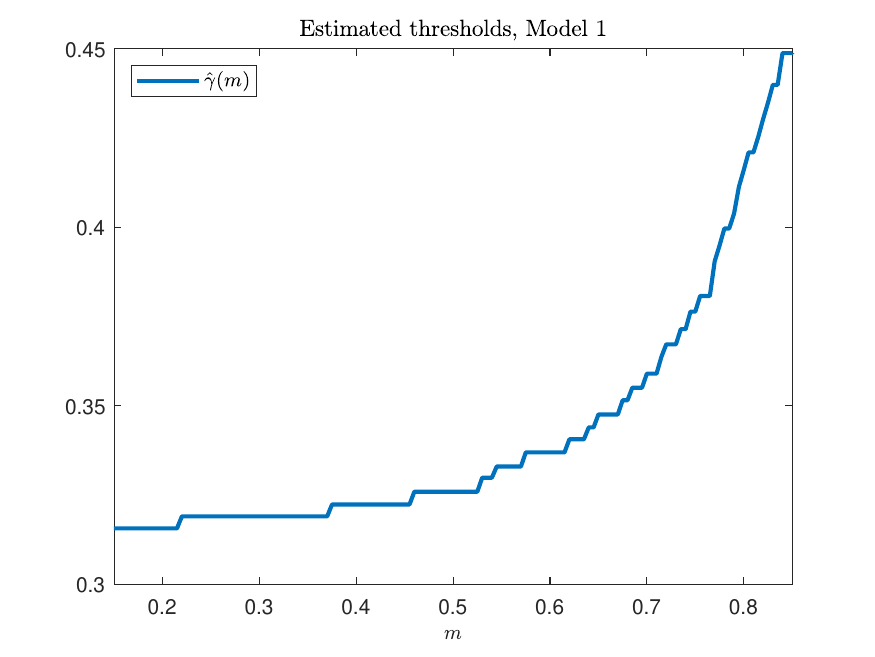}
\includegraphics[width=0.45\linewidth]{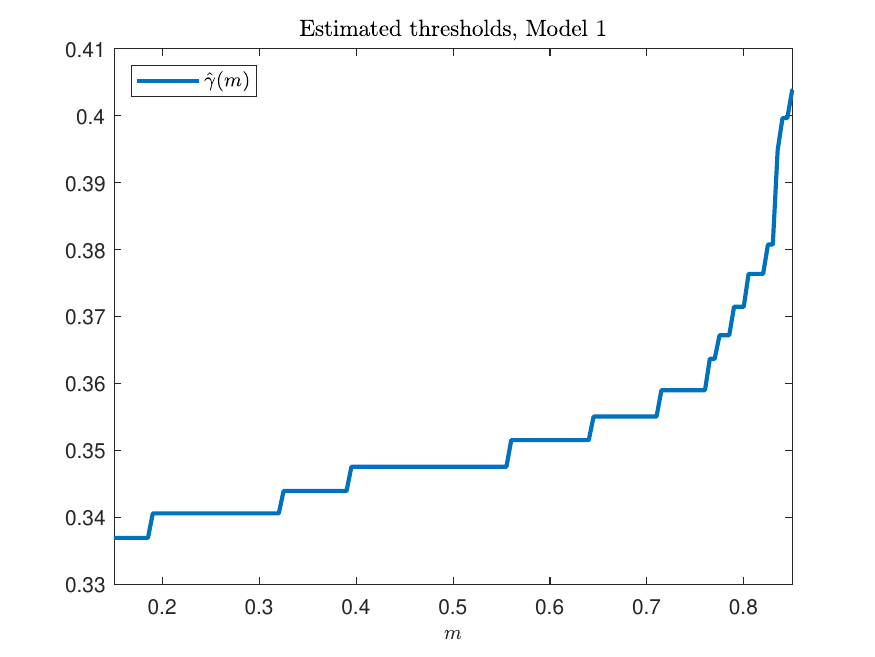}
\caption{This figure depicts the estimated threshold $\hat{\gamma}(m)$ with $m$ referring to the $m$-th quantile of $m_i$ in the range of $m\in [0.15,0.85]$. The left panel refers to models that employ quantiles of \texttt{Financial cost}$_i$  and the right panel to ones that employ \texttt{Fixed assets}$_i$ as an observable fixed-cost metric.} 
\label{fig:emp_threshold1}
\end{figure}

\begin{figure}[H]
\centering
\includegraphics[width=0.45\linewidth]{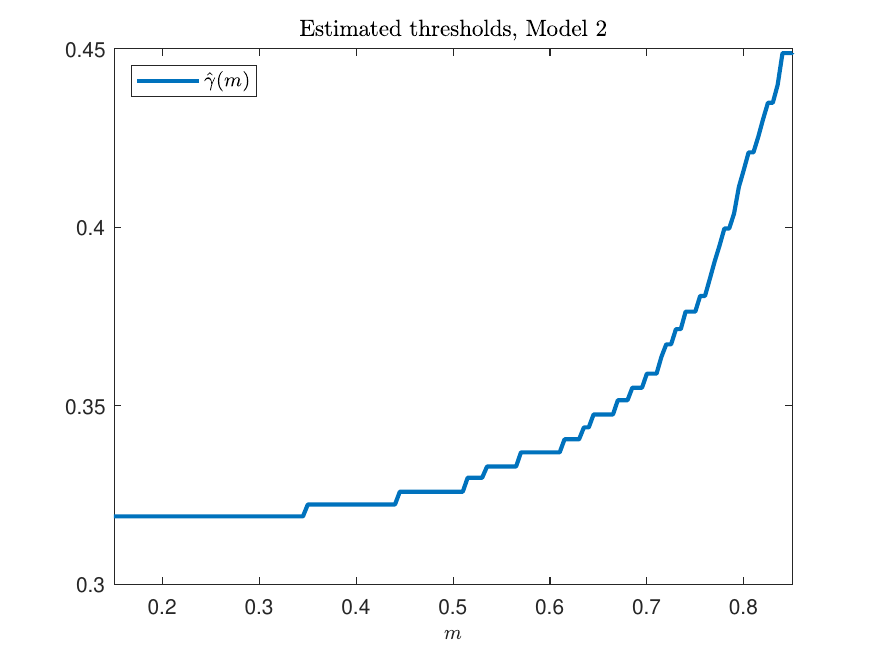}
\includegraphics[width=0.45\linewidth]{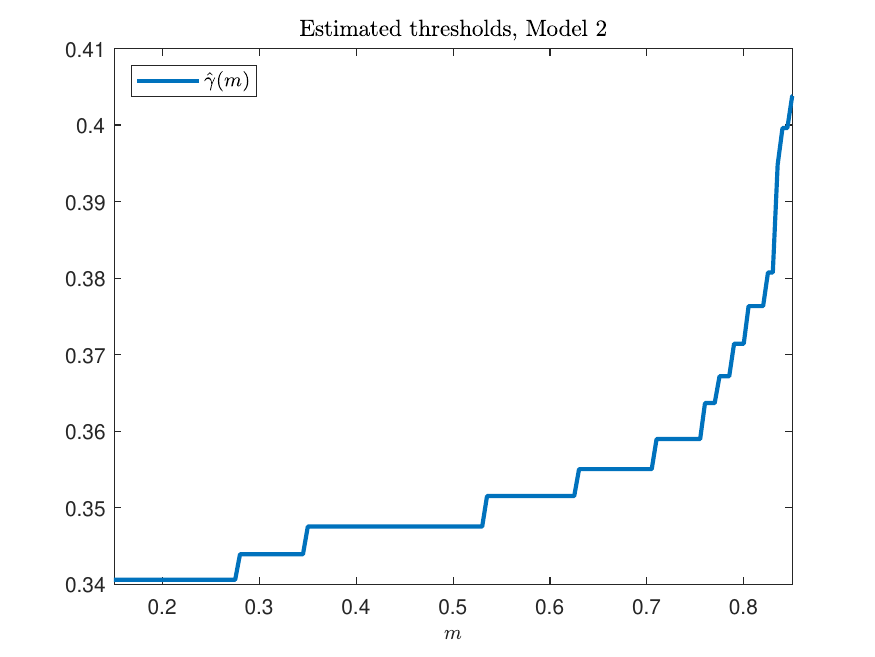}
\caption{This figure depicts the estimated threshold $\hat{\gamma}(m)$ using Model 2, with $m$ referring to the $m$-th quantile of $m_i$ in the range of $m\in [0.15,0.85]$. 
The left panel refers to models that employ quantiles of \texttt{Financial cost}$_i$  and the right panel to ones that employ \texttt{Fixed assets}$_i$ as an observable fixed-cost metric.} 
\label{fig:emp_threshold2}
\end{figure}

\begin{figure}[H]
\centering
\includegraphics[width=0.45\linewidth]{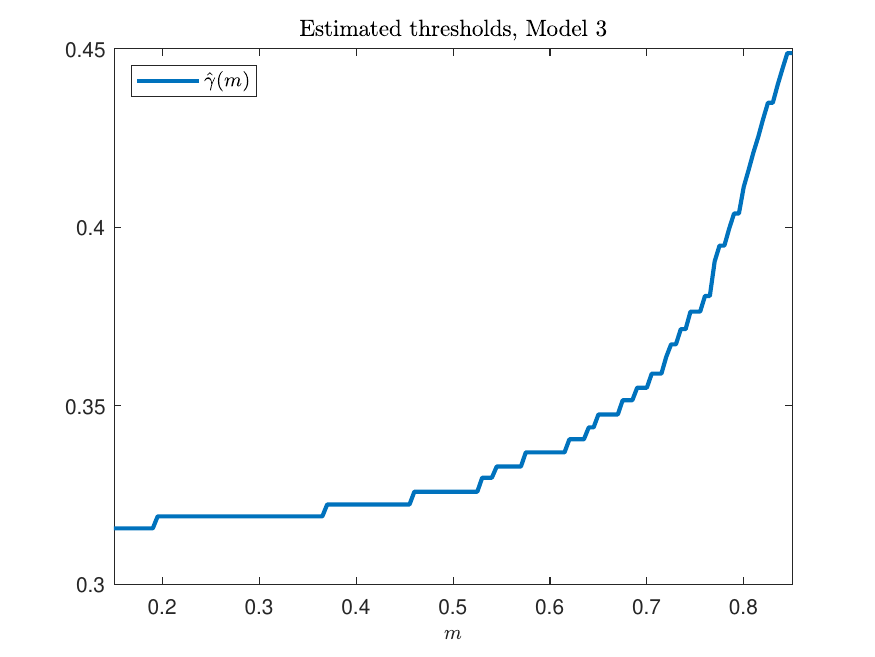}
\includegraphics[width=0.45\linewidth]{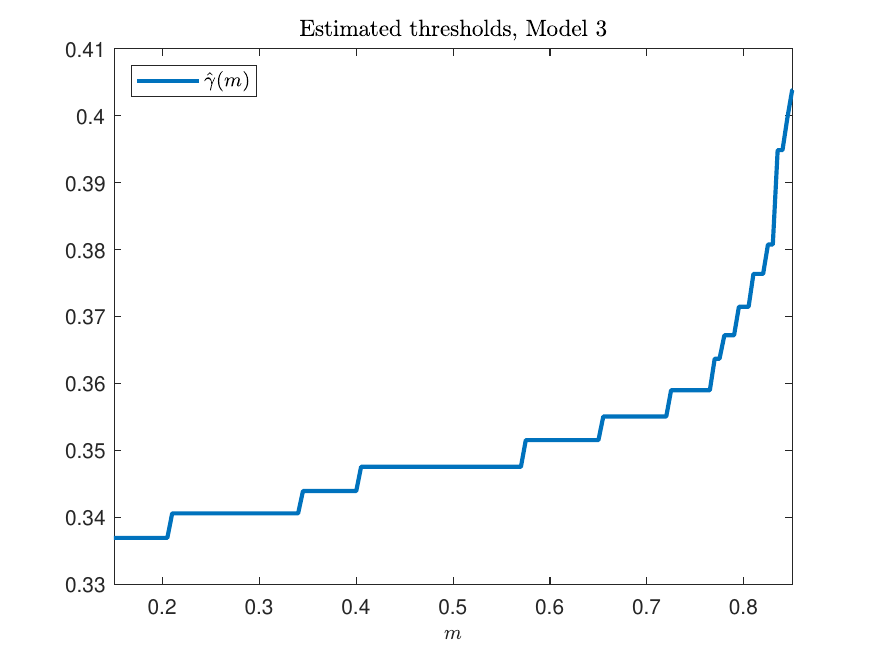}
\caption{This figure depicts the estimated threshold $\hat{\gamma}(m)$ using Model 3, with $m$ referring to the $m$-th quantile of $m_i$ in the range of $m\in [0.15,0.85]$. The left panel refers to models that employ quantiles of \texttt{Financial cost}$_i$  and the right panel to ones that employ \texttt{Fixed assets}$_i$ as an observable fixed-cost metric.} 
\label{fig:emp_threshold3}
\end{figure}

Moreover, we report the estimation results for the coefficients $(\beta_G,\beta_X,\beta_C')'$ with the aforementioned Models 1-3 in Table \ref{tab:regression_results}. 
The interior set $\mathrm{int}(\mathcal{M})$ in \eqref{eq:beta_star} is specified as the interior 98\% quantile range of $m_i$. 
The left panel uses the financial cost as a measure of $m_i$, and the right panel uses total assets instead. 

Let us highlight some interesting findings based on Table \ref{tab:regression_results}. 
First, the coefficient $\hat{\beta}_{X}$ is estimated to be positive and statistically significant, yet smaller than $\hat{\beta}_{G}$, in Model 1. Hence, the profits of firms to the right of the threshold rise less rather than more with an increase in latent productivity as captured by domestic sales. This evidence contradicts the assumptions and associated results in customary Melitz-model applications. When drawing Figure \ref{fig:illus1}, we assumed that the profits of exporters had a lower intercept (higher total fixed market-access costs) and a steeper slope that nonexporters, aligned with the arguments in \cite{melitz2003impact}. The evidence does not support such a mechanism for China, but this is consistent with the aforementioned empirical evidence on the productivity (non-)premium for exporters in China (see \citealp{lu2010exporting,lu2014pure,dai2016}).
In any case, the results based on Model 1 have to be taken with a grain of salt due to the assumption that domestic sales are deterministic and exogenous. We argued above that this assumption appears strong for plausible reasons. 

Second, when relaxing the assumption of exogeneity of $g_i$ and using instrumental-variable estimation through the control-function approach adopted in Models 2 and 3, the evidence of a lack of an exporter productivity or profit premium becomes even stronger. In those models, $\hat{\beta}_{X}$ turns even negative with $\hat{\beta}_{G}$ remaining positive. On the one hand, this aligns with the findings by \cite{lu2010exporting},  \cite{lu2014pure}, \cite{dai2016}, and \cite{gao2022pure} in favor of negative selection. However, it goes beyond those findings in that exporting becomes more likely above if domestic sales and productivity are lower.  

Finally, the considered covariates in $x_i$ are positive and statistically different from zero. Hence, the involvement of foreign companies in China tends to boost profits, and there are positive spillovers from neighboring companies' domestic sales on the average firm $i$. We argued that work in network econometrics suggests that \texttt{Avg.neighb.sales}$_i$ acting as a peer-effect term must be endogenous. The latter is ignored in Models 1 and 2 but not in Model 3. We see that considering the endogeneity of peer effects leads to evidence of considerably larger spillovers than when ignoring it.\footnote{Econometric theory suggests that the peer-effect parameter should be smaller than unity in absolute value, when using weights for other firms that sum up to unity. This is what we find in the two variants of Model 3 as well as in the biased variants of Models 1 and 2. However, the critical threshold for that parameter is not unity here (but lower), due to the normalization of the variables.}

All of the mentioned findings are robust to the alternative considered measurements of observable surrogate variables for fixed costs in $m_i$. The latter can be seen from a comparison of the qualitative findings in the left versus the right block of results in Table \ref{tab:regression_results}.

\begin{table}[ht]
\centering
\begin{threeparttable}

\begin{adjustbox}{width=\textwidth}
\begin{tabular}{lccclccc}
\toprule
\multicolumn{4}{c}{ $m_i=\texttt{Financial cost}_i$}  &  \multicolumn{4}{c}{ $m_i=\texttt{Fixed assets}_i$}  \\
\midrule
\textbf{Variable} & \textbf{Model 1}  &  \textbf{Model 2} & \textbf{Model 3}    & \textbf{Variable}  & \textbf{Model 1} & \textbf{Model 2} & \textbf{Model 3}  \\
\midrule
$(g_i-\gamma_0(m_i))_{-}$     & $0.861^{***}$ & $0.251^{***}$  & $0.229^{***}$  & $(g_i-\gamma_0(m_i))_{-}$     & $1.991^{***}$  & $0.619^{***}$  & $-0.316^{***}$\\
                              & (0.012)       & (0.031)        & (0.033)        &                               &   (0.010)      & (0.042)        & (0.054)  \\
$(g_i-\gamma_0(m_i))_{+}$     & $0.169^{***}$ & $-1.136^{***}$ & $-1.159^{***}$ & $(g_i-\gamma_0(m_i))_{+}$     & $0.006$        & $-1.331^{***}$ & $-2.225^{***}$\\
                              & (0.021)       & (0.069)        & (0.072)        &                               &   (0.005)      & (0.040)        & (0.052) \\
\texttt{Rel.for.cap.}$_i$     & $0.007^{*}$   & $0.006^{*}$    & $0.009^{*}$    &  \texttt{Rel.for.cap.}$_i$    & $0.012$        & $0.012$        & $0.016 $ \\
                              & (0.004)       & (0.004)        & (0.005)        &                               & (0.008)        & (0.008)        & (0.010) \\      
\texttt{Avg.neighb.sales}$_i$& $0.006^{**}$   & $0.008^{***}$  & $0.686^{***}$  & \texttt{Avg.neighb.sales}$_i$  & $0.004^{*}$   & $0.008^{**}$   & $0.923^{***}$ \\
                              & (0.003)       & (0.002)        & (0.021) &                                      & (0.002)        & (0.003)        & (0.024) \\       
                              \hline
$R^2$ & 0.098 & 0.226 & 0.232 & $R^2$ & 0.285 & 0.292 & 0.300 \\
Observations & 187,720 & 187,720  & 187,720 & Observations & 187,720 & 187,720 & 187,720 \\
\bottomrule
\end{tabular}
\end{adjustbox}
\caption{Estimation Results. Dependent variable are profits, $\pi_i$. Running variable are standardized domestic sales $g_i$. Standard errors in parentheses. *$p<0.1$, **$p<0.05$, ***$p<0.01$. Models 1-3 are specified as in \eqref{eq:emp:model1}-\eqref{eq:emp:model3}, respectively. }\label{tab:regression_results}
\end{threeparttable}
\end{table}

In sum, the results attest to theoretically motivated kinks in the functional relationship between profits and domestic sales, with domestic sales exhibiting a monotone, positive mapping with latent firm productivity. Moreover, the kink is not a point but a contour. This will generally be the case whenever not only one running variable (productivity or any monotone reflection of it) but also another variable such as fixed market-access costs is firm-specific and sufficiently independent of the first running variable.

\subsection{Discussion}

The estimation approach and empirical analysis in this paper are inspired by and targeting exporter productivity cutoffs. It is useful to draw the reader's attention to the fact that the only instance where export data had been used up until now is to construct the domestic-sales variable $G_i$ and its normalized counterpart $g_i$. But export data (neither in binary nor continuous form) had not been used in the estimation. However, ex post it is useful to bring exports to the fore in the interest of providing background for differences in the firms' propensity to export on the two sides of the threshold.

Recall the discussion in \cite{melitz2003impact} as well as the illustration in Figure \ref{fig:illus1}. With positive firm selection (i.e., with more productive and larger domestic sellers selecting into exporting), all or most firms to the right of the threshold (above-cutoff productivity and domestic sales) would be exporters and all or most to the left would be nonexporters. Stochastic profit shocks as well as heterogeneous fixed costs (or other fundamentals) generate some fuzziness about this pattern but the qualitative result would persist. If the profit function of domestic sellers would be steeper than the one of exporters and exporting fixed costs would be stochastically dominated by domestic market-access costs, we could have negative selection. Then, exporters would be found more likely below the cutoff productivity and domestic sales than above. 

We commented earlier that the finding of a slope parameter on normalized domestic sales above the threshold of $\beta_X<\beta_G$ with a positive parameter on normalized domestic sales below the threshold, $\beta_G>0$, is consistent with a negative selection of firms into exporting in China. We also said that a negative parameter is in stark contrast with the expectations formed against the backdrop of the customary assumptions adopted by \cite{melitz2003impact} regarding firm selection.

Let us now benchmark the latter finding and conclusion with the probability of exporting above and below the estimated threshold. We do so by generating a $10$-percentile grid on both sides of the threshold contour in $g_i$-$m_i$ space. Then we determine the fraction of exporters in the $10\times 10$ppt.~cells. We illustrate the findings in Figure \ref{fig:emp_export}.

\begin{figure}[H]
\centering
\includegraphics[width=0.45\linewidth]{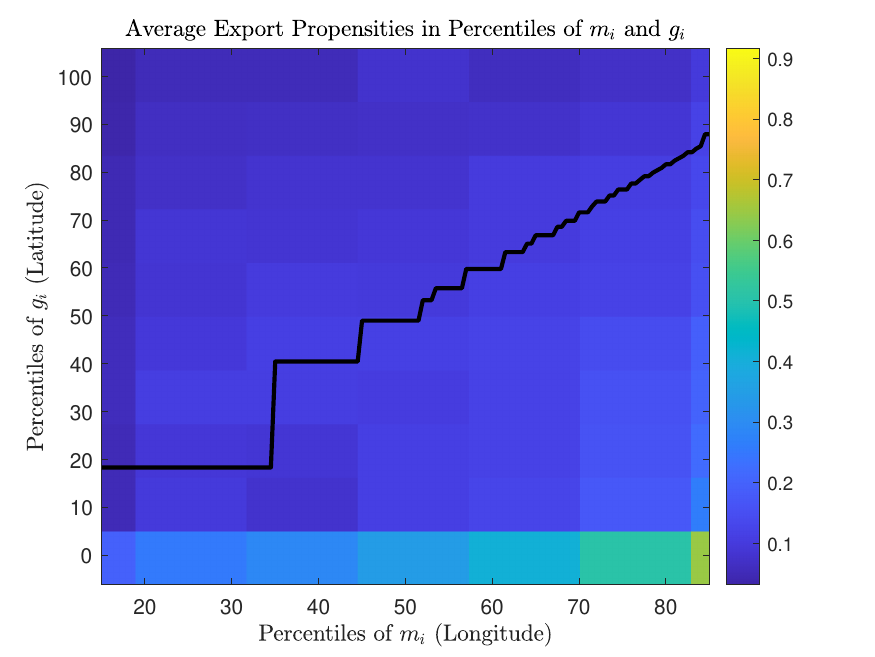}
\includegraphics[width=0.45\linewidth]{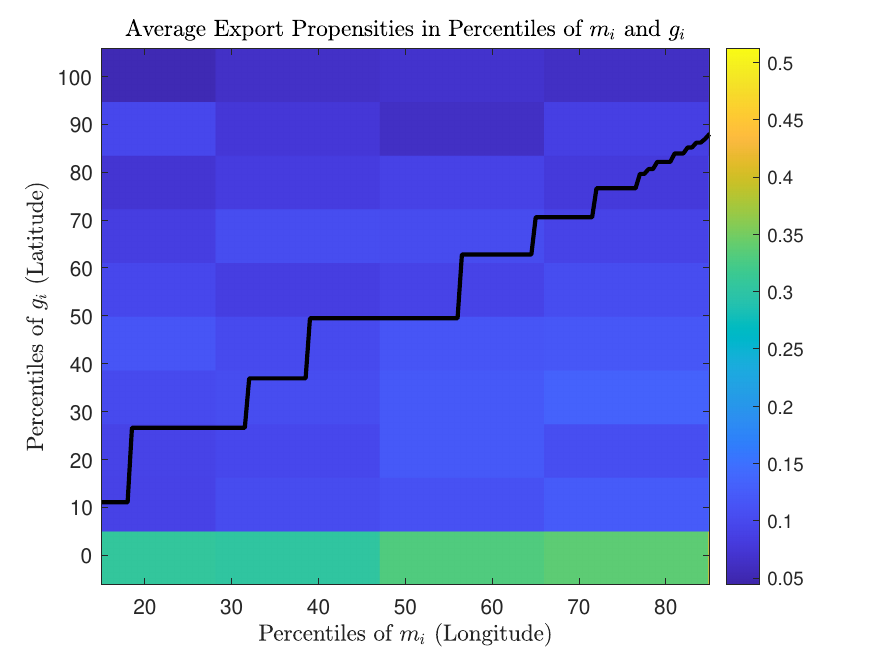}
\caption{This figure depicts the fraction of exporting firms in each tenth (10ppt.) percentile of $m_i$ and $g_i$ and the estimated threshold $\hat{\gamma}(m)$ using Model 3. The horizontal axis is the percentile of $m_i$, scaled from 0 to 100. The left panel refers to models that employ quantiles of \texttt{Financial cost}$_i$  and the right panel to ones that employ \texttt{Fixed assets}$_i$ as an observable fixed-cost metric $m_i$.} 
\label{fig:emp_export}
\end{figure}

Figure \ref{fig:emp_export} contains two panels, depending on the use of surrogate variables for fixed costs in $m_i$. The one on the left employs quantiles of \texttt{Financial cost}$_i$, and the one on the right employs \texttt{Fixed assets}$_i$. We use lighter colors for higher exporter fractions (propensities) and darker ones for lower ones. 
Interpreted through the lens of a standard \citet{melitz2003impact} model, the data would suggest strong positive exporter selection on productivity and a single export cutoff. Our estimates show that this is not the case: heterogeneity in fixed export costs generates a sizable range of effective cutoffs, weakening the link between productivity alone and export participation. Consequently, policy changes affecting marginally the fixed costs, trade costs, or factor prices would not lead to the entry and exit of firms that are marginally small only. Hence, marginal policy changes could have entry and exit implications for relatively large players in the market, which would not be the case under customary assumptions in the model of \cite{melitz2003impact}.

In line with data, essentially for all countries we know of, exporters occur in a wide domain of domestic sales and surrogate variables for fixed costs. Hence, selection is fuzzy rather than sharp. Moreover, we see that exporters are particularly prevalent in the lowest decile of domestic sales, $g_i$. This observation is consistent with earlier-cited work on China. On average, exporters are least prevalent in the lowest-$m_i$ decile and the higher $g_i$ deciles in the left panel and  and in the lower-$m_i$ and higher-$g_i$ deciles in the right panel. The main take-away message is that exporters are less frequent below (to the left of) the cutoff than above (to the right of) the estimated threshold. Hence, exporters emerge less likely in a higher-domestic-sales and lower-fixed-cost environment than in a lower-domestic-sales and higher-fixed-cost environment. This suggests that the assumptions generating positive selection are violated in China. And government or foreign-firm interference may be at the roots of this pattern.

Finally, we note that using data from China may be special in specific regards, and some patterns of behavior may not be simply taken to other data. However, the set of assumptions invoked in this paper is relatively mild. And specific results do not extend to other countries, not only because we use Chinese data, but also because specific links and functions between variables should not be assumed to be the same between countries. In particular, the slopes of the profit functions of exporters and nonexporters depend on the domestic and foreign market potential. The latter depends on market size as well as sales frictions. Neither is the same between countries (China is larger than most markets; the profit tax rates differ between countries; China applies a foreign tariff-rebate policy; etc., all of those aspects affect the slope of profit functions). Moreover, the fixed market-entry costs in a given market differ between firms from different origin countries due to natural frictions (language and cultural barriers, etc.) as well as industrial policy (fixed-cost deductibility rates, subsidies, export promotion, etc.). This means that the intercepts of the profit functions are not the same between countries. Finally, how collateral and other observables affect the creditworthiness and access to external credit varies across countries due to differences in financial frictions, bank liquidity, etc. With this in mind, we clearly would not recommend using estimates for one country, a specific set of sectors, or a specific time period and apply them elsewhere. However, we propose an algorithm that can be applied with any data that permit measurement of the relevant factors. We claim that there is enough evidence across countries on two key aspects: (i) firm productivity distributions largely overlap between exporter and nonexporters, and (ii) any firm-level metric of fixed costs displays variation across firms. 

Specifically, \cite{blum2024abc} demonstrate in their Figure 1 that there is a large overlap in the distributions of (log) productivity between exporters and nonexporters among manufacturing firms in Chile. \cite{blank2022structural} document the same for firms in German manufacturing and services sectors. \cite{helpman2016trade} specifically introduced fixed-cost heterogeneity in their model, because the data on Brazilian firms they used did not corroborate scalar-valued cutoff-productivity thresholds. A large overlap in productivity exporters and nonexoporters means that productivity is unlikely to be the single source of selection. In addition, any accounting data set would suggest that there is a large variance in fixed assets, financial expenses, and other variables between all sellers and, where identifiable, between domestic sellers and exporters. Using such observable variables to parameterize fixed market-entry costs means that the variation in fixed costs within firm type is large. One could document that in China and elsewhere, there is an overlap in observable fixed-cost measures -- ranging from asset tangibility, fixed assets, liquid assets, to financial expenses, and even interest rates paid -- between exporters and domestic sellers, as is the case for productivity. At this generic level, firms in China are not special, and the arguments extend beyond the specific empirical application.

\section{Conclusions}
This paper introduces a novel framework to estimate productivity thresholds in international trade firm models. Specifically, the estimation framework allows for more than one dimension of firm heterogeneity, as is commonly encountered in data but rarely addressed in economic theory (see \citealp{arkolakis2010market}, \citealp{helpman2016trade}, or \citealp{adao2020aggregate}, for exceptions): productivity and fixed costs. In principle, the framework would be capable of considering even more than two such dimensions. In any setting where productivity (or another heterogeneous variable) is not the single firm-indexed 
parameter, and firms make choices about discrete alternatives such as exporting or not, the critical value of the firm parameter such as the exporter cutoff productivity is not scalar but set-valued.

Moreover, the (scaled) firm productivity is proportional to several observable outcomes such as prices or domestic sales. However, fixed costs or other firm-specific attributes are not clearly parametrically related to observable surrogate variables of them. For that reason, we develop the estimator as to combine parametric and nonparametric forms in the threshold-generation function. We prove the large sample properties of the estimator and document its suitability in finite samples by simulation. 

Using a large cross section of non-exporting and exporting firms in China, we document that productivity-threshold estimates are indeed not constant but vary systematically with fixed costs. As theory predicts, firms which have higher fixed market-entry costs require a higher productivity to survive and make profits. The variation of fixed costs across firms is large and this translates into a large variation of the set of productivity thresholds. This translation is not mechanical, as we permit the relationship to be nonparametric. We consider this evidence as particularly strong and more powerful than what one could obtain from parametric estimates such those based on interaction effects.

Our results have broader implications for international trade theory and policy. In particular, trade-cost or market-size shocks do not trigger effects and responses primarily and only in the left tail of the firm-size distribution. Depending on the covariance structure between productivity and fixed costs among firms, trade-cost shocks will affect firms in a large support area of their size distribution. For instance, tariff increases as envisaged by the Trump II administration have the potential of triggering significant exit of medium-sized and even large firms from exporting, if productivity and fixed costs are positively correlated. We find the latter to be the case with the data at hand. On the contrary, trade-liberalization events might lead to an entry of large enterprises in a market which they did not find attractive before, because of high firm-specific fixed entry costs. 

We used the consideration of the choice between exporting and non-exporting in a context of firm-specific productivity and fixed costs as a guiding example. As argued in the paper, the benchmarking of some operating profits or other variable gains against some fixed costs is a fairly generic problem in international economics and industrial organization alike. Many contexts there involve discrete choices of companies and lead to subsets of companies choosing one alternative over another.
Apart from export-market entry closely-related problems include the adoption of new (single or multiple) products, the set up of (domestic or foreign) affiliates, the integration of suppliers or buyers, etc. The search for kinks values of fundamental running variables in the functional relationships with such problems is of a generic importance in those contexts. A natural stochastic approach to identify these cutoff or breakpoints in running variables is the kink-regression framework. If more than a single variable  (such as productivity) has threshold-generating power but there is at least one more determinant, the threshold does not degenerate to a point but is an at least two-dimensional contour. Whenever the surrogate variables of fixed costs of adopting an alternative strategy (exporting, multi-plant production, multinational production, outsourcing or offshoring, integration of suppliers or buyers, etc.) are measured with some variation across firms, one is faced with exactly this problem. And in absence of a tight theoretical guidance on the mapping of surrogate variables to fixed costs or other firm-specific parameters, one would then wish to permit a nonparametric relationship between observables and theoretical fixed costs underlying firms' choices. The present paper provides an estimation framework for exactly such problems.

Future research could extend the proposed methodology to dynamic settings or investigate its applicability to other economic contexts, where firm heterogeneity plays a critical role. Examples are settings where firms choose multinomially to produce a set of products (multi-product firms) or firms choose multinomially to set up plants (multi-plant or multinational firms).
In addition, recent work by \citet{lee2021factor} and \citet{yu2021threshold} develops methods for parametric threshold models. We view the development of parametric regression-kink estimators and specification tests as a challenging and important of possible future research.

\newpage
\bibliographystyle{apalike}
\bibliography{references}
\newpage
\newpage
\appendix
\section{Details of the Asymptotic Properties}\label{sec:detail}
We now study the asymptotic behavior of the proposed estimator. 
Recall the kink regression model \eqref{eq:kink2}, which is 
\begin{equation}
\pi _{i}=\beta _{G0}\left( g_{i}-\gamma_{0}\left( m_{i}\right) \right)_{-}+\beta_{X0}\left( g_{i}-\gamma_{0}\left( m_{i}\right) \right)_{+}+x_{i}^{\prime }\beta_{C0}+u_{i}\text{,}  \label{eq:model}
\end{equation}
where the observations $\{(\pi_{i},x_{i}',g_{i},m_{i})'\in \mathbb{R}^{1+\dim(x)+1+1}; i=1,\dots,n\}$ are i.i.d. 
We will use the subscript $0$ to denote the true parameter values. 
The threshold function $\gamma_{0}:\mathbb{R}\rightarrow \Gamma$ and the regression coefficients $\beta_{0}=(\beta_{G0},\beta_{X0},\beta_{C0}^{\prime })^{\prime }\in \mathbb{R}^{2+\dim(x)}$ are unknown parameters of interest. 
We let $\mathcal{G}\subset \mathbb{R}$ and $\mathcal{M}\subset \mathbb{R} $ denote the supports of $g_{i}$ and $m_{i}$, respectively.
We also let the space of $\gamma_{0}\left( m\right) $ for any $m$ be a compact set $\Gamma \subsetneq \mathcal{G}$.

It is convenient to write
\begin{equation*}
z_{i}\left( \gamma \left( m\right) \right) =\left( 
\begin{array}{c}
\left( g_{i}-\gamma \left( m\right) \right)_{-} \\ 
\left( g_{i}-\gamma \left( m\right) \right)_{+} \\ 
x_{i}%
\end{array}%
\right), 
\end{equation*}%
so that \eqref{eq:kink2} can be written as $\pi _{i}=\beta_{0}^{\prime}z_{i}\left( \gamma_{0}\left( m_{i}\right) \right) +u_{i}$.

We first establish the identification, which requires the following conditions.

\begin{assumption}[Identification]\label{ass:id}
\begin{enumerate}
    \item $\mathbb{E}\left[ u_{i}|x_{i},g_{i},m_{i}\right] =0$ almost surely.
    \item $\mathbb{E}\left[ 1\left[ g_{i}\in A\right] z_{i}\left( \gamma \left( m\right) \right) z_{i}\left( \gamma \left( m\right) \right) ^{\prime }|m_{i}=m\right] >0$\textit{\ for any }$\left( \gamma ,m\right) \in \Gamma \times \mathcal{M}$ and any set $A\in \mathcal{G} $ such that $\mathbb{P}\left( g_{i}\in A\right) >0$.
    \item $\Gamma$ is a compact subset of $\mathcal{G}$. Also, $g_{i}$ is continuously distributed with a conditional density $f_{g|m}(s|t)$ satisfying $0<C_{1}<f_{g|m}(s|t)<C_{2}<\infty $  for all $\left( s,t\right) \in \mathcal{G}\times \mathcal{M}$ and some constants $C_{1}$ and $C_{2}$.
\end{enumerate}
\end{assumption}

Assumption \ref{ass:id} is relatively mild. 
Assumption 1.1 rules out endogeneity, a restriction we will relax in Section \ref{sec:endo}. 
Assumption 1.2 imposes a full-rank condition, which is standard for identifying $\beta_{0}$. 
Assumption 1.3 ensures that the threshold does not lie at the boundary of the support of $q_{i}$ for any $m\in \mathcal{M}$.
This condition is essential for identification and is commonly imposed in the threshold regression literature \citep[e.g.,][]{hansen2017regression}.

Under Assumption \ref{ass:id}, the following theorem identifies the semiparametric threshold regression model.

\begin{theorem}\label{thm:id}
Under Assumption \ref{ass:id}, $(\beta_{0}',\gamma_{0}(m))'$ is the unique minimizer of 
\begin{equation*}
  L\left( \beta ,\gamma \left( m\right) \right) =\mathbb{E}[( \pi_{i}-\beta^{\prime }z_{i}\left( \gamma \left( m\right) \right))^{2}|m_{i}=m] 
\end{equation*}
for each $m\in \mathcal{M}$.
\end{theorem}

We impose the following assumptions for asymptotic analysis.

\begin{assumption}[Asymptotics]\label{ass:a}
\begin{enumerate}
    \item $\mathbb{E}\left[ \pi_{i}^{4}|m_i = m\right] <\infty $ and $\mathbb{E}\left[ \left\vert \left\vert z_{i}\left( \gamma \right)\right\vert \right\vert ^{4} | m_i = m\right] <\infty $ uniformly over $\gamma$ and $m$. The parameter space of $\beta$ is compact. 
    \item The kernel function $K\left( \cdot \right) $ satisfies that $K(\cdot)>0$ and is continuously differentiable with uniformly bounded derivatives, $K(t) =K( -t) $ for $t\in \mathbb{R}$ and $\lim_{\left\vert t\right\vert \rightarrow \infty}t^{2}K\left(t\right) =0$. 
    \item The bandwidth $b_{n}$ satisfies that $b_{n}\rightarrow 0$ and $nb_n\rightarrow\infty$ as $n\rightarrow\infty$.
    \item The function $L\left( \beta ,\gamma \right) $ is twice continuously differentiable with uniformly bounded derivatives over $\theta=\left( \beta' ,\gamma \right)' $. The densities $f_m(m)$ and $f_{g|m}(g|m)$ are also twice continuously differentiable with uniformly bounded derivatives and satisfies that $f_{m}(m) >0$ for all $m$. 
    \item The kink function $\gamma_0(\cdot)$ is twice continuously differentiable with uniformly bounded derivatives.
\end{enumerate}
\end{assumption}

Assumption \ref{ass:a} is also mild and aligns with standard practice in the literature. 
Assumption \ref{ass:a}.1 imposes bounded moments and assumes a compact parameter space. 
Assumption \ref{ass:a}.2 places regularity conditions on the kernel function. 
Assumption \ref{ass:a}.3 governs the choice of the bandwidth and accommodates the commonly used rule-of-thumb choice $b_n = O(n^{-1/5})$. 
Assumption \ref{ass:a}.4 requires smoothness of the limiting criterion function and boundedness of the density. 
Finally, Assumption \ref{ass:a}.5 imposes smoothness on the true kink function.

For any query point $m\in \textrm{int}(\mathcal{M})$, denote $\theta_0(m) = (\beta_0',\gamma_0(m))'$ and $\hat{\theta}(m) = (\hat{\beta}',\hat{\gamma}(m))'$. 
The following theorem establishes the consistency of $\hat{\theta}(m)$ at given $m$.
\begin{theorem}\label{thm:consistency}
Under Assumptions \ref{ass:id} and \ref{ass:a}, for any $m\in \mathrm{int}(\mathcal{M})$, $\hat{\theta}(m)-\theta _{0}(m)=o_{p}\left( 1\right)$.
\end{theorem}

To derive the asymptotic distribution, we introduce the following notation.
Denote
\begin{eqnarray*}
D_{i}\left( \theta \right)  &=&-\frac{\partial }{\partial \theta }\left( \pi_{i}-\beta ^{\prime }z_{i}\left( \gamma \right) \right)  \\
&=&\binom{z_{i}\left( \gamma \right) }{-\beta_{G}1\left[g_{i}<\gamma \right] -\beta_{X}1\left[ g_{i}>\gamma \right] }
\end{eqnarray*}
and $D_{i}=D_{i}\left( \theta _{0}\right) $. 
Define
\begin{eqnarray*}
Q(m) &=&\mathbb{E}\left[ D_{i}D_{i}^{\prime }|m_{i}=m \right]  \\
V(m) &=&\mathbb{E}\left[ u_{i}^{2}D_{i}D_{i}^{\prime }|m_{i}=m\right]\times \left( \frac{\int K^{2}\left(u\right)du}{f_{m}\left( m\right) }\right).
\end{eqnarray*}%
The following theorem establishes the asymptotic normality of our estimator locally at any interior point $m$.
Denote $\gamma_0^{(k)}(m)=\partial^k\gamma_0(m)/\partial m^k$ as the $k$-th derivative of the function $\gamma_0(m)$, and similarly for $f^{(k)}(m)$. 

\begin{theorem}
\label{thm:normal}Under Assumptions \ref{ass:id} and \ref{ass:a}, for any $m\in \mathrm{int}(\mathcal{M})$, 
\begin{equation*}
\sqrt{nb_{n}}\left\{ \hat{\theta}\left( m\right) -\theta _{0}\left( m\right)-b_{n}^{2}Q\left( m\right) ^{-1}B\left( m\right) \right\} \overset{d}{\rightarrow} \mathcal{N}\left( 0,Q\left( m\right) ^{-1}V\left( m\right) Q\left( m\right)^{-1}\right) ,
\end{equation*}
where 
\begin{eqnarray*}
B\left( m\right) &=&\int u^{2}K\left( u\right) du\times \left\{ Q\left( m\right) \left( \gamma _{0}^{(1)}\left( m\right) f_{m}^{(1)}\left( m\right) +\frac{1}{2}\gamma _{0}^{(2)}(m) \right)
\right. \\
&&\left. +\frac{\beta _{X0}+\beta _{G0}}{2}\mathbb{E}\left[ \left. D_{i} \right\vert g_{i}=\gamma_{0}\left( m\right),m_{i}=m\right] f_{g|m}\left( \gamma _{0}\left( m\right) |m\right) \gamma_{0}^{(1)}(m) ^{2}\right\} .
\end{eqnarray*}
\end{theorem}

Theorem \ref{thm:normal} derives the pointwise asymptotic normality of $\hat{\theta}(m) = (\hat{\beta}',\hat{\gamma}(m))'$. 
The proof of this theorem is non-trivial because of the kink structure, and accordingly the bias term is substantially more complicated than in the standard case. 
In particular, the first term in $B(m)$ corresponds to the asymptotic bias in the standard kernel regression \citep[e.g.,][Chapter 2]{li2007nonparametric}.  
The second term characterizes the effect of the kink structure. 
Interestingly, this term becomes zero if the threshold function reduces to a constant, i.e., $\gamma_0^{(1)}(m)=0$.
To our knowledge, this result is new to the literature. 

Note that, unlike $\gamma_0=\gamma_0(m)$ for $m\in \mathcal{M}$, the slope parameter $\beta_0$ is a global constant that does not vary across $m$. 
Therefore, we can improve the estimation of $\beta_0$ to achieve the root-$n$ rate. 
To this end, we construct the leave-one-out estimator $\hat{\gamma}_{-i}(m_i)$ without using the $i$-th observation and perform the OLS estimator
\begin{equation}\label{eq:beta_star}
    \hat{\beta}^{*} = \left( \sum_{i=1}^n z_i(\hat{\gamma}_{-i}(m_i))z_i(\hat{\gamma}_{-i}(m_i))'1[m_i\in\mathrm{int}(\mathcal{M})] \right)^{-1} 
                      \left( \sum_{i=1}^n z_i(\hat{\gamma}_{-i}(m_i))y_i 1[m_i\in\mathrm{int}(\mathcal{M})]  \right),
\end{equation}
where $1[\cdot]$ denotes the indicator function. 
We note that the threshold estimator $\hat{\gamma}(m)$ might perform poorly at boundary values of $m$. 
Therefore, we only use the interior values $m\in\mathrm{int}(\mathcal{M})$ by employing the indicator function. 
In practice, we set $\mathrm{int}(\mathcal{M})$ as the middle 98\% percent of $m_i$ as in Section \ref{sec:empirical}.

The following theorem derives the asymptotic property of $\hat{\beta}^{*}$. 
Denote $\iota = (0,\dots,0,1)'$ and $d_i = -\beta_{G0}1[g_{i}<\gamma_0(m_i) ] -\beta_{X0} 1[g_{i}>\gamma_0(m_i)]$
\begin{theorem}
\label{thm:rootn} Suppose Assumptions \ref{ass:id} and \ref{ass:a} hold and $b_n = O(n^{-\lambda})$ for some $\lambda\in(1/4,1/3)$, then
\begin{equation*}
\sqrt{n}\left( \hat{\beta}^{\ast }-\beta _{0}\right) =\frac{1}{\sqrt{n}}\sum_{j=1}^{n}\psi _{j}+o_{p}\left( 1\right) ,
\end{equation*}
where
\begin{equation*}
\psi_{j}=\left( z_{j}\left( \gamma _{0}\left( m_{j}\right) \right) -\delta\left( m_{j}\right) \right) u_{j}1\left[ m_{j}\in \mathrm{int}\left( \mathcal{M}\right) \right] 
\end{equation*}
and
\begin{equation*}
\delta \left( m_{j}\right) = \left. \mathbb{E}\left[ z_{i}\left( \gamma _{0}\left(m_{i}\right) \right) d_{i}\iota ^{\prime }Q\left( m_{i}\right) ^{-1}|m_{i}=m\right]\right\vert_{m=m_{j}}D_{j}(\theta_0(m_j)).
\end{equation*}
\end{theorem}

Theorem \ref{thm:rootn} derives the influence function for the estimator in \eqref{eq:beta_star}. 
This result is novel and warrants careful discussion.

First, the influence function reflects the two-step nature of the estimator: the first step involves the nonparametric estimation of $\hat{\gamma}_{-i}(m_i)$, while the second step is a linear regression.
Two-step semiparametric estimation under smoothness assumptions has been extensively studied in the econometrics literature \citep[e.g.,][]{newey1994asymptotic,newey1994large}. However, our setting presents additional challenges due to the presence of a kink structure, where the threshold function $\gamma_0(\cdot)$ enters non-differentiably.

Second, the term $\delta(m_j)$ captures how the first-step nonparametric estimation impacts the second-step regression. 
In classical semiparametric estimation frameworks \citep[e.g.,][]{ichimura1993semiparametric}, such terms typically reflect orthogonality conditions that ensure robustness to a first-step estimation error.
\citet{ichimura2022influence} recently develop a direct method to compute these effects. 
In contrast, our estimator does not satisfy the orthogonality condition: $\hat{\gamma}(m)$ and $\hat{\beta}$ are not orthogonal at any $m$ as shown in Theorem \ref{thm:normal}. 
As a result, $\delta(m_j)$ involves the term $\iota'Q(m_j)^{-1}D_j(\theta_0(m_j))$, which arises from the leading term in the last component of $\hat{\theta}(m)-\theta_0(m)$. 

Third, the estimator $\hat{\gamma}(\cdot)$ inherits the standard bandwidth selection requirements of kernel regression. 
The optimal choice takes the rate $b_n = O(n^{-1/5})$. 
However, achieving root-$n$ consistency for $\hat\beta^{*}$ necessitates undersmoothing to control the bias from the first-step estimation. 
As a consequence, this step essentially reinforces Theorem \ref{thm:consistency} by establishing the uniform convergence of $\hat{\gamma}(m)$ over $m$.

Specifically, we require $b_n = O(n^{-\lambda})$ for some $\lambda\in(1/4,1/3)$ in line with similar conditions in the semiparametric estimation literature \citep[e.g.,][]{klein1993efficient}. 
In our simulation studies, we adopt $b_n = n^{-3.5}$, which satisfies this requirement.
In the empirical application, we report results using both a rule-of-thumb bandwidth ($b_n = n^{-1/5}$) and the undersmoothing choice ($b_n = n^{-3.5}$) to assess robustness. 

Finally, the asymptotic variance of $\hat\beta^{*}$ is given by $\mathbb{E}[\psi_j\psi_j']$, which in principle can be estimated via a plug-in approach.
However, the complexity of the $\delta(m_j)$ term renders this approach difficult to implement. 
We therefore adopt the wild bootstrap procedure \citep[e.g.,][]{mammen1993bootstrap,davidson2008wild}.
Specifically, given the estimates $\hat{\beta}^{*}$ and $\hat{\gamma}_{-i}(m_i)$, we compute the residuals
\[
\hat{u}_i = \pi_i - z_i(\hat{\gamma}_{-i}(m_i))'\hat{\beta}^{*}.
\]
We then draw i.i.d.~Rademacher multipliers $\varepsilon_i \in \{-1,1\}$, with equal probability and construct the bootstrap outcome:
\[
\pi_i^* = z_i(\hat{\gamma}_{-i}(m_i))'\hat{\beta}^{*} + \hat{u}_i \varepsilon_i.
\]
Keeping the covariates $z_i(\hat{\gamma}_{-i}(m_i))$ fixed, we regress $\pi_i^*$ on $z_i(\hat{\gamma}_{-i}(m_i))$ to re-estimate $\hat{\beta}^{*}$. 
Repeating this over many bootstrap samples allows us to compute standard errors and construct confidence intervals based on the empirical distribution of the bootstrapped estimates.
We note that, although standard bootstrap methods are known to fail for inference on threshold parameters (\citealp{yu2014bootstrap}), this issue does not arise here because we use the bootstrap only for inference on the regression coefficient $\beta$, which is root-$n$ estimable and asymptotically linear.

\section{Monte Carlo Simulation}\label{sec:simu}
In this section, we evaluate the finite sample performance of the proposed estimator. 
Subsection \ref{sec:simu:exo} studies the case with exogenous $g_i$, and Subsection \ref{sec:simu:end} addresses the case of an endogenous $g_i$. 

\subsection{Exogenous $g_i$}\label{sec:simu:exo}

We start with the model without endogeneity and generate random draws from the following data-generating process
$$
\pi _{i}=\beta _{G0}\left( g_{i}-\gamma_{0}\left( m_{i}\right) \right)_{-}+\beta_{X0}\left( g_{i}-\gamma_{0}\left( m_{i}\right) \right)_{+}+\beta_{C10}+x_{i}\beta_{C20}+0.5 u_{i},
$$
where $(g_i,m_i,x_i,u_i)\sim \mathcal{N}(0,I_4)$ and $I_k$ denotes the $k\times k$ identity matrix.  
The threshold function $\gamma_0(m) = (m+1)^3/8$, which is increasing for $m>-1$. 
As shown in Figures \ref{fig:simu_threshold} and \ref{fig:simu_threshold_iv} below, this function mimics our empirical findings in Figures \ref{fig:emp_threshold1}-\ref{fig:emp_threshold3}.
The true coefficients are $\beta_{G0} \in \{1,2,3,4\}$ and $\beta_{X0}=\beta_{C10}=\beta_{C20}=0$. 
Since we choose $\beta_{X0}=0$ in the simulations, $\beta_{G0}$ parametrizes compactly the difference in the regression slope to the left versus the right of the kink.
The sample size is $n \in \{100,200,500\}$, and the results are based on 1,000 simulation draws.
We use the standard Gaussian kernel and the bandwidth of $b_n = n^{-1/3.5}$, which satisfies the condition in Theorem \ref{thm:rootn}. 

\begin{table}[ht]
\centering
\begin{tabular}{lccccc}
\hline
&  $n/\ \beta_{G0} $    & 1    & 2    & 3    & 4    \\ \hline
\textbf{SNR}    &      & 4.90 & 19.58 & 44.07 & 78.34 \\ \hline
\textbf{Bias}   & 100  & 0.00 & 0.00 & 0.00 & 0.00 \\
                & 200  & 0.00 & 0.00 & 0.00 & 0.00 \\
                & 500  & 0.00 & 0.00 & 0.00 & 0.00 \\ \hline
\textbf{RMSE}   & 100  & 0.10 & 0.11 & 0.11 & 0.10 \\
                & 200  & 0.07 & 0.07 & 0.07 & 0.07 \\
                & 500  & 0.04 & 0.05 & 0.04 & 0.04 \\ \hline
\end{tabular}
\caption{Signal-to-noise ratio (SNR), bias and RMSE of the estimates of $\beta_{G0}$. See the main text for details of the data-generating process.}
\label{tab:simu_theta1}
\end{table}

Table \ref{tab:simu_theta1} presents the bias and the root mean-squared error (RMSE) of the estimate $\hat{\beta}_{G0}^{*}$ as in \eqref{eq:beta_star}. 
We see that the bias is very small for any sample size, and the RMSE is decreasing in the sample size as expected. 
Also, the size of the kink effect $\beta_{G0}$ is relatively immaterial for these results. 
To quantify variance of the signal relative to the noise in a way that is comparable across values of $\beta_{G0}$, we also report the signal-to-noise ratio (SNR)
\[
\mathrm{SNR}\;\equiv\;\frac{Var\!\left(\beta_{G0}(g_i-\gamma_0(m_i))_-\right)}{Var\!\left(0.5\,u_i\right)}\;=\; \frac{\beta_{G0}^2\,Var\!\left((g_i-\gamma_0(m_i))_-\right)}{0.25}.
\]
We compute $Var\!\left((g_i-\gamma_0(m_i))_-\right)$ under the stated DGP $(g_i,m_i)\sim \mathcal{N}(0,I_2)$ and $\gamma_0(m)=(m+1)^3/8$. 

\begin{table}[ht]
\centering
\resizebox{\textwidth}{!}{
\begin{tabular}{lcccccccccccccc}
\hline
         &      & \multicolumn{4}{c}{$m=0$} & \multicolumn{4}{c}{$m=0.25$} & \multicolumn{4}{c}{$m=0.5$} \\
         & $n/\ \beta_{G0}$    & 1    & 2    & 3    & 4    & 1    & 2    & 3    & 4    & 1    & 2    & 3    & 4    \\ \hline
\textbf{SNR}    &      & 1.57 & 6.27 & 14.11 & 25.08 & 1.77 & 7.08 & 15.93 & 28.32  & 2.08 & 8.31 & 18.70 & 33.25 \\ \hline      
\textbf{Bias}   & 100  & -0.01 & 0.01 & 0.01 & 0.02 & -0.02 & -0.01 &  0.00 &  0.00 & -0.13 & -0.05 & -0.04  & -0.05 \\
                & 200  & 0.01 &  0.02 & 0.01 & 0.01 & -0.01 &  0.01 &  0.00 &  0.01 & -0.03 & -0.02 & -0.02  & -0.02 \\
                & 500  & 0.01 &  0.01 & 0.00 & 0.01 &  0.01 &  0.00 &  0.00 &  0.00 & -0.01 & -0.01 & -0.01  & -0.01 \\ \hline
\textbf{RMSE}   & 100  & 0.53  & 0.25  & 0.14  & 0.11  & 0.59  & 0.27  & 0.18  & 0.13  & 0.72  & 0.32  & 0.24  & 0.18 \\
                & 200  & 0.39  & 0.14  & 0.09  & 0.07  & 0.44  & 0.17  & 0.11  & 0.09  & 0.50  & 0.22  & 0.14  & 0.12 \\
                & 500  & 0.22  & 0.09  & 0.06  & 0.05  & 0.25  & 0.11  & 0.07  & 0.05  & 0.30  & 0.12  & 0.08  & 0.06 \\ \hline
\end{tabular}
}
\caption{Signal-to-noise ratio (SNR), bias and RMSE of the estimates of $\gamma_0(m)$ at different $m$. See the main text for details of the data-generating process.}
\label{tab:simu_gamma}
\end{table}

Table \ref{tab:simu_gamma} summarizes the bias and RMSE of $\hat{\gamma}(m)$ at $m\in\{0,0.5,1\}$. 
The SNR is also calculated at each $m$ respectively.
Both bias and RMSE of $\hat{\gamma}(m)$ are decreasing in $n$, as expected. 
Additionally, the bias and RMSE both decrease in $\beta_{G0}$, which is because a larger threshold effect $\beta_{G0}$ makes the threshold $\gamma_{0}$ easier identifiable. 
This finding is similar to the one for the customary threshold-regression model \citep[e.g.,][]{hansen2000sample}.

Figure \ref{fig:simu_threshold} depicts the average of 1,000 estimates $\hat{\gamma}(\cdot)$ against the true function $\gamma_0(\cdot)$. 
Figure \ref{fig:simu_density} depicts the histogram based on the 1,000 estimates $\hat{\beta}_{G0}$, which suggests a good normal approximation. 
Both figures are based on a sample size of $n=1,000$ and use the same data-generating process as in the previous tables with $\beta_{G0}=4$.

\begin{figure}[H]
\centering
\includegraphics[width=0.7\linewidth]{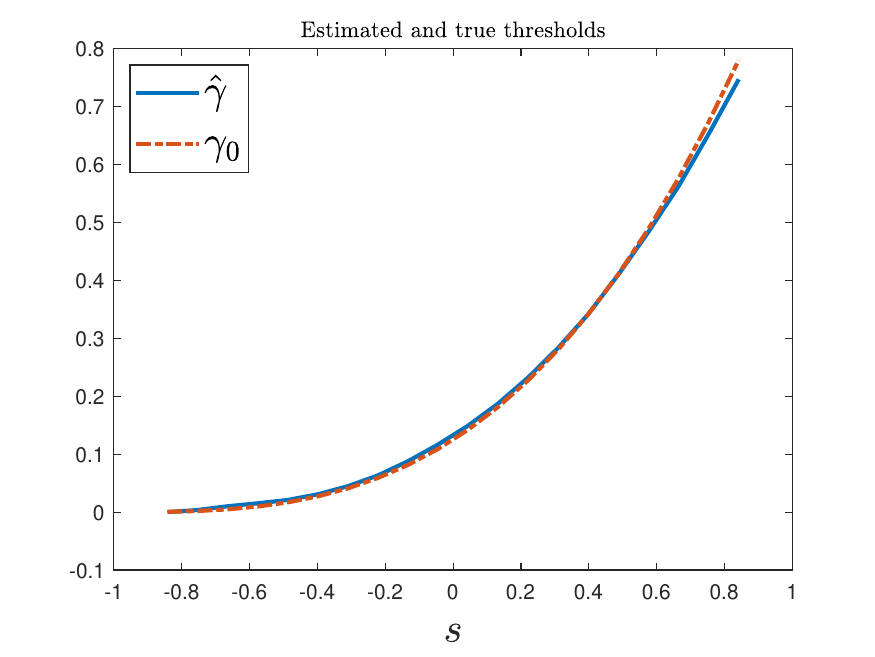}
\caption{This figure depicts the average of $\hat{\gamma}(\cdot)$ over 1,000 simulation draws against the true value $\gamma_0(\cdot)$. } 
\label{fig:simu_threshold}
\end{figure}
\begin{figure}[H]
\centering
\includegraphics[width=0.7\linewidth]{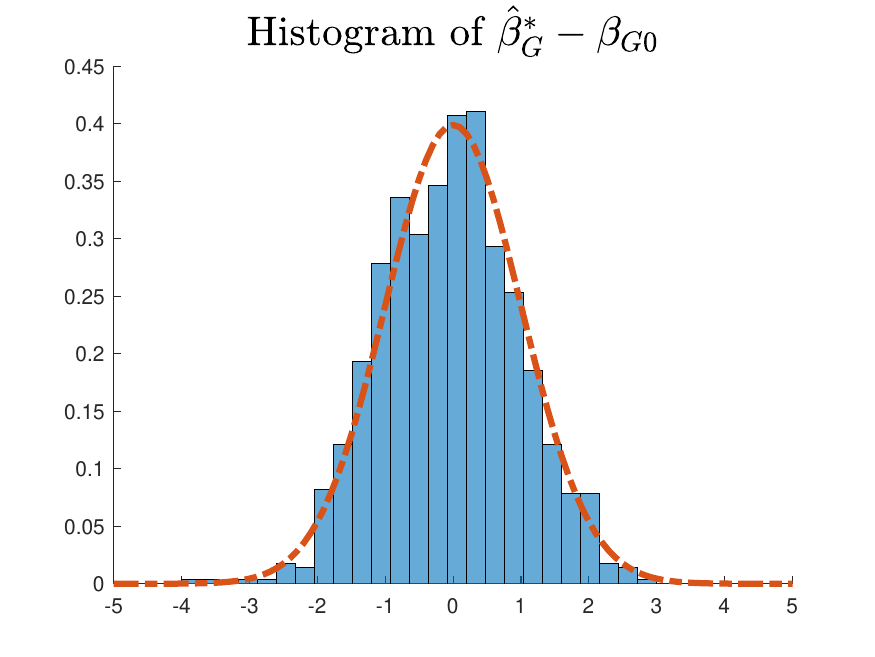}
\caption{This figure depicts the histogram plot using 1,000 simulation draws of $\hat{\beta}_{G}^{*}-\beta_{G0}$. The red dash line is the standard normal density.}
\label{fig:simu_density}
\end{figure}

\subsection{Endogenous $g_i$}\label{sec:simu:end}
Our second study examines the estimator when $g_i$ is endogenous. We generate random draws from the following data-generating-process:
$$
\pi _{i}=\beta _{G0}\left( g_{i}-\gamma_{0}\left( m_{i}\right) \right)_{-}+\beta_{X0}\left( g_{i}-\gamma_{0}\left( m_{i}\right) \right)_{+}+\beta_{C10}+x_{i}\beta_{C20}+ u_{i},
$$
where $u_i = 0.5\epsilon_i + 0.5v_i$, $g_i = v_i + w_i$, and $(m_i,x_i,\epsilon_i,v_i,w_i)\sim \mathcal{N}(0,I_5)$. 
This model captures the dependence structure that $g_i$ is endogenous and $w_i$ is the valid instrument that is correlated with $g_i$ and uncorrelated with the error term $u_i$. For this setting, we employ the estimation procedure outlined in Section \ref{sec:endo}. 
All other parameters are the same as in the simulation setup in the previous subsection. 

Figures \ref{fig:simu_threshold_iv} and \ref{fig:simu_density_iv} respectively plot the average estimated threshold function $\hat{\gamma}(\cdot)$ across simulations and the histogram of the estimated coefficient $\hat{\beta}_G^{*}$ centered at its true value $\beta_{G0}$. 
Tables \ref{tab:simu_gamma_iv} and \ref{tab:simu_delta_iv} respectively summarize the bias and the RMSE of $\hat{\beta}_G^{*}$ and $\hat{\gamma}(m)$. We report the signal-to-noise ratio
\[
\mathrm{SNR}\;\equiv\;\frac{Var\!\bigl(\beta_{G0}(g_i-\gamma_0(m_i))_{-}\bigr)}{Var(u_i)}=\frac{\beta_{G0}^2\,Var\!\bigl((g_i-\gamma_0(m_i))_{-}\bigr)}{0.5},
\]
where the denominator uses $Var(u_i)=0.5$ under the stated DGP.
The findings are similar to the previous ones without endogeneity. Hence, the proposed estimators permit identification of the threshold (or cutoff) value $\gamma_0$ of the running variable of interest, $g_i$, as a nonparametric function of $m_i$, irrespective of whether $g_i$ is exogenous or not. Clearly, the latter requires the availability of relevant and adequate instruments, but this is standard.

\begin{figure}[H]
\centering
\includegraphics[width=0.7\linewidth]{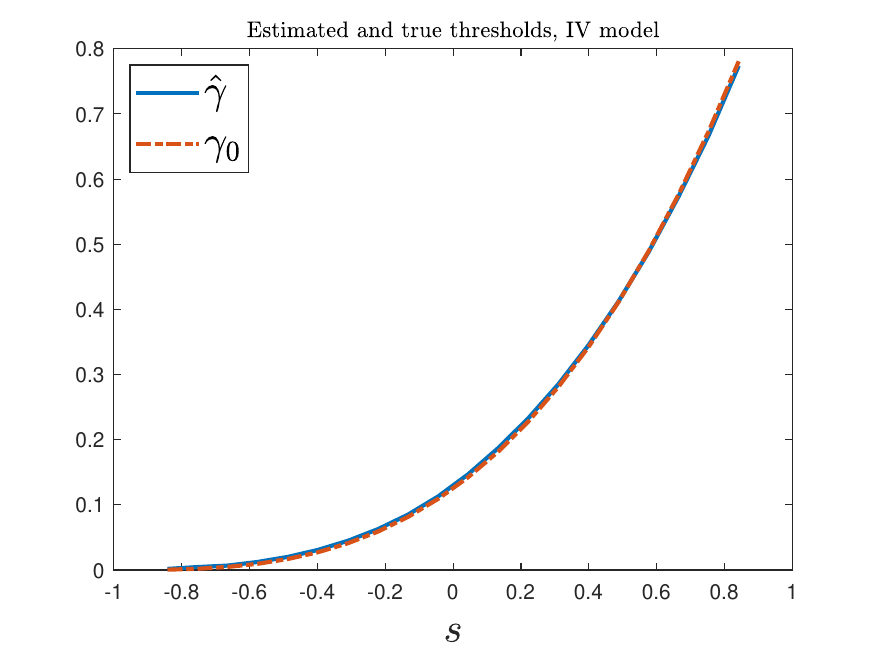}
\caption{This figure depicts the average of $\hat{\gamma}(\cdot)$ over 1,000 simulation draws against the true value $\gamma_0(\cdot)$ with endogenous $g_i$. See the main text for details of the data-generating process. } 
\label{fig:simu_threshold_iv}
\end{figure}
\begin{figure}[H]
\centering
\includegraphics[width=0.7\linewidth]{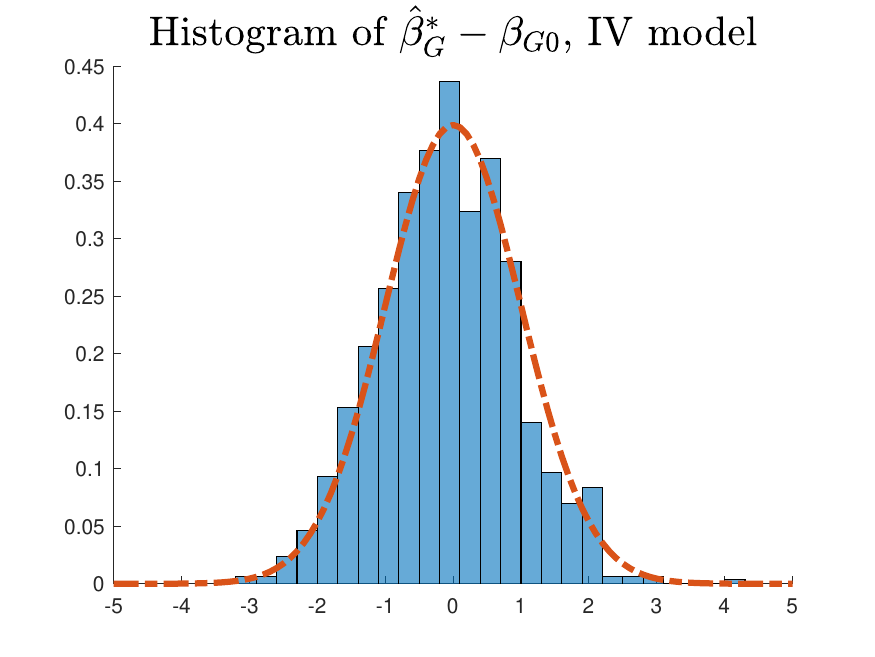}
\caption{This figure depicts the histogram plot using 1,000 simulation draws of $\hat{\beta}_{G}^{*}-\beta_{G0}$ with endogenous $g_i$. The red dash line is the standard normal density. See the main text for details of the data-generating process.}
\label{fig:simu_density_iv}
\end{figure}

\begin{table}[ht]
\centering
\begin{tabular}{lccccc}
\hline
         &  $n/\ \beta_{G0} $    & 1    & 2    & 3    & 4    \\ \hline
SNR             &  & 3.17 & 12.68 & 28.53 & 50.72 \\ \hline
\textbf{Bias}   & 100  & 0.00 & 0.00 & 0.00 & 0.00 \\
                & 200  & 0.00 & 0.00 & 0.00 & 0.00 \\
                & 500  & 0.00 & 0.00 & 0.00 & 0.00 \\ \hline
\textbf{RMSE}   & 100  & 0.10 & 0.10 & 0.10 & 0.10 \\
                & 200  & 0.07 & 0.07 & 0.07 & 0.07 \\
                & 500  & 0.04 & 0.04 & 0.04 & 0.04 \\ \hline
\end{tabular}
\caption{Bias and RMSE of the estimates of $\beta_{G0}$ with endogenous $g_i$. See the main text for details of the data-generating process.}
\label{tab:simu_gamma_iv}
\end{table}

\begin{table}[ht]
\centering
\resizebox{\textwidth}{!}{
\begin{tabular}{lcccccccccccccc}
\hline
         &      & \multicolumn{4}{c}{$m=0$} & \multicolumn{4}{c}{$m=0.25$} & \multicolumn{4}{c}{$m=0.5$} \\
         & $n/\ \beta_{G0}$    & 1    & 2    & 3    & 4    & 1    & 2    & 3    & 4    & 1    & 2    & 3    & 4    \\ \hline
\textbf{SNR}    &      & 1.51 & 6.03 & 13.56 & 24.11 & 1.65 & 6.59 & 14.83 & 26.37  & 1.86 & 7.46 & 16.78 & 29.82 \\ \hline 
\textbf{Bias}   & 100  & -0.01 &  0.00 & 0.01 & 0.00 & -0.06 &  0.00 &  0.00 &  0.00 & -0.08 & -0.03 &  0.00  & -0.01 \\
                & 200  & -0.02 &  0.02 & 0.01 & 0.01 & -0.01 &  0.00 &  0.00 &  0.01 & -0.06 &  0.00 & -0.01  & 0.00 \\
                & 500  &  0.00 &  0.01 & 0.00 & 0.00 &  0.00 &  0.00 &  0.00 &  0.00 & -0.01 &  0.00 &  0.00  & 0.00 \\ \hline
\textbf{RMSE}   & 100  & 0.58  & 0.26  & 0.16  & 0.11  & 0.61  & 0.27  & 0.18  & 0.13  & 0.74  & 0.33 & 0.22  & 0.17 \\
                & 200  & 0.43  & 0.17  & 0.11  & 0.08  & 0.45  & 0.17  & 0.11  & 0.09  & 0.46  & 0.21 & 0.14  & 0.11 \\
                & 500  & 0.26  & 0.11  & 0.07  & 0.05  & 0.26  & 0.11  & 0.08  & 0.06  & 0.29  & 0.13 & 0.08  & 0.06 \\ \hline
\end{tabular}
}
\caption{Bias and RMSE of the estimates of $\gamma_0(m)$ at different $m$ with endogenous $g_i$. See the main text for details of the data-generating process.}
\label{tab:simu_delta_iv}
\end{table}

\section{Additional Empirical Result}\label{sec:additional}
Section \ref{sec:empirical} presents the results using the rule-of-thumb bandwidth $b_n = n^{-1/5}$. 
For robustness, we now repeat the analysis using the undersmoothing bandwidth $b_n = n^{-1/3.5}$, with all other setups the unchanged. 

We first present the kernel estimates of the threshold $\hat{\gamma}(\cdot)$ based on the undersmoothing kernel in Figures \ref{fig:emp_threshold1_undersmooth}-\ref{fig:emp_threshold3_undersmooth}, corresponding to Figures \ref{fig:emp_threshold1}-\ref{fig:emp_threshold3}, respectively. 
As expected, the undersmoothing bandwidth leads to less smooth estimates in all models. 
But the shape of the threshold contour remains qualitatively robust. 

\begin{figure}[H]
\centering
\includegraphics[width=0.45\linewidth]{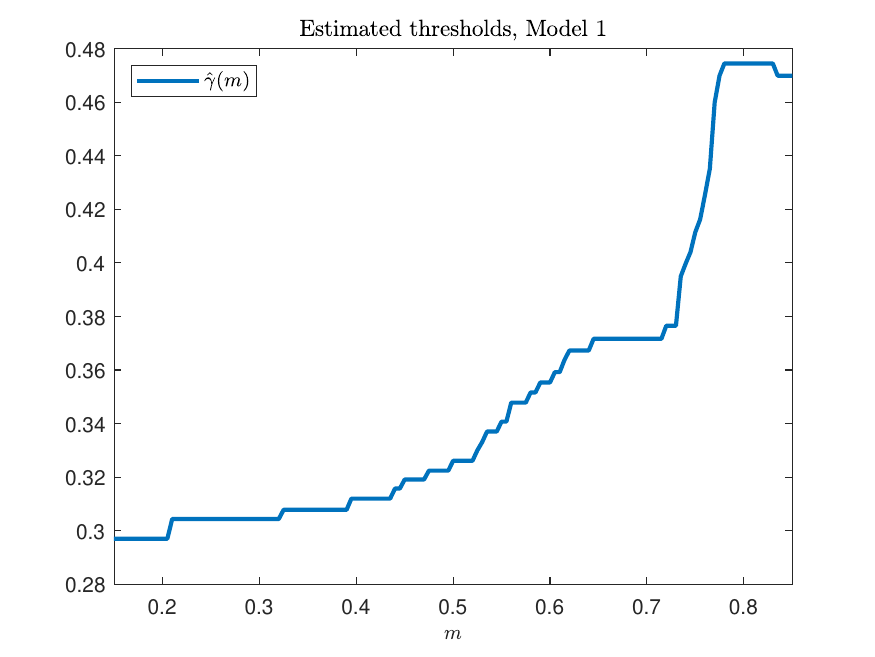}
\includegraphics[width=0.45\linewidth]{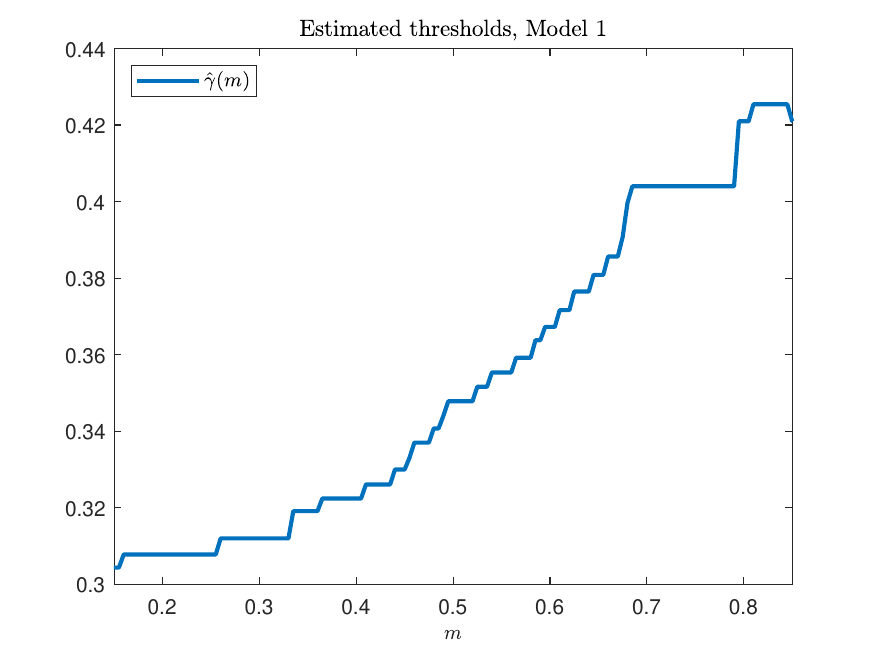}
\caption{This figure depicts the estimated threshold $\hat{\gamma}(m)$ and the undersmoothing bandwidth $b_n=n^{-1/3.5}$, with $m$ referring to the $m$-th quantile of $m_i$ in the range of $m\in [0.15,0.85]$. The left panel refers to models that employ quantiles of \texttt{Financial cost}$_i$  and the right panel to ones that employ \texttt{Fixed assets}$_i$ as an observable fixed-cost metric.} 
\label{fig:emp_threshold1_undersmooth}
\end{figure}

\begin{figure}[H]
\centering
\includegraphics[width=0.45\linewidth]{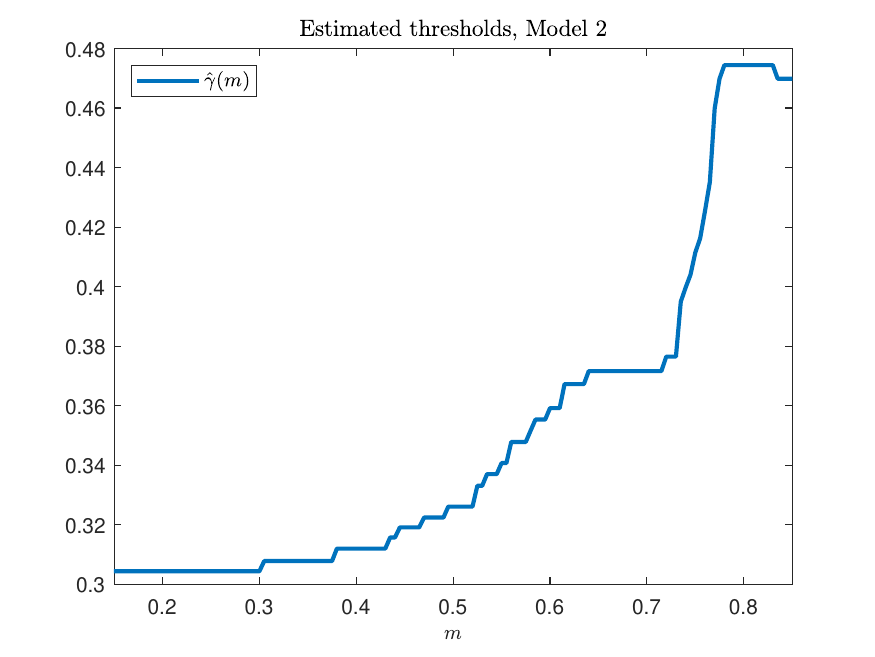}
\includegraphics[width=0.45\linewidth]{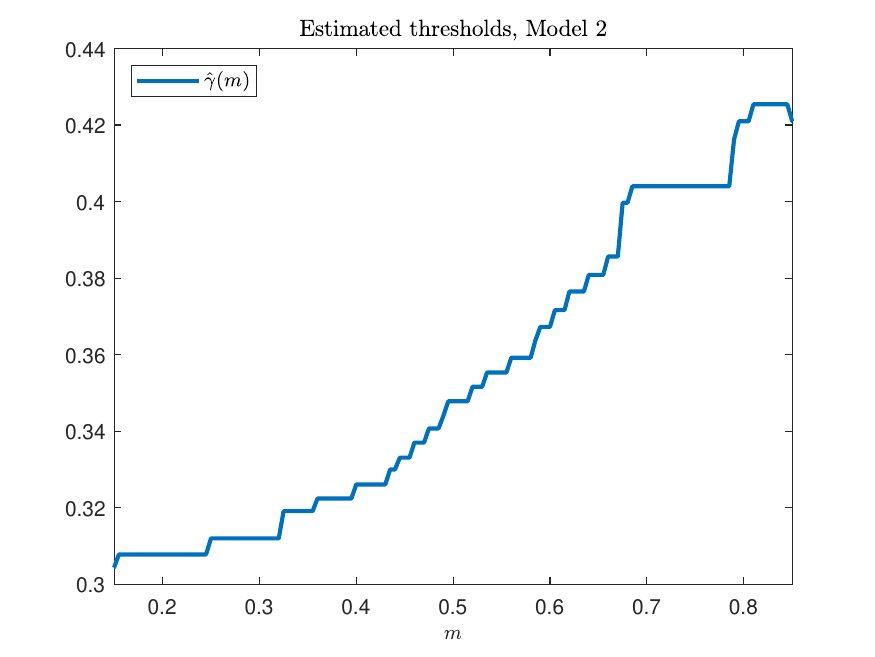}
\caption{This figure depicts the estimated threshold $\hat{\gamma}(m)$ using Model 2 and the undersmoothing bandwidth $b_n=n^{-1/3.5}$, with $m$ referring to the $m$-th quantile of $m_i$ in the range of $m\in [0.15,0.85]$. 
The left panel refers to models that employ quantiles of \texttt{Financial cost}$_i$  and the right panel to ones that employ \texttt{Fixed assets}$_i$ as an observable fixed-cost metric.} 
\label{fig:emp_threshold2_undersmooth}
\end{figure}

\begin{figure}[H]
\centering
\includegraphics[width=0.45\linewidth]{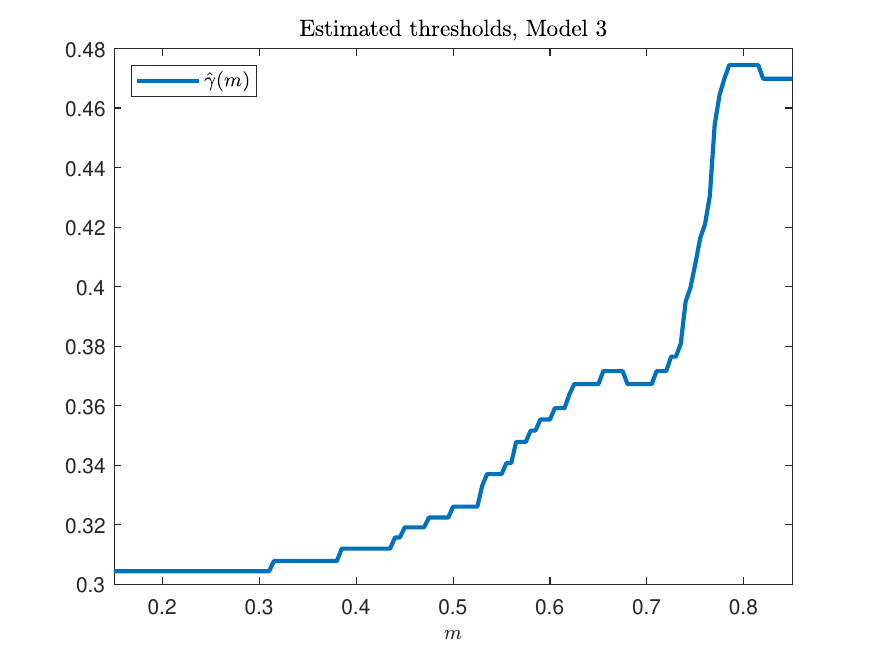}
\includegraphics[width=0.45\linewidth]{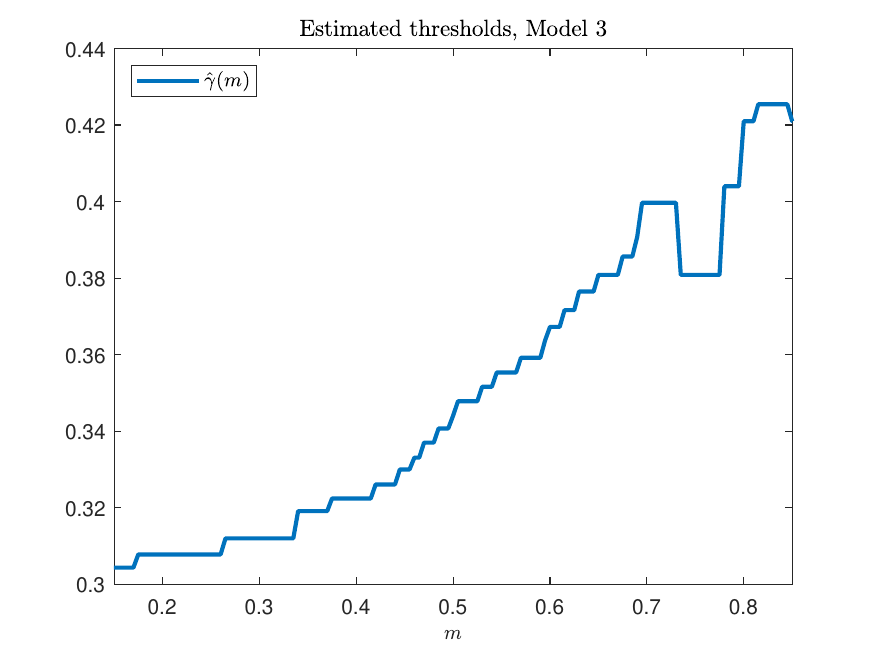}
\caption{This figure depicts the estimated threshold $\hat{\gamma}(m)$ using Model 3 and the undersmoothing bandwidth $b_n=n^{-1/3.5}$, with $m$ referring to the $m$-th quantile of $m_i$ in the range of $m\in [0.15,0.85]$. The left panel refers to models that employ quantiles of \texttt{Financial cost}$_i$  and the right panel to ones that employ \texttt{Fixed assets}$_i$ as an observable fixed-cost metric.} 
\label{fig:emp_threshold3_undersmooth}
\end{figure}

Now we present the regression results using the undersmoothing bandwidth in Table \ref{tab:regression_results_undersmooth}. 
The numbers are very similar to those in Table \ref{tab:regression_results}, suggesting that our findings are robust to the bandwidth choice.

\begin{table}[ht]
\centering
\begin{threeparttable}

\begin{adjustbox}{width=\textwidth}
\begin{tabular}{lccclccc}
\toprule
\multicolumn{4}{c}{ $m_i=\texttt{Financial cost}_i$}  &  \multicolumn{4}{c}{ $m_i=\texttt{Fixed assets}_i$}  \\
\midrule
\textbf{Variable} & \textbf{Model 1}  &  \textbf{Model 2} & \textbf{Model 3}    & \textbf{Variable}  & \textbf{Model 1} & \textbf{Model 2} & \textbf{Model 3}  \\
\midrule
$(g_i-\gamma_0(m_i))_{-}$     & $0.891^{***}$ & $0.239^{***}$  & $0.215^{***}$  & $(g_i-\gamma_0(m_i))_{-}$     & $1.964^{***}$  & $0.026 $       & $-0.603^{***}$\\
                              & (0.011)       & (0.036)        & (0.036)        &                               &   (0.010)      & (0.032)        & (0.038)  \\
$(g_i-\gamma_0(m_i))_{+}$     & $0.164^{***}$ & $-1.159^{***}$ & $-1.184^{***}$ & $(g_i-\gamma_0(m_i))_{+}$     & $0.011^{**}$   & $-1.871^{***}$ & $-2.468^{***}$\\
                              & (0.020)       & (0.074)        & (0.074)        &                               &   (0.005)      & (0.031)        & (0.035) \\
\texttt{Rel.for.cap.}$_i$     & $0.007^{*}$   & $0.006^{*}$    & $0.009^{*}$    &  \texttt{Rel.for.cap.}$_i$    & $0.011$        & $0.011$        & $0.017^{*} $ \\
                              & (0.004)       & (0.004)        & (0.004)        &                               & (0.008)        & (0.008)        & (0.010) \\      
\texttt{Avg.neighb.sales}$_i$& $0.006^{**}$   & $0.008^{***}$  & $0.691^{***}$  & \texttt{Avg.neighb.sales}$_i$  & $0.005^{**}$  & $0.010^{***}$  & $0.989^{***}$ \\
                              & (0.002)       & (0.003)        & (0.021) &                                      & (0.002)        & (0.004)        & (0.019) \\       
                              \hline
$R^2$ & 0.101 & 0.228 & 0.235 & $R^2$ & 0.271 & 0.291 & 0.302 \\
Observations & 187,720 & 187,720  & 187,720 & Observations & 187,720 & 187,720 & 187,720 \\
\bottomrule
\end{tabular}
\end{adjustbox}
\caption{Estimation Results using the undersmoothing bandwidth $b_n=n^{-1/3.5}$. Dependent variable are profits, $\pi_i$. Running variable are standardized domestic sales $g_i$. Standard errors in parentheses. *$p<0.1$, **$p<0.05$, ***$p<0.01$. Models 1-3 are specified as in \eqref{eq:emp:model1}-\eqref{eq:emp:model3}, respectively. }\label{tab:regression_results_undersmooth}
\end{threeparttable}
\end{table}

Finally, we repeat the heatmap analysis. 
Figure \ref{fig:emp_export_undersmooth} appears less smooth than Figure \ref{fig:emp_export} due to the undersmoothing bandwidth. 
But the shape of the threshold contour is again robust. 

\begin{figure}[H]
\centering
\includegraphics[width=0.45\linewidth]{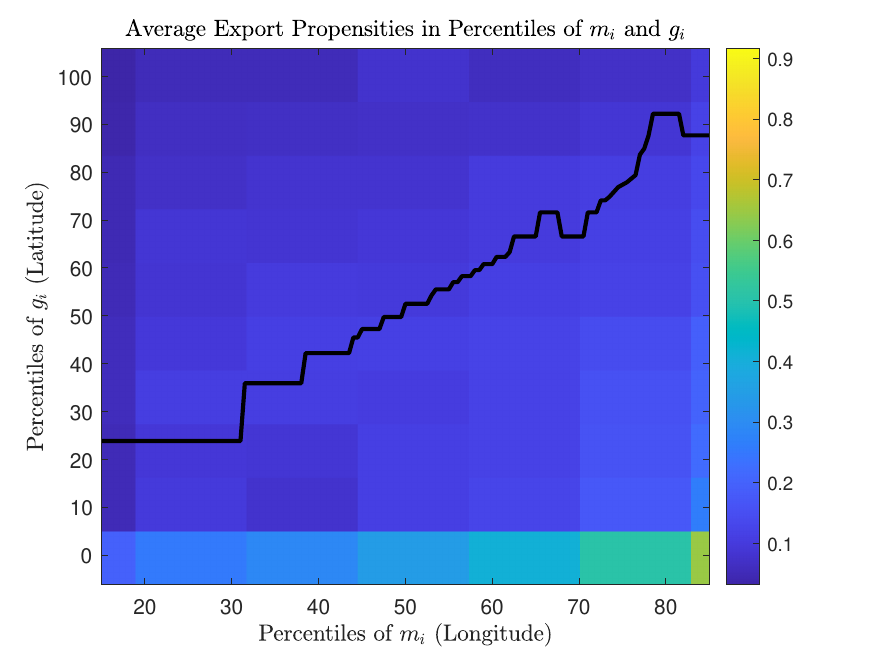}
\includegraphics[width=0.45\linewidth]{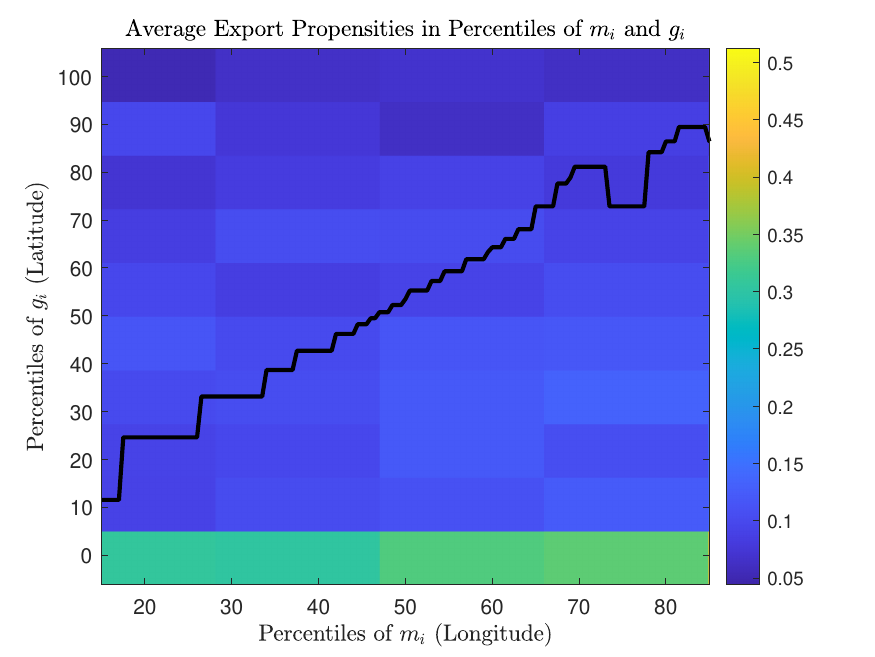}
\caption{This figure depicts the fraction of exporting firms in each tenth (10ppt.) percentile of $m_i$ and $g_i$ and the estimated threshold $\hat{\gamma}(m)$ using Model 3 and the undersmoothing bandwidth $b_n=n^{-1/3.5}$. The horizontal axis is the percentile of $m_i$, scaled from 0 to 100. The left panel refers to models that employ quantiles of \texttt{Financial cost}$_i$  and the right panel to ones that employ \texttt{Fixed assets}$_i$ as an observable fixed-cost metric $m_i$.} 
\label{fig:emp_export_undersmooth}
\end{figure}

\section{Proofs}\label{sec:proof}
\begin{proof}[Proof of Theorem \ref{thm:id}]
Since the proof is pointwise in $m$, we use the short-handed notation $\gamma =\gamma \left( m\right) $ for simplicity. 
For any constant $\gamma \in \Gamma $ with given $m\in \mathcal{M}$, we define a conditional $L_{2}$-loss as%
\begin{eqnarray*}
&&R(\beta ,\gamma |m) \\
&=&L\left( \beta ,\gamma \right) -L\left( \beta_{0},\gamma _{0}\left( m\right) \right)  \\
&=&\mathbb{E}\left[ \left. \left\{ \pi _{i}-\beta ^{\prime }z_{i}\left(\gamma  \right) \right\} ^{2}\right\vert m_{i}=m\right] -\mathbb{E}\left[ \left. \left\{ \pi _{i}-\beta _{0}^{\prime }z_{i}\left(\gamma _{0}\left( m\right) \right) \right\} ^{2}\right\vert m_{i}=m\right] \\
&=&\mathbb{E}\left[ \left\{ \beta _{0}^{\prime }z_{i}\left( \gamma_{0}\left( m\right) \right) -\beta ^{\prime }z_{i}\left( \gamma \right) \right\} ^{2}|m_{i}=m\right]  \\
&=&\mathbb{E}\left[ \left. \left\{ 
\begin{array}{l}
x_{i}^{\prime }\left( \beta_{C0}-\beta_C \right) +\left( \beta _{G0}-\beta_{G}\right) \left( g_{i}-\gamma _{0}\left( m_{i}\right) \right) _{-} \\ 
+\left( \beta _{X0}-\beta _{X}\right) \left( g_{i}-\gamma _{0}\left(m_{i}\right) \right) _{+} \\ 
+\beta _{G}\left( \left( g_{i}-\gamma \right) _{-}-\left(g_{i}-\gamma _{0}\left( m\right) \right) _{-}\right)  \\ 
+\beta _{X}\left( \left( g_{i}-\gamma \right) _{+}-\left(g_{i}-\gamma _{0}\left( m\right) \right) _{+}\right)
\end{array}%
\right\} ^{2}\right\vert m_{i}=m\right] \text{,}
\end{eqnarray*}
which is continuous in $(\beta _{G},\beta _{X},\beta_C^{\prime },\gamma)^{\prime }$. 
By construction, $R(\beta_{G},\beta _{X},\beta_C^{\prime},\gamma |m)\geq 0$ for any $(\beta _{G},\beta _{X},\beta_C^{\prime },\gamma )^{\prime }\in \mathbb{R}^{2+\dim (x)}\times \Gamma $ and $R(\beta_{G0},\beta _{X0},\beta_{C0}^{\prime },\gamma _{0}\left( m\right)|m)=0$. 
Hence, it suffices to show that $R(\beta _{G},\beta _{X},\beta_C^{\prime },\gamma |m)>0$ for any vector $(\beta _{G},\beta_{X},\beta_C^{\prime },\gamma )\neq (\beta _{G0},\beta _{X0},\beta_{C0}^{\prime},\gamma _{0}\left( m\right) )^{\prime }$ given $m\in \mathcal{M}$. 
To this end, we split the event $(\beta_{G},\beta_{X},\beta_C^{\prime },\gamma)\neq (\beta_{G0},\beta_{X0},\beta_{C0}^{\prime },\gamma _{0}\left( m\right) )^{\prime }$ into two disjoint cases: 
(i) $(\beta_{G},\beta_{X},\beta_C^{\prime })^{\prime }\neq (\beta _{G0},\beta_{X0},\beta_{C0}^{\prime })^{\prime }$ regardless of $\gamma $ ; 
or (ii) $\gamma \neq
\gamma _{0}(m)$ but $(\beta_{G},\beta_{X},\beta_C^{\prime })^{\prime}=(\beta_{G0},\beta_{X0},\beta_{C0}^{\prime })^{\prime }$. These two cases are comprehensive.

For case (i), let $I_{i\gamma }(m)=\{g_{i}\leq \min \{\gamma ,\gamma _{0}\left( m\right) \}\}\cup \{g_{i}\geq \max \{\gamma ,\gamma _{0}\left( m\right) \}\}$ and note that 
\begin{eqnarray*}
&&R(\beta,\gamma |s) \\
&\geq &\mathbb{E}\left[ \left. 1[I_{i\gamma }(m)]\left\{ 
\begin{array}{l}
x_{i}^{\prime }\left( \beta_{C0}-\beta_C \right) +\left( \beta_{G0}-\beta_{G}\right) \left( g_{i}-\gamma _{0}\left( m_{i}\right) \right) _{-} \\ 
+\left( \beta _{X0}-\beta_{X}\right) \left( g_{i}-\gamma _{0}\left( m_{i}\right) \right) _{+} \\ 
+\beta_{G}\left( \left( g_{i}-\gamma  \right) _{-}-\left( g_{i}-\gamma _{0}\left( m\right) \right) _{-}\right)  \\ 
+\beta_{X}\left( \left( g_{i}-\gamma  \right) _{+}-\left( g_{i}-\gamma _{0}\left( m\right) \right) _{+}\right) 
\end{array}%
\right\} ^{2}\right\vert m_{i}=m\right]  \\
&=&\mathbb{E}\left[ 1[I_{i\gamma }(m)]\left. \left\{ 
\begin{array}{c}
x_{i}^{\prime }\left( \beta_{C0}-\beta_C \right) +\left( \beta_{G0}-\beta_{G}\right) \left( g_{i}-\gamma_{0}\left( m_{i}\right) \right) _{-} \\ 
+\left(\beta_{X0}-\beta_{X}\right) \left( g_{i}-\gamma _{0}\left(m_{i}\right) \right)_{+}
\end{array}%
\right\} ^{2}\right\vert m_{i}=m\right]  \\
&>& 0
\end{eqnarray*}
when $(\beta_{G},\beta_{X},\beta_C^{\prime })^{\prime }\neq (\beta_{G0},\beta_{X0},\beta_{C0}^{\prime })^{\prime }$ from the Assumption \ref{ass:id}.2.

For case (ii), we consider $\gamma  <\gamma _{0}(m)$ without loss of generality. Note that 
\begin{eqnarray*}
&&R(\beta_{0},\gamma |m) \\
&=&\mathbb{E}\left[ \left. 
\begin{array}{l}
\{\beta _{G0}\left( \left( g_{i}-\gamma \right) _{-}-\left(g_{i}-\gamma _{0}\left( m\right) \right) _{-}\right)  \\ 
+ \beta _{X0}\left( \left( g_{i}-\gamma \right) _{+}-\left(g_{i}-\gamma _{0}\left( m\right) \right) _{+}\right) \}^{2}%
\end{array}%
\right\vert m_{i}=m\right]  \\
&=&\mathbb{E}\left[ \left. 
\begin{array}{l}
\{1\left[ g_{i}\leq \gamma \right] \beta_{G0}\left( \gamma _{0}\left( m\right) -\gamma \right)  \\ 
+1\left[ g_{i}>\gamma _{0}\left( m\right) \right] \beta _{X0}\left( \gamma _{0}\left( m\right) -\gamma \right)  \\ 
+1\left[ \gamma  <g_{i}\leq \gamma _{0}\left( m\right) \right] \left[ \beta_{G0}\left( \gamma _{0}\left( m\right) -g_{i}\right) +\beta_{X0}\left( g_{i}-\gamma \right) \right] \}^{2}%
\end{array}%
\right\vert m_{i}=m\right]  \\
&=&\mathbb{E}\left[ 1\left[ g_{i}\leq \gamma \left( m\right) \right] \beta_{G0}^{2}\left( \gamma _{0}\left( m\right) -\gamma \left( m\right) \right) ^{2}|m_{i}=m\right]  \\
&&+\mathbb{E}\left[ 1\left[ g_{i}>\gamma _{0}\left( m\right) \right] \beta_{X0}^{2}\left( \gamma _{0}\left( m\right) -\gamma \left( m\right) \right) ^{2}|m_{i}=m\right]  \\
&&+\mathbb{E}\left[ 1\left[ \gamma \left( m\right) <g_{i}\leq \gamma_{0}\left( m\right) \right] \left[ \beta _{G0}\left( \gamma_{0}\left( m\right) -g_{i}\right) +\beta_{X0}\left( g_{i}-\gamma \left( m\right) \right) \right] ^{2}|m_{i}=m\right]  \\
&>&0,
\end{eqnarray*}
where the last probability is strictly positive, because $\beta_{G0}\neq \beta_{X0}$ and $\Gamma \subset \mathcal{G}$ from Assumption \ref{ass:id}.3. 
\end{proof}

\begin{proof}[Proof of Theorem \ref{thm:consistency}]

Recall the notation $\theta=(\beta',\gamma(m))'$ for fixed $m$.
With slight abuse of notations, we write
\begin{equation*}
\hat{\theta}=\arg \min_{\theta \in \Theta }S_{n}\left( \theta \right) ,
\end{equation*}
where $\Theta$ denotes the parameter space of $\theta$,
\begin{equation*}
S_{n}\left( \theta \right) =\frac{1}{n}\sum_{i=1}^{n}\frac{1}{b_{n} f_m(m)}K\left( \frac{m_{i}-m}{b_{n}}\right) \left( \pi _{i}-\beta ^{\prime }z_{i}\left(\gamma \right) \right)^{2}(1+o_p(1)),
\end{equation*}
and 
\begin{equation*}
z_{i}\left( \gamma \right) =\left( 
\begin{array}{c}
\left( g_{i}-\gamma \right) _{-} \\ 
\left( g_{i}-\gamma \right) _{+} \\ 
x_{i}
\end{array}
\right).
\end{equation*}
The $o_p(1)$ term results from the approximation
\begin{equation*}
    \frac{1}{nb_n}\sum_{i=1}^n K\Big(\frac{m_i-m}{b_n}\Big) = f_m(m)+o_p(1),
\end{equation*}
which is from standard arguments \citep[e.g.,][Chapter 1]{li2007nonparametric}.

Define
\begin{equation*}
h_{i,n}\left( \theta \right) =\frac{1}{b_{n}f_m(m)}K\left( \frac{m_{i}-m} {b_{n}}\right) \left( \pi _{i}-\beta ^{\prime }z_{i}\left( \gamma \right) \right)^{2},
\end{equation*}
which is a row-wise i.i.d.~triangular array. We aim to show that
\begin{equation}
\sup_{\theta \in \Theta }\left\vert \frac{1}{n}\sum_{i=1}^{n}h_{i,n}\left(\theta \right) -\mathbb{E}\left[ h_{i,n}\left( \theta \right) \right]
\right\vert \overset{p}{\rightarrow }0.  \label{eq:ulln}
\end{equation}%
Then, the consistency follows from Theorem \ref{thm:id}, Lemma \ref{lem:bias}, and the standard argument \citep[e.g.,][Theorem 4.1.1]{amemiya1985advanced}. 
Since $h_{i,n}\left( \theta \right) $ is a quadratic function of $\beta $, the only issue is to establish the stochastic equicontinuity with respect to $\gamma$. 
We thus simplify the notation by writing $h_{i,n}\left( \theta \right) =h_{i,n}\left( \gamma \right) $ and illustrate with $\beta =\beta_{0}$.
Moreover, since we focus on a fixed $m$, we write $\gamma =\gamma \left(m\right) $ for simplicity. 
It suffices to check the following three steps. 
Then (\ref{eq:ulln}) follows by checking the conditions of Theorem 2.8.1 of \citet{van1996weak}:

1. $h_{i,n}\left( \gamma \right) $ is Lipschitz continuous with constant $\kappa _{i,n}=\kappa \left( \pi _{i},z_{i},m_{i},g_i\right) $. 

2. $\left\vert \left\vert \kappa _{i,n}^{2}\right\vert \right\vert $ and $h_{i,n}\left( \gamma \right) ^{2}$ are dominated by a function $H\left( \pi_{i},z_{i},m_{i},g_i\right) $. 

3. $\sup_{n}\mathbb{E}\left[ H\left( \pi_{i},z_{i},m_i,g_i\right) \right] <\infty $. 

For 1, without loss of generality, we suppress the linear term $x_i$. Then
\begin{eqnarray*}
&&\left\vert h_{i,n}\left( \gamma \right)-h_{i,n}\left( \gamma ^{\prime }\right)\right\vert  \\
&=&\frac{1}{b_{n}f_m(m)}K\left( \frac{m_{i}-m}{b_{n}}\right) \left\vert \left( \pi _{i}-\beta _{0}^{\prime }z_{i}\left( \gamma \right) \right) ^{2}-\left( \pi _{i}-\beta _{0}^{\prime }z_{i}\left( \gamma ^{\prime }\right) \right)^{2}\right\vert  \\
&\leq &\frac{1}{b_{n}f_m(m)}K\left( \frac{m_{i}-m}{b_{n}}\right) \left( 
\begin{array}{c}
\left\vert \beta_{G0}\left( g_{i}-\gamma \right)_{-}-\beta _{G0}\left(g_{i}-\gamma ^{\prime }\right) _{-}\right\vert  \\ 
+\left\vert \beta_{X0}\left( g_{i}-\gamma \right) _{+}-\beta_{X0}\left(g_{i}-\gamma ^{\prime }\right) _{+}\right\vert 
\end{array}%
\right)  \\
&&\times \left\vert 2\pi _{i}-\beta_{0}^{\prime }z_{i}\left( \gamma \right)-\beta_{0}^{\prime }z_{i}\left( \gamma ^{\prime }\right) \right\vert  \\
&\leq &\left\vert \gamma -\gamma ^{\prime }\right\vert \times \frac{1}{b_{n}f_m(m)}K\left( \frac{m_{i}-m}{b_{n}}\right) \left( \left\vert \beta_{G0}\right\vert +\left\vert \beta_{X0}\right\vert \right) \left\vert 2\pi_{i}-\beta_{0}^{\prime }z_{i}\left( \gamma \right) -\beta_{0}^{\prime}z_{i}\left( \gamma ^{\prime }\right) \right\vert  \\
&\equiv &\left\vert \gamma -\gamma ^{\prime }\right\vert \times \kappa_{i,n}.
\end{eqnarray*}

For 2, set 
\begin{equation*}
H\left( \pi _{i},z_{i},m_{i},g_i\right) =\frac{4}{b_{n}^{2}f_m(m)^2}K^{2}\left( \frac{m_{i}-m}{b_{n}}\right)\sup_{\beta_{G},\beta_{X}}\left( \left\vert\beta_{G}\right\vert +\left\vert \beta_{X}\right\vert \right)^{2}\sup_{\beta ,\gamma }\left\vert \pi _{i}-\beta ^{\prime }z_{i}\left(\gamma \right) \right\vert ^{2}.
\end{equation*}%
Then 
\begin{equation*}
\kappa _{i,n}^{2}\leq H\left( \pi_{i},z_{i},m_{i},g_i\right) \text{ and }h_{i,n}\left( \gamma \right) ^{2}\leq H\left( \pi_{i},z_{i},m_{i},g_i\right) .
\end{equation*}

For 3, by Assumptions \ref{ass:a}.1 and \ref{ass:a}.2,
\begin{eqnarray*}
\mathbb{E}\left[ H\left( \pi _{i},z_{i},m_{i},g_i\right) \right]  
&\leq &C\sup_{t\in \mathbb{R}^{+}}t^{2}K\left( t\right) \mathbb{E}\left[ \sup_{\beta ,\gamma}\left\vert \pi _{i}-\beta ^{\prime }z_{i}\left( \gamma \right) \right\vert^{2}\right]  \\
&\leq &C\sup_{t\in \mathbb{R}^{+}}t^2 K\left( t\right) \mathbb{E}\left[ \pi _{i}^{2}+2\sup_{\beta_{G},\beta _{X}}\left( \left\vert \beta _{G}\right\vert +\left\vert \beta_{X}\right\vert \right) ^{2}\sup_{\gamma\in\Gamma }\left( g_{i}-\gamma \right) ^{2}\right]  \\
&<&\infty ,
\end{eqnarray*}%
which holds uniformly over $n$. 

\end{proof}
\begin{lemma}
\label{lem:bias}Under Assumptions \ref{ass:id} and \ref{ass:a}, 
\begin{equation*}
\mathbb{E}\left[ S_{n}\left( \theta \right) \right] -L\left( \theta \right)=O\left( b_{n}^{2}\right) .
\end{equation*}
\end{lemma}

\begin{proof}[Proof of Lemma \protect\ref{lem:bias}]
By Taylor expansion and Assumptions \ref{ass:a}.1 and \ref{ass:a}.4, 
\begin{eqnarray*}
\mathbb{E}\left[ S_{n}\left( \theta \right) \right]  
&=&\frac{1}{b_{n}}%
\mathbb{E}\left[ \frac{1}{f_{m}\left( m\right) }K\left( \frac{m_{i}-m}{b_{n}}\right) \left( \pi _{i}-\beta ^{\prime }z_{i}\left(\gamma \right) \right) ^{2}\right]  \\
&=&\int \mathbb{E}\left[ \left( \pi _{i}-\beta ^{\prime }z_{i}\left( \gamma\right) \right) ^{2}|m_{i}=m+b_{n}t\right] K\left( t\right) \frac{f_{m}\left( m +b_{n}t\right) }{f_{m}\left( m\right) }dt\\
&=&\mathbb{E}\left[ \left( \pi _{i}-\beta ^{\prime }z_{i}\left( \gamma\right) \right) ^{2}|m_{i}=m\right] +O\left(b_{n}^{2}\right) ,
\end{eqnarray*}%
where the $O\left( b_{n}^{2}\right) $ is uniform over $\theta$. 
\end{proof}

\begin{proof}[Proof of Theorem \ref{thm:normal}]
Since this theorem applies pointwise at $m$, we suppress $m$ to simplify notations. 
Denote $e_{i}\left( \theta \right) =\pi _{i}-\beta ^{\prime }z_{i}\left(\gamma \right) $. 
The estimator $\hat{\theta} = \hat{\theta}(m)$ approximately solves the FOC
\begin{eqnarray*}
0 &=&\sum_{i=1}^{n}k_{i,n}D_{i}\left( \theta \right) e_{i}\left( \theta\right) \\
&=&\sum_{i=1}^{n}k_{i,n}D_{i}\left( \theta \right) \left( u_{i}+\beta_{0}^{\prime }z_{i}\left( \gamma _{0}\left( m_{i}\right) \right) -\beta^{\prime }z_{i}\left( \gamma \right) \right) \\
&=&\sum_{i=1}^{n}k_{i,n}D_{i}\left( \theta \right)u_{i}+\sum_{i=1}^{n}k_{i,n}D_{i}\left( \theta \right) \left( \beta_{0}-\beta \right) ^{\prime }z_{i}\left( \gamma \right) \\
&&+\sum_{i=1}^{n}k_{i,n}D_{i}\left( \theta \right) \beta _{0}^{\prime}\left( z_{i}\left( \gamma _{0}\left( m\right) \right) -z_{i}\left( \gamma\right) \right) \\
&&+\sum_{i=1}^{n}k_{i,n}D_{i}\left( \theta \right) \beta _{0}^{\prime}\left( z_{i}\left( \gamma _{0}\left( m_{i}\right) \right) -z_{i}\left(\gamma _{0}\left( m\right) \right) \right) \\
&\equiv &A_{1n}+A_{2n}+A_{3n}+A_{4n}.
\end{eqnarray*}%
We study the limits of $A_{1n},A_{2n},A_{3n},A_{4n}$ one-by-one. 
By standard argument, 
\begin{equation*}
\frac{1}{nb_{n}}\sum_{i=1}^{n}K\left( \frac{m_{i}-m}{b_{n}}\right) = f_{m}\left( m\right) +O_{p}\left( b_{n}^{2}+\frac{1}{\sqrt{nb_{n}}}\right) ,
\end{equation*}
provided that $f_{m}(m)>0$ and it is twice continuously differentiable with uniformly bounded derivatives. 
We therefore write 
\begin{equation*}
k_{i,n}=\frac{1}{nb_{n}}\sum_{i=1}^{n}K\left( \frac{m_{i}-m}{b_{n}}\right) \frac{1}{f_{m}\left( m\right)}(1+o_p(1)).
\end{equation*}

For $A_{1n}$, we have that 
\begin{eqnarray*}
&&\sqrt{nb_{n}}A_{1n} \\
&=&\frac{1}{\sqrt{nb_{n}}f_{m}\left( m\right) }\sum_{i=1}^{n}K\left( \frac{m_{i}-m}{b_{n}}\right) D_{i}\left( \theta_{0}\right) u_{i} (1+o_p(1))\\
&&+\frac{1}{\sqrt{nb_{n}}f_{m}\left( m\right) }\sum_{i=1}^{n}K\left( \frac{m_{i}-m}{b_{n}}\right) \left( D_{i}\left( \theta \right) -D_{i}\left( \theta_{0}\right) \right) u_{i}+o_{p}\left( 1\right) \\
&\equiv &\sqrt{nb_{n}}A_{11n}+\sqrt{nb_{n}}A_{12n}+o_{p}\left( 1\right),
\end{eqnarray*}
where we evaluate $\theta = \hat{\theta}$.  
The first item is asymptotically normal by the standard CLT with asymptotic variance
\begin{equation*}
V(m)=\frac{1}{f_{m}\left( m\right) }\mathbb{E}\left[ D_{i}\left( \theta_{0}\right) D_{i}\left( \theta _{0}\right) ^{\prime }u_{i}^{2}|m_{i}=m\right]\int K\left( u\right) ^{2}du.
\end{equation*}
We now show that the second term is $o_{p}\left( 1\right) $ using Cauchy-Schwarz inequality. Specifically, 
\begin{eqnarray*}
&&nb_{n}\mathbb{E}\left[ A_{12n}^{2}\right] \\
&=&\frac{1}{b_{n}f_{m}\left( m\right) ^{2}}\mathbb{E}\left[ K\left( \frac{m_{i}-m}{b_{n}}\right) ^{2}\left\vert \left\vert D_{i}\left( \theta \right)-D_{i}\left( \theta _{0}\right) \right\vert \right\vert ^{2}u_{i}^{2}\right] \\
&=&\frac{1}{f_{m}\left( m\right) ^{2}}\int \mathbb{E}\left[ \left\vert\left\vert D_{i}\left( \theta \right) -D_{i}\left( \theta _{0}\right) \right\vert \right\vert ^{2}u_{i}^{2}|m_{i}=m+b_{n}u\right] f_{m}\left(m+b_{n}u\right) K\left( u\right) du \\
&\leq &\frac{\left\vert \left\vert \beta -\beta _{0}\right\vert \right\vert }{f_{m}\left( m\right) }\left( \mathbb{E}\left[ u_{i}^{2}|m_{i}=m\right]+o(1)\right) \\
&&+\frac{\left\vert \gamma -\gamma _{0}\left( m\right) \right\vert \times \left\vert \left\vert \bar{\beta}\right\vert \right\vert ^{2}}{f_{m}\left(m\right) }\left( \mathbb{E}\left[ u_{i}^{2}|m_{i}=m,g_{i}=\gamma _{0}\left(m\right) \right] f_{g|m}\left( \gamma _{0}\left( m\right) |m\right) +o\left(1\right) \right) ,
\end{eqnarray*}
where the last inequality follows from Assumption \ref{ass:a}.4 that $f_{g|m}\left(r|m\right) $ and $\mathbb{E}\left[ u_{i}^{2}|m_{i}=m,g_{i}=r\right] $ are both continuous differentiable w.r.t.~$m$ and $r$ with uniformly bounded derivatives, and $||\bar\beta||=\sup_{\beta}||\beta||<\infty$. Therefore, $\sqrt{nb_{n}}A_{12n}=o_{p}(1)$ follows from that $\hat{\theta}-\theta _{0}=o_{p}(1)$. Combining $A_{11n}$ and $A_{12n}$ yields that 
\begin{equation*}
\sqrt{nb_{n}}A_{1n} \overset{d}{\rightarrow} \mathcal{N}\left( 0,V\left( m\right) \right) .
\end{equation*}

We now study $A_{2n}$. By the uniform convergence in \eqref{eq:ulln} and consistency $\hat{\theta}-\theta_{0}=o_{p}(1)$, we have that
\begin{equation*}
\frac{1}{nb_{n}}\sum_{i=1}^{n}K\left( \frac{m_{i}-m}{b_{n}}\right) \binom{z_{i}\left( \gamma \right) z_{i}\left( \gamma \right) ^{\prime }}{d_{i}\left( \theta \right) z_{i}\left( \gamma \right) ^{\prime }}=\mathbb{E}\left[ \binom{z_{i}\left( \gamma _{0}\right) z_{i}\left( \gamma _{0}\right)^{\prime }}{d_{i}\left( \theta _{0}\right) z_{i}\left( \gamma _{0}\right)^{\prime }}\Big\vert m_{i}=m\right] f_{m}\left( m\right) +o_{p}\left( 1\right),
\end{equation*}
where $d_i(\theta) = -\beta_G 1[g_i<\gamma(m)]-\beta_X 1[g_i>\gamma(m)]$.
It follows that
\begin{eqnarray*}
\sqrt{nb_{n}}A_{2n} &=&\sum_{i=1}^{n}k_{i,n}D_{i}\left( \theta \right) \left( \beta _{0}-\hat\beta \right) ^{\prime }z_{i}\left( \gamma \right) \\
&=&\frac{1}{\sqrt{nb_{n}}f_{m}\left( m\right) }\sum_{i=1}^{n}K\left( \frac{m_{i}-m}{b_{n}}\right) \binom{z_{i}\left( \gamma \right) z_{i}\left( \gamma \right) ^{\prime }}{d_{i}\left( \theta \right) z_{i}\left( \gamma \right)^{\prime }}\left( \beta_{0}-\hat\beta \right) +o_{p}\left( 1\right) \\
&=&\mathbb{E}\left[ \binom{z_{i}\left( \gamma _{0}\right) z_{i}\left( \gamma_{0}\right) ^{\prime }}{d_{i}\left( \theta _{0}\right) z_{i}\left( \gamma_{0}\right) ^{\prime }} \Big\vert m_{i}=m\right] \sqrt{nb_{n}}\left( \beta_{0}-\hat\beta\right) +o_{p}\left( 1\right) .
\end{eqnarray*}

For $A_{3n}$, without loss of generality, we assume $\gamma _{0}\left(m\right) <\gamma $. Some linear algebra yields that
\begin{eqnarray*}
&&\left( 
\begin{array}{c}
\left( g_{i}-\gamma _{0}\left( m\right) \right) _{-} \\ 
\left( g_{i}-\gamma _{0}\left( m\right) \right) _{+} \\ 
x_{i}
\end{array}
-
\begin{array}{c}
\left( g_{i}-\gamma \right) _{-} \\ 
\left( g_{i}-\gamma \right) _{+} \\ 
x_{i}
\end{array}
\right) ^{\prime }\beta _{0} \\
&=&\left\{ 
\begin{array}{ll}
\beta _{X0}\left( \gamma -\gamma _{0}\left( m\right) \right) & \text{if }g_{i}>\gamma \\ 
\left( \beta _{X0}-\beta _{G0}\right) \left( g_{i}-\gamma _{0}\left(m\right) \right) +\beta _{G0}\left( \gamma -\gamma _{0}\left( m\right) \right) & \text{if }\gamma _{0}\left( m\right) \leq g_{i}\leq \gamma \\ 
\beta _{G0}\left( \gamma -\gamma _{0}\left( m\right) \right) & \text{if }g_{i}<\gamma _{0}\left( m\right) .
\end{array}
\right. ,
\end{eqnarray*}
By the uniform convergence in \eqref{eq:ulln} and consistency $\hat{\theta}-\theta_{0}=o_{p}(1)$ again, we have that 
\begin{eqnarray*}
A_{3n} &=&\frac{1}{b_{n}f_{m}\left( m\right) }\mathbb{E}\left[ K\left( \frac{m_{i}-m}{b_{n}}\right) D_{i}\left( \theta \right) \beta _{0}^{\prime }\left(z_{i}\left( \gamma _{0}\left( m\right) \right) -z_{i}\left( \gamma \right)\right) \right] +o_{p}\left( 1\right) \\
&=&\mathbb{E}\left[ \left. \binom{z_{i}\left( \gamma \right) }{d_{i}\left(\theta \right) }\left( 
\begin{array}{c}
\left( g_{i}-\gamma _{0}\left( m\right) \right) _{-} \\ 
\left( g_{i}-\gamma _{0}\left( m\right) \right) _{+} \\ 
x_{i}
\end{array}
-
\begin{array}{c}
\left( g_{i}-\gamma \right) _{-} \\ 
\left( g_{i}-\gamma \right) _{+} \\ 
x_{i}
\end{array}
\right) ^{\prime }\beta _{0}\right\vert m_{i}=m\right] +o_{p}\left( 1\right)
\\
&=&\mathbb{E}\left[ \left. \binom{z_{i}\left( \gamma \right) }{d_{i}\left(\theta \right) }\left\{ 
\begin{array}{l}
\beta _{X0}\left( \gamma -\gamma _{0}\left( m\right) \right) 1[g_{i}>\gamma ]\\ 
+\left( \beta _{X0}-\beta _{G0}\right) \left( g_{i}-\gamma _{0}\left(m\right) \right) 1\left[ \gamma _{0}\left( m\right) \leq g_{i}\leq \gamma \right] \\ 
-\beta _{X0}\left( \gamma -\gamma _{0}\left( m\right) \right) 1\left[ \gamma_{0}\left( m\right) \leq g_{i}\leq \gamma \right] \\ 
+\beta _{G0}\left( \gamma -\gamma _{0}\left( m\right) \right) 1[g_{i}<\gamma_{0}\left( m\right) ]
\end{array}
\right\} \right\vert m_{i}=m\right] +o_{p}\left( 1\right) \\
&=&\left( \gamma -\gamma _{0}\left( m\right) \right) \mathbb{E}\left[ \left. \binom{z_{i}\left( \gamma _{0}\right) }{d_{i}\left( \theta _{0}\right) }\left\{ \beta _{G0}1(g_{i}<\gamma _{0})+\beta _{X0}1(g_{i}>\gamma )\right\}\right\vert m_{i}=m\right] +o_{p}\left( 1\right) .
\end{eqnarray*}
By combining $A_{2n}$ and $A_{3n}$ and stacking $\hat{\theta}=( \hat{\beta}^{\prime },\hat{\gamma}) ^{\prime }$, we then obtain that%
\begin{equation*}
\sqrt{nb_{n}}\left( A_{2n}+A_{3n}\right) = - Q\left( m\right) \sqrt{nb_{n}}\left( \hat{\theta}-\theta_{0}\right) +o_{p}(1),
\end{equation*}
where we recall
\begin{equation*}
Q\left( m\right) =\mathbb{E}\left[ D_{i}\left( \theta _{0}\right) D_{i}\left( \theta _{0}\right) ^{\prime }|m_{i}=m\right],
\end{equation*}
with
\begin{eqnarray*}
D_{i}\left( \theta \right)  &=&-\frac{\partial }{\partial \theta }\left( \pi_{i}-\beta ^{\prime }z_{i}\left( \gamma \right) \right)  \\
&=&\binom{z_{i}\left( \gamma \right) }{-\beta_{G}1\left[g_{i}<\gamma \right] -\beta_{X}1\left[ g_{i}>\gamma \right] }.
\end{eqnarray*}

It remains to study $A_{4n}$, which introduces the asymptotic bias. Note that 
\begin{eqnarray}
&&\mathbb{E}\left[ A_{4n}\right]  \notag \\
&=&\frac{1}{\left( f_{m}\left( m\right) +o_{a.s.}(1)\right) b_{n}}\mathbb{E}\left[K\left( \frac{m_{i}-m}{b_{n}}\right) D_{i}\left( \theta \right) \left(
z_{i}\left( \gamma _{0}\left( m_{i}\right) \right) -z_{i}\left( \gamma_{0}\left( m\right) \right) \right) ^{\prime }\right] \beta _{0}  \notag \\
&=&\int \mathbb{E}\left[ \left. D_{i}\left( \theta \right) \left(z_{i}\left( \gamma _{0}\left( m+ub_{n}\right) \right) -z_{i}\left( \gamma
_{0}\left( m\right) \right) \right) ^{\prime }\right\vert m_{i}=m\right] \frac{f_{m}\left( m+ub_{n}\right) }{f_{m}\left( m\right) +o_{a.s.}(1)}K\left(u\right) du,  \label{eq:EA4n}
\end{eqnarray}
where the $o_{a.s.}(1)$ term follows from the almost sure convergence of the kernel estimator \citep[e.g.,][]{hansen2008uniform}. 
Repeat the calculation in $A_{3n}$ with $\gamma $ replaced by $\gamma_{0}\left( m_{i}\right) $. 
Without loss of generality, consider $\gamma _{0}\left( m+ub_{n}\right) >\gamma _{0}\left( m\right)$. 
We have that 
\begin{eqnarray*}
&&\int\mathbb{E}\left[ \left. D_{i}\left( \theta \right) \left(z_{i}\left( \gamma _{0}\left( m+ub_{n}\right) \right) -z_{i}\left( \gamma
_{0}\left( m\right) \right) \right) ^{\prime }\right\vert m_{i}=m\right] \frac{f_{m}\left( m+ub_{n}\right) }{f_{m}\left( m\right) +o_{a.s.}(1)}K\left(
u\right) du \\
&=&\int\mathbb{E}\left[ \left. D_{i}\left( \theta \right) \left\{ 
\begin{array}{c}
B_{1}(g_{i},ub_{n})+B_{2}(g_{i},ub_{n}) \\ 
+B_{3}(g_{i},ub_{n})+B_{4}(g_{i},ub_{n}) 
\end{array}%
\right\} \right\vert m_{i}=m\right] \left( 1+f_{m}^{(1)}\left( m\right) ub_{n}+O\left( b_{n}^{2}\right) \right) K\left( u\right) du,
\end{eqnarray*}
where
\begin{equation*}
\begin{array}{l}
B_{1}(g_{i},ub_{n})=\beta _{X0}\left( \gamma _{0}\left( m+ub_{n}\right)-\gamma _{0}\left( m\right) \right) 1[g_{i}>\gamma _{0}\left(m+ub_{n}\right) ] \\ 
B_{2}(g_{i},ub_{n})=\left( \beta _{X0}-\beta _{G0}\right) \left(g_{i}-\gamma _{0}\left( m\right) \right) 1\left[ \gamma _{0}\left( m\right)\leq g_{i}\leq \gamma _{0}\left( m+ub_{n}\right) \right] \\ 
B_{3}(g_{i},ub_{n})=\beta _{G0}\left( \gamma _{0}\left( m+ub_{n}\right)-\gamma _{0}\left( m\right) \right) 1\left[ \gamma _{0}\left( m\right) \leq g_{i}\leq \gamma _{0}\left( m+ub_{n}\right) \right] \\ 
B_{4}(g_{i},ub_{n})=\beta _{G0}\left( \gamma _{0}\left( m+ub_{n}\right) -\gamma _{0}\left( m\right) \right) 1[g_{i}<\gamma _{0}\left( m\right)].
\end{array}
\end{equation*}
For $B_{1}(g_{i},ub_{n})$ and $B_{4}(g_{i},ub_{n})$, substituting $\gamma_{0}\left( m+ub_{n}\right) -\gamma _{0}\left( m\right) =\gamma_{0}^{(1)}(m) u b_{n}+\frac{1}{2}\gamma_{0}^{(2)}(m) u^{2}b_{n}^2+o(b_{n}^{2}) $ (which is implied by Assumption 2.5) yields that 
\begin{eqnarray*}
&&\int \mathbb{E}\left[ \left. D_{i}\left( \theta \right) \left\{B_{1}(g_{i},ub_{n})+B_{4}(g_{i},ub_{n})\right\} \right\vert m_{i}=m\right]\left( 1+f_{m}^{(1)}\left( m\right) ub_{n}+O\left( b_{n}^{2}\right) \right) K\left( u\right) du \\
&=&b_{n}^{2}\mathbb{E}\left[ \left. D_{i}\left( \theta _{0}\right) d_{i}\left( \theta _{0}\right) \right\vert m_{i}=m\right] \int u^{2}K
\left( u\right) du\left( \gamma_{0}^{(1)}\left( m\right) f_{m}^{(1)}\left( m\right) +\frac{1}{2}\gamma_{0}^{(2)}\left( m\right) \right) +o(b_{n}^{2}).
\end{eqnarray*}
For $B_{2}\left( g_{i},ub_{n}\right) $, we have that 
\begin{eqnarray*}
&&\int\mathbb{E}\left[ \left. D_{i}\left( \theta \right) B_{2}(g_{i},ub_{n})\right\vert m_{i}=m\right] \left( 1+f_{m}^{(1)}\left(
m\right) ub_{n}+O\left( b_{n}^{2}\right) \right) K\left( u\right) du \\
&=&\left( \beta _{X0}-\beta _{G0}\right) \int \left( \int_{0}^{\gamma_{0}\left( m+ub_{n}\right) -\gamma _{0}\left( m\right) }s\mathbb{E}\left[\left. D_{i}\left( \theta \right)\right\vert m_{i}=m\right] f_{g|m}\left(s+\gamma _{0}\left( m\right) |m\right) ds\right) K\left( u\right) du\left(1+O(b_{n})\right) \\
&=&b_{n}^{2}\frac{\beta_{X0}-\beta_{G0}}{2}\gamma_{0}^{(1)}(m)^{2}\mathbb{E}\left[ \left. D_{i}\left( \theta _{0}\right)
\right\vert g_{i}=\gamma_{0}\left( m\right),m_{i}=m\right] f_{g|m}\left(\gamma _{0}\left( m\right) |m\right) \int u^{2}K\left( u\right)
du+o(b_{n}^{2}).
\end{eqnarray*}
For $B_{3}\left( g_{i},ub_{n}\right) $, we have that
\begin{eqnarray*}
&&\int\mathbb{E}\left[ \left. D_{i}\left( \theta \right) B_{3}(g_{i},ub_{n})\right\vert m_{i}=m\right] \left( 1+f_{m}^{(1)}\left(
m\right) ub_{n}+O\left( b_{n}^{2}\right) \right) K\left( u\right) du \\
&=&\beta _{G0}\int \left( \gamma _{0}\left( m+ub_{n}\right) -\gamma_{0}\left( m\right) \right) \left( \int_{\gamma _{0}\left( m\right)
}^{\gamma _{0}\left( m+ub_{n}\right) }\mathbb{E}\left[ \left. D_{i}\left(\theta \right) \right\vert g_{i}=g,m_{i}=m\right] f_{g|m}\left( g|m\right) dg\right) K\left( u\right) du\left( 1+O(b_{n})\right) \\
&=&b_{n}^{2}\beta _{G0}\gamma_{0}^{(1)}(m)^{2}\mathbb{E}\left[ \left. D_{i}\left( \theta _{0}\right) \right\vert g_{i}=\gamma
_{0}\left( m\right) ,m_{i}=m\right] f_{g|m}\left( \gamma _{0}\left( m\right)|m\right) \int u^{2}K\left( u\right) du+o(b_{n}^{2}).
\end{eqnarray*}
The argument for $\gamma _{0}\left( m+ub_{n}\right) \leq \gamma _{0}\left( m\right)$ is the same. 
Combining $B_{1}\left( g_{i},ub_{n}\right) $ to $B_{4}\left(g_{i},ub_{n}\right) $ yields that 
\begin{equation*}
\mathbb{E}\left[ A_{4n}\right] =b_{n}^{2}\times B\left( m\right),
\end{equation*}
where 
\begin{eqnarray*}
B\left( m\right) &=&\int u^{2}K\left( u\right) du\times \left\{ \mathbb{E}\left[ \left. D_{i}\left( \theta _{0}\right) d_{i}\left( \theta _{0}\right) \right\vert m_{i}=m\right] \left( \gamma _{0}^{(1)}\left( m\right)
f_{m}^{(1)}\left( m\right) +\frac{1}{2}\gamma _{0}^{(2)}\left( m\right) \right) \right. \\
&&\left. +\frac{\beta_{X0}+\beta_{G0}}{2}\mathbb{E}\left[ \left. D_{i}\left( \theta _{0}\right) \right\vert g_{i}=\gamma _{0}\left( m\right)
,m_{i}=m\right] f_{g|m}\left( \gamma _{0}\left( m\right) |m\right) \gamma_{0}^{(1)}(m)^{2}\right\} .
\end{eqnarray*}
By standard argument and the consistency, we can show that $Var\left[ A_{4n}\right] =O\left( \frac{b_{n}^{2}}{nb_{n}}\right) $, which is of smaller order than $O(\frac{1}{nb_{n}})$. Therefore,
\begin{equation*}
\sqrt{nb_{n}}A_{4n}=\sqrt{nb_{n}}b_{n}^{2}\times B\left( m\right) \left(1+o_{p}(1)\right) .
\end{equation*}%
Combining $A_{1n}$ to $A_{4n}$ yields that%
\begin{eqnarray*}
&&\sqrt{nb_{n}}\left\{ Q\left( m\right) \left( \hat{\theta}-\theta_{0}\right) -b_{n}^{2}\times B\left( m\right) \right\} \\
&=&-\sqrt{nb_{n}}A_{1n}+o_{p}(1) \\
&\overset{d}{\rightarrow} &\mathcal{N}\left( 0,V\left( m\right) \right),
\end{eqnarray*}%
as desired.
\end{proof}

\begin{proof}[Proof of Theorem \ref{thm:rootn}]
We introduce the short-handed notation: 
$\gamma _{i}^{\text{max}}=\max\{\gamma_{0}(m_{i}), \hat{\gamma}_{-i}(m_{i}) \}$,
$\gamma _{i}^{\text{min}}=\min \{\gamma _{0}\left( m_{i}\right),
\hat{\gamma}_{-i}\left( m_{i}\right) \}$, 
$\Delta \gamma _{i}=\hat{\gamma}_{-i}\left( m_{i}\right) -\gamma _{0}\left( m_{i}\right) $,
and $1_i\mathcal{M} = 1[m_i\in\mathcal{M}]$. 
Also denote
\begin{equation*}
z_{i}\left( \gamma _{0}\left( m_{i}\right) \right) -z_{i}\left( \hat{\gamma}_{-i}\left( m_{i}\right) \right) =z_{i}\left( \Delta \gamma_{i}\right) .
\end{equation*}
Then, substitute the true model to obtain that
\begin{eqnarray}
\sqrt{n}\left( \hat{\beta}^{\ast }-\beta _{0}\right) 
&=&\left( \frac{1}{n}\sum_{i=1}^{n}z_{i}\left( \hat{\gamma}_{-i}\left(m_{i}\right) \right) z_{i}\left( \hat{\gamma}_{-i}\left( m_{i}\right)\right)^{\prime }1_{i}\mathcal{M}\right) ^{-1}  \notag \\
&&\times \left( 
\begin{array}{l}
\frac{1}{\sqrt{n}}\sum_{i=1}^{n}z_{i}\left( \hat{\gamma}_{-i}\left(m_{i}\right) \right) u_{i}1_{i}\mathcal{M}+ \\ 
\frac{1}{\sqrt{n}}\sum_{i=1}^{n}z_{i}\left( \hat{\gamma}_{-i}\left(m_{i}\right) \right) \beta _{0}^{\prime }z_{i}\left( \Delta \gamma
_{i}\right) 1_{i}\mathcal{M}
\end{array}
\right)  \label{eq:bstar} \\
&\equiv &A_{n}^{-1}\times \left( B_{1n}+B_{2n}\right) .  \notag
\end{eqnarray}

We now establish
\begin{eqnarray}
A_{n} &=&\mathbb{E}\left[ z_{i}\left( \gamma _{0}\left( m_{i}\right) \right)z_{i}\left( \gamma _{0}\left( m_{i}\right) \right) ^{\prime }1_{i}\mathcal{M}\right] +o_{p}(1)  \label{eq:bstar_A} \\
B_{1n} &=&\frac{1}{\sqrt{n}}\sum_{i=1}^{n}z_{i}\left( \gamma _{0}\left(m_{i}\right) \right) u_{i}1_{i}\mathcal{M}+o_{p}(1) \label{eq:bstar_B1} \\
B_{2n} &=&\frac{1}{\sqrt{n}}\sum_{i=1}^{n}\delta \left( m_{i}\right) u_{i}1_{i}\mathcal{M}+o_{p}(1),  \label{eq:bstar_B2}
\end{eqnarray}
and the proof is completed by combining these terms.

We first study $A_{n}$. 
By the proof of Theorem \ref{thm:normal}, we have that
\begin{equation}
\hat{\gamma}_{-i}\left( m\right) -\gamma \left( m\right) =Q\left( m\right)^{-1}\left( \sum_{i=1}^{n}k_{i,n}D_{i}\left( \theta_{0}(m) \right)
u_{i}+b_{n}^{2}B\left( m\right) \right) \left( 1+o_{p}(1)\right) \label{eq:gam_first_order} .
\end{equation}
By standard argument (e.g., Li and Racine, 2007, Chapter 2) and Assumption \ref{ass:a}.4, the bias term $Q\left( m\right) ^{-1}B\left( m\right) $ is uniformly bounded. 
Also, the stochastic term satisfies 
\begin{equation*}
\sup_{m\in \mathrm{int}(\mathcal{M})}\left\vert \sum_{i=1}^{n}k_{i,n}D_{i}\left(\theta_{0}(m) \right) u_{i}\right\vert =O_{p}\left( \frac{\ln n}{\sqrt{nb_{n}}}\right) ,
\end{equation*}
yielding that
\begin{equation}
\sup_{m\in \mathrm{int}(\mathcal{M})}\left\vert \hat{\gamma}_{-i}\left( m\right)-\gamma \left( m\right) \right\vert =O_{p}\left( \frac{\ln n}{\sqrt{nb_{n}}} +b_{n}^{2}\right) .  \label{eq:gam_uniform}
\end{equation}
Then, (\ref{eq:bstar_A}) follows from the continuous mapping theorem, law of large numbers, and \eqref{eq:gam_uniform}.

Now we study $B_{1n}$, which can be decomposed as
\begin{equation*}
B_{1n}=\frac{1}{\sqrt{n}}\sum_{i=1}^{n}z_{i}\left( \gamma _{0}\left(m_{i}\right) \right) u_{i}1_{i}\mathcal{M}-R_{n},
\end{equation*}
where
\begin{equation*}
R_{n}=\frac{1}{\sqrt{n}}\sum_{i=1}^{n}z_{i}\left( \Delta \gamma _{i}\right)u_{i}1_{i}\mathcal{M}.
\end{equation*}
It suffices to show that $R_{n}=o_{p}(1)$. To this end, the leave-one-out estimator $\hat{\gamma}_{-i}\left( m_{i}\right) $, the i.i.d.ness, and \eqref{eq:gam_uniform} imply that 
\begin{equation*}
\mathbb{E}\left[ \left\vert \left\vert R_{n}\right\vert \right\vert ^{2}\right] \leq \frac{1}{n}\sum_{i=1}^{n}\mathbb{E}\left[ \left( \Delta \gamma_{i}\right) ^{2}u_{i}^{2}1_{i}\mathcal{M}\right] =O\left( b_{n}^{4}+\frac{\left( \ln n\right) ^{2}}{nb_{n}}\right) .
\end{equation*}%
Then $R_{n}=o_{p}(1)$ follows from the Cauchy-Schwarz inequality, yielding \eqref{eq:bstar_B1}.

Finally we study $B_{2n}$. By simple algebra, we have that
\begin{eqnarray*}
\beta _{0}^{\prime }z_{i}\left( \Delta \gamma _{i}\right) 
&=&\beta _{X0}\Delta \gamma _{i}1[g_{i}>\gamma _{i}^{\text{max}}]+\beta_{G0}\Delta\gamma _{i}1[ g_{i}\leq \gamma _{i}^{\text{min}} ] \\
&&+\left( \beta _{X0}-\beta _{G0}\right) \left( g_{i}-\gamma _{i}^{\text{min}}\right) 1[ \gamma _{i}^{\text{min}}<g_{i}\leq \gamma_{i}^{\text{max}} ] \\
&&-\beta _{X0}\Delta \gamma _{i}1\left[ \gamma _{i}^{\text{min}}<g_{i}\leq \gamma _{i}^{\text{max}}\right] \\
&=&-\Delta \gamma _{i}d_{i}+r_{n,i},
\end{eqnarray*}
where we recall that
\begin{equation*}
d_{i} =-\beta _{X0}1 [ g_{i}>\gamma _{0}\left(m_{i}\right) ]-\beta _{G0}1 [ g_{i}\leq \gamma _{0}\left(m_{i}\right) ]
\end{equation*}
and $r_{n,i}$ satisfies
\begin{equation*}
\left\vert r_{n,i}\right\vert \leq 2\left( \left\vert \beta _{X0}\right\vert+\left\vert \beta _{G0}\right\vert \right) |\Delta \gamma _{i}| 1\left[ \gamma_{i}^{\text{min}}<g_{i}\leq \gamma _{i}^{\text{max}}\right] .
\end{equation*}

Substituting $r_{n,i}$ into \eqref{eq:bstar} yields 
\begin{eqnarray*}
B_{2n} &=&\frac{1}{\sqrt{n}}\sum_{i=1}^{n}z_{i}\left( \hat{\gamma}_{-i}\left( m_{i}\right) \right) \beta _{0}^{\prime }z_{i}\left( \Delta
\gamma _{i}\right) 1_{i}\mathcal{M} \\
&=&-\frac{1}{\sqrt{n}}\sum_{i=1}^{n}z_{i}\left( \gamma _{0}\left(m_{i}\right) \right) \Delta \gamma _{i}d_{i}1_{i}\mathcal{M}+\frac{1}{\sqrt{n%
}}\sum_{i=1}^{n}z_{i}\left( \gamma _{0}\left( m_{i}\right) \right) r_{n,i}1_{i}\mathcal{M} \\
&&-\frac{1}{\sqrt{n}}\sum_{i=1}^{n}z_{i}\left( \Delta \gamma _{i}\right)\Delta \gamma _{i}d_{i}1_{i}\mathcal{M}+\frac{1}{\sqrt{n}}\sum_{i=1}^{n}z_{i}\left( \Delta \gamma _{i}\right) r_{n,i}1_{i}\mathcal{M} \\
&\equiv &-B_{21n}+B_{22n}-B_{23n}+B_{24n}.
\end{eqnarray*}%
By simple algebra again, $B_{23n}$ and $B_{24n}$ both involve an additional $\Delta \gamma _{i}$ and hence are of smaller order than $B_{21n}$ and $B_{22n}$, respectively. We therefore focus on $B_{21n}$ and $B_{22n}$.

We first study $B_{22n}$ using \eqref{eq:gam_uniform}. The Cauchy-Schwarz inequality implies that 
\begin{eqnarray*}
&&\mathbb{E}\left[ \left\vert \left\vert B_{22n}\right\vert \right\vert \right] \\
&\leq &\sqrt{n}\mathbb{E}\left[ \left\vert \left\vert z_{i}\left( \gamma_{0}\left( m_{i}\right) \right) \right\vert \right\vert |r_{n,i}|1_{i}\mathcal{M}\right] \\
&\leq &\sqrt{n}C\left( \mathbb{E}\left[ \left\vert \left\vert z_{i}\left(\gamma _{0}\left( m_{i}\right) \right) \right\vert \right\vert ^{4}\right]\right)^{1/4} \mathbb{E}\left[ \left\vert \Delta \gamma_{i}\right\vert ^{4}1_{i} \mathcal{M}\right]^{1/4} \left(\mathbb{E}\left[ 1\left[ \gamma _{i}^{\text{min}}<g_{i}\leq \gamma _{i}^{\text{max}}\right] \right] \right) ^{1/2} \\
&\leq &\sqrt{n}C\left( \left( \ln n\right) \left( nb_{n}\right)^{-1/2}+b_{n}^{2}\right) ^{3/2}\rightarrow 0,
\end{eqnarray*}
which is guaranteed by $b_{n}=O\left( n^{\lambda}\right) $ with $\lambda\in (-1/4,-1/3)$ and $f_{g|m}\left( g|m\right) $ is uniformly bounded over $\left( g,m\right) $. Therefore, $B_{22n}=o_{p}(1)$.
Here, $C$ denotes a generic finite constant, whose value may vary across lines.

It remains to derive $B_{21n}$. To this end, we write that
\begin{eqnarray*}
&&\frac{1}{\sqrt{n}}\sum_{i=1}^{n}z_{i}\left( \gamma _{0}\left(m_{i}\right) \right) d_{i}\Delta \gamma _{i}1_{i}\mathcal{M} \\
&=&\frac{1}{\sqrt{n}}\sum_{i=1}^{n}z_{i}\left( \gamma _{0}\left(m_{i}\right) \right) d_{i}\left( \hat{\gamma}_{-i}\left( m_{i}\right) -\mathbb{E}\left[ \hat{\gamma}_{-i}\left( m\right) \right] _{m=m_{i}}\right)1_{i}\mathcal{M} \\
&&+\frac{1}{\sqrt{n}}\sum_{i=1}^{n}z_{i}\left( \gamma _{0}\left(m_{i}\right) \right) d_{i}\left( \mathbb{E}\left[ \hat{\gamma}_{-i}\left(
m\right) \right] _{m=m_{i}}-\gamma _{0}\left( m_{i}\right) \right) 1_{i}\mathcal{M} \\
&\equiv &T_{11n}+T_{12n}.
\end{eqnarray*}%
For $T_{11n}$, plug in the first-order approximation of $\hat{\gamma}_{-i}\left( m_{i}\right) $ as in (\ref{eq:gam_first_order}) to obtain that
\begin{equation*}
\hat{\gamma}_{-i}\left( m_{i}\right) -\mathbb{E}\left[ \hat{\gamma}_{-i}\left( m\right) \right]_{m=m_{i}} =\iota'Q\left(m_{i}\right)^{-1}\sum_{j\neq i}^{n}w_{ij}D_{j}u_{j},
\end{equation*}%
where $\iota =(0,\dots,0,1)'$, $D_j$ is short for $D_j(\theta_0(m_j))$, and 
\begin{equation}
w_{ij}=\frac{K\left( \frac{m_{j}-m_{i}}{b_{n}}\right) }{\sum_{j\neq i}^{n}K\left( \frac{m_{j}-m_{i}}{b_{n}}\right) }=\frac{K\left( \frac{m_{j}-m_{i}}{b_{n}}\right) }{nb_{n} f_{m} ( m_{i} )  }(1+o_{a.s.}(1)) \equiv \frac{\tilde{w}_{ij}}{nb_n}(1+o_{a.s.}(1)). \label{eq:wij}
\end{equation}
The $o_{a.s}(1) $ is uniform over $m_{i}\in \mathrm{int}(\mathcal{M})$. 
Denote $\Delta _{i}\left( m_{i}\right) =z_{i}\left( \gamma_{0}\left( m_{i}\right) \right) d_{i}\iota ^{\prime }Q\left( m_{i}\right)^{-1}$. 
Then, 
\begin{eqnarray}
T_{11n}&=&\frac{1}{\sqrt{n}}\sum_{i=1}^{n}\Delta _{i}\left( m_{i}\right) \sum_{j\neq i}^{n}w_{ij}D_{j}u_{j}\cdot 1_{i}\mathcal{M}  \notag \\
&=&\frac{1}{\sqrt{n}}\sum_{j=1}^{n}\left( \sum_{i\neq j}w_{ij}\Delta_{i}\left( m_{i}\right)1_{i}\mathcal{M}  \right) D_{j} u_{j}  \notag \\
&=&\frac{1}{\sqrt{n}}\sum_{j=1}^{n}\mathbb{E}\left[ \phi _{j}|m_{j}\right]D_{j}u_{j}+\frac{1}{\sqrt{n}}\sum_{j=1}^{n}\left( \phi _{j}-\mathbb{E}\left[\phi _{j}|m_{j}\right] \right) D_{j}u_{j},  \label{eq:T11n}
\end{eqnarray}
where
\begin{equation*}
\phi_{j}=\sum_{i\neq j}w_{ij}\Delta _{i}\left( m_{i}\right) \cdot 1_{i}\mathcal{M}.
\end{equation*}%
For the first term in (\ref{eq:T11n}), by the change of variable $m_i-m_j = ub_n$, we obtain 
\begin{eqnarray*}
\mathbb{E}\left[ \phi _{j}|m_{j}\right]  &=&\frac{1}{b_{n}}\mathbb{E}\left[\left. \frac{K\left( \frac{m_{j}-m_{i}}{b_{n}}\right) }{f_{m}\left(
m_{j}\right) }\Delta _{i}\left( m_{i}\right) \cdot 1_{i}\mathcal{M}\right\vert m_{j}\right]  \\
&=&\int K\left( u\right) \frac{f_{m}\left(m_i \right) \mathbb{E}\left[ \Delta _{i}\left( m\right) |m_{i}=m\right] |_{m=m_{j}+ub_{n}}}{f_{m}\left( m_{j}+ub_n \right) }1\left[ m_{j}+ub_{n}\in \mathrm{int}\left( \mathcal{M}\right) \right] du \\
&=&\mathbb{E}\left[ \Delta _{i}\left( m_{i}\right) |m_{i}=m\right]_{m=m_{j}}1_{j}\mathcal{M}+O\left( b_{n}^{2}\right) ,
\end{eqnarray*}
where the $O\left( b_{n}^{2}\right) $ is uniform over $m_{j}\in \mathrm{int}(\mathcal{M})$. 

For the second term in (\ref{eq:T11n}), by the i.i.d.ness and the fact that $\mathbb{E[}u_{j}D_{j}^{\prime }(\phi _{j}-\mathbb{E}\left[ \phi _{j}|m_{j}\right])|m_j ]=0$, we have that
\begin{eqnarray*}
&&\mathbb{E}\left[ \left\vert \left\vert \frac{1}{\sqrt{n}}\sum_{j=1}^{n}u_{j}D_{j}^{\prime }\left( \phi _{j}-\mathbb{E}\left[ \phi_{j}|m_{j}\right] \right) \right\vert \right\vert ^{2}\right]  \\
&=&\mathbb{E}\left[ \left\vert \left\vert \frac{1}{\sqrt{n}}\sum_{j=1}^{n}u_{j}D_{j}^{\prime }\left( \sum_{i\neq j}\tilde{w}_{ij}\Delta
_{i}\left( m_{i}\right) \cdot 1_{i}\mathcal{M}-\mathbb{E}\left[ \phi_{j}|m_{j}\right] \right) \right\vert \right\vert^{2}\right]  \\
&=&\frac{1}{n^{3}b_{n}^{2}}\mathbb{E}\left[ 
\begin{array}{c}
\sum_{j}\sum_{j'}\sum_{i\neq j}\sum_{i'\neq j'}\left( \tilde{w}_{ij}\Delta _{i}\left( m_{i}\right) 1_{i}\mathcal{M}-\mathbb{E}
\left[ \phi _{j}|m_{j}\right] \right)  \\ 
\times \left( \tilde{w}_{i^{\prime }j'}\Delta _{i^{\prime }}\left( m_{i^{\prime}}\right) 1_{i^{\prime }}\mathcal{M}-\mathbb{E}\left[ \phi _{j'}|m_{j'}\right]\right) u_{j^{\prime }}D_{j^{\prime }}^{\prime }D_{j}u_{j}
\end{array}
\right]  \\
&=&\frac{1}{n^{3}b_{n}^{2}}\sum_{j}\mathbb{E}\left[ 
\begin{array}{c}
\sum_{i\neq j}\sum_{i^{\prime }\neq j}\left( \tilde{w}_{ij}\Delta_{i}\left(m_{i}\right)1_{i}\mathcal{M}-\mathbb{E}\left[\phi_{j}|m_{j}\right]\right)\\ 
\times \left( \tilde{w}_{i^{\prime }j}\Delta _{i^{\prime }}\left( m_{i^{\prime}}\right) 1_{i^{\prime }}\mathcal{M}-\mathbb{E}\left[ \phi _{j}|m_{j}\right]\right) u_{j}^{2}||D_{j}||^2
\end{array}
\right]  \\
&\leq &\frac{1}{n^{2}b_{n}^{2}}\mathbb{E}\left[ \sum_{i\neq j}\left\vert \tilde{w}_{ij}\Delta_{i}\left(m_{i}\right)1_{i}\mathcal{M}-\mathbb{E}\left[\phi_{j}|m_{j}\right] \right\vert^2 u_{j}^{2}||D_{j}||^2\right]  \\
& = & O\left(\frac{1}{nb_n}\right) \rightarrow 0,
\end{eqnarray*}
where $\tilde{w}_{ij}$ is defined in \eqref{eq:wij}.
Therefore, 
\begin{equation*}
T_{11n}=\frac{1}{\sqrt{n}}\sum_{j=1}^{n}\delta \left( m_{j}\right) u_{j}1_{j}\mathcal{M} + o_p(1),
\end{equation*}
where
\begin{eqnarray*}
\delta \left( m_{j}\right)  &=&\mathbb{E}\left[ \Delta _{i}\left( m_{i}\right) |m_{i}=m\right] |_{m=m_{j}}D_{j} \\
&=&\left. \mathbb{E}\left[ z_{i}\left( \gamma _{0}\left(m_{i}\right)\right)d_{i}\iota ^{\prime }Q\left( m_{i}\right)^{-1}|m_{i}=m\right]\right\vert_{m=m_{j}}D_{j}.
\end{eqnarray*}

For $T_{12n}$, it follows from (\ref{eq:gam_uniform}) that%
\begin{eqnarray*}
\left\vert \left\vert T_{12n}\right\vert \right\vert &\leq &\sqrt{n}b_{n}^{2}\cdot C\frac{1}{n}\sum_{i=1}^{n}\left\vert \left\vert z_{i}\left(
\gamma _{0}\left( m_{i}\right) \right) d_{i}\right\vert \right\vert 1_{i}\mathcal{M} \\
&=&o_{p}\left( 1\right),
\end{eqnarray*}%
given that $nb_{n}^{4}\rightarrow 0$. 
Then, combining $T_{11n}$ and $T_{12n}$ yields the leading term of $B_{21n}$ as desired.
\end{proof}

\end{document}